\documentclass[aip,jmp,preprint,a4paper]{revtex4-1}

\setcitestyle{numbers,square}

\usepackage{cancel}
\usepackage{amsmath}
\usepackage{amssymb}
\usepackage{amsthm}
\usepackage{graphicx}% Include figure files
\usepackage{dcolumn}% Align table columns on decimal point
\usepackage{bm}% bold math
\usepackage{upgreek}
\usepackage{esint}
\usepackage{mathtools}

\theoremstyle{definition}
\newtheorem{definition}{Definition}

\theoremstyle{remark}
\newtheorem{remark}{Remark}

\theoremstyle{theorem}
\newtheorem{theorem}{Theorem}

\theoremstyle{lemma}
\newtheorem{lemma}{Lemma}

\theoremstyle{corollary}
\newtheorem{corollary}{Corollary}

\theoremstyle{proposition}
\newtheorem{proposition}{Proposition}

\theoremstyle{remark}
\newtheorem{example}{Example}

\usepackage{hyperref}
\hypersetup{
  colorlinks   	= true, 	%Colours links instead of ugly boxes
  urlcolor     	= blue, 	%Colour for external hyperlinks
  linkcolor    	= blue, 	%Colour of internal links
  citecolor   	= blue 	%Colour of citations
}
\newcommand{\calD}{\mathcal{D}}
\newcommand{\bB}{\bm{B}}
\newcommand{\bD}{\bm{D}}
\newcommand{\bE}{\bm{E}}
\newcommand{\bF}{\bm{F}}
\newcommand{\bS}{\bm{S}}
\newcommand{\bu}{\bm{u}}
\newcommand{\bU}{\bm{U}}
\newcommand{\bv}{\bm{v}}
\newcommand{\bw}{\bm{w}}
\newcommand{\bz}{\bm{z}}
\newcommand{\bZ}{\bm{Z}}
\newcommand{\bbP}{\mathbb{P}}
\newcommand{\ds}{\displaystyle}
\draft % marks overfull lines with a black rule on the right

\begin{document}

\tableofcontents

%\preprint{}

\title{Slow manifold reduction for plasma science} %Title of paper



\author{J. W. Burby}
\affiliation{Los Alamos National Laboratory, Los Alamos, New Mexico 87545, USA}
\author{T. J. Klotz}
\affiliation{University of Colorado Boulder, Boulder, Colorado 80309, USA}
 %\affiliation{New York University, New York, New York 10012, USA}
%\author{A. Cerfon}
%\affiliation{Courant Institute of Mathematical Sciences, New York, New York 10012, USA}
%\thanks{}


\date{\today}

\begin{abstract}
The classical Chapman-Enskog procedure admits a substantial geometrical generalization known as slow manifold reduction. This generalization provides a paradigm for deriving and understanding most reduced models in plasma physics that are based on controlled approximations applied to problems with multiple timescales. In this Review we develop the theory of slow manifold reduction with a plasma physics audience in mind. In particular we illustrate (a) how the slow manifold concept may be used to understand \emph{breakdown} of a reduced model over sufficiently-long time intervals, and (b) how a discrete-time analogue of slow manifold theory provides a useful framework for developing implicit integrators for temporally-stiff plasma models. For readers with more advanced mathematical training we also use slow manifold reduction to explain the phenomenon of inheritance of Hamiltonian structure in dissipation-free reduced plasma models. Various facets of the theory are illustrated in the context of the Abraham-Lorentz model of a single charged particle experiencing its own radiation drag. As a culminating example we derive the slow manifold underlying kinetic quasineutral plasma dynamics up to first-order in perturbation theory. This first-order result incorporates several physical effects associated with small deviations from exact charge neutrality that lead to slow drift away from predictions based on the leading-order approximation $n_e = Z_i \,n_i$.
%New stuff: slow manifolds in degenerate fast-slow systems, slow manifold integrators definition, SMIs have formal slow manifolds, error estimates for SMIs, quasineutral slow manifold and first-order effects caused by imperfect charge neutrality, new proof of zero-derivative principle, connections to plasma modeling, aimed at plasma physics audience, (all Hamiltonian systems arise as slow manifold reductions)
\end{abstract}


\maketitle %\maketitle must follow title, authors, abstract and \pacs

\section{Introduction}
While plasma is often referred to as the fourth state of matter, naturally occurring plasmas frequently behave like fluids. Plasma physicists learn that this counterintuitive behavior arises whenever the shortest timescale characterizing a plasma's dynamics is dominated by collisions. From the physical point of view the explanation for this fact is simple: sufficiently rapid collisions should cause a plasma's distribution function to relax to local thermal equilibrium, which may be characterized completely by fluid moments. From the mathematical perspective the famous Chapman-Enskog procedure not only formalizes the local-thermal-equilibrium picture, but provides a systematic method for deriving the hydrodynamic equations that govern a plasma in the collisional regime.
%
%The famous Chapman-Enskog procedure is often presented as the mathematical to a systematic derivation of the ensuing hydrodynamical plasma model.means
%
%Moreover most plasma physicists understand that the famous Chapman-Enskog procedure provides the mathematical means to a systematic derivation of the ensuing hydrodynamical plasma model.

While the Chapman-Enskog procedure applied to kinetic theory is familiar to many plasma physicists, an appreciation for the power of the method's underlying principles is less widespread. In fact the Chapman-Enskog procedure generalizes significantly to explain a wide variety of plasma model reductions besides the famous kinetic-to-fluid limit. These reductions include, but are by no means limited to, the passage from two-fluid theory to magnetohydrodynamics (MHD); magnetohydrodynamics to reduced MHD; and Vlasov-Maxwell dynamics to Vlasov-Poisson, Vlasov-Darwin, kinetic MHD, and gyrokinetics.

A far-reaching generalization of the classical Chapman-Enskog method, along with its role in plasma theory, serves as the central topic for this Review. In focusing on this model reduction tool, the Review aims to familiarize a greater fraction of the plasma physics community with its application, limitations, interpretation, and extensions. We contend that a broader and more sophisticated appreciation for this method within plasma physics would serve to increase the quality and pace of a large fraction of both theoretical and computational plasma physics research, particularly in the areas of simulation algorithm design and analysis.

Any presentation of this generalized Chapman-Enskog method must choose between one of two presentation styles, one analytic  and one geometric. Each of these choices carries with it corresponding historical baggage and betrays a bias on part of the presenter toward a particular view of the subject. The present 
Review's bias leans in favor of the geometric approach, but not out of any disrespect for the analytic side of the subject. Instead we pursue the geometric angle on utilitarian grounds; we believe the geometric viewpoint will be easier to digest, remember, and apply by plasma physicists.

The analytic variant of the theory aligns most closely with Chapman's and Enskog's original work,\,\citep{Chapman_Cowling_1970} and garners favor among mathematicians with proclivities toward real analysis. Such mathematicians often refer to the generalization of the Chapman-Enskog procedure as the method of Hilbert expansions --- terminology that explicitly acknowledges the connection between kinetic-to-fluid reduction and Hilbert's famous $6^{\text{th}}$ problem, ``Mathematical Treatment of the Axioms of Physics." \citep{Gorban_2018} Within plasma physics, a subset of researchers who are intimately familiar with the original Chapman-Enskog method either wittingly or unwittingly have applied the method of Hilbert expansions to various problems in plasma theory. (See for instance the derivations of kinetic MHD in \citep{Kulsrud_book_1983,Kulsrud_1962,Grad_1956,CGL_1956}; slab gyrokinetics and many of its limiting forms in \citep{Schekochihin_2009}; the fluid-kinetic hybrid modeling in \citep{Wang_1992}; and the radio-frequency wave-fluid modeling in \citep{Hegna_2009}.) On the mathematical side, some good recent examples of the Hilbert expansion method applied to plasma-relevant problems appear in the works of Degond, Filbet, Bostan, Frénod, Sonnendrucker, Golse, and Saint-Raymond. \citep{Degond_2010,Degond_2017,Degond_2016,Filbet_2010,Bostan_2010}

The geometric variant of the generalized Chapman-Enskog method originates in the works of Tikhonov \citep{Tikhonov_1952}  and Fenichel \citep{Fenichel_1979}. It frequently appears under the moniker of ``slow manifold reduction" in connection with such weighty terms as ``inertial manifold,'' \citep{Temam_1990} ``slow invariant manifold,'' or ``almost invariant set.'' \citep{Kristiansen_2016} While the geometric perspective on generalized Chapman-Enskog has gone mostly unnoticed in the plasma physics literature until quite recently, (see \citep{Burby_two_fluid_2017,Burby_Sengupta_2017_pop,Burby_loops_2019,Burby_Ruiz_2019}) prominent examples of slow manifold reduction do appear in other parts of physics. See the lengthy review by Gorban, Karlin, and Zinovyev \citep{Gorban_2004} for a development of the slow manifold picture of kinetic-to-fluid reduction, and the sequence of papers by Lorenz \citep{Lorenz_1986,Lorenz_1987,Lorenz_1992} that popularized slow manifold reduction as a means for understanding the concept of quasi-geostrophic balance from geophysical fluid dynamics.

The geometric generalization of Chapman-Enskog theory embodied by slow manifold reduction leads to a vividly rich understanding of the role played by reduced models in the physics of plasmas. It replaces the purely analytic notion of \textbf{closure} with the more visual, geometric notion of an \textbf{invariant set} in a system's phase space. Slow manifold reduction theory therefore fits naturally within the differential-topological approach to plasma theory championed by Kaufman \citep{Tracy_2009,Kaufman_1987,Kaufman_pra_1987} and his Berkeley School of the 1980's. \citep{Tracy_1993,Littlejohn_1981} However, where the Berkeley school relied heavily on methods suitable for non-dissipative systems, slow manifold reduction works just as nicely in the dissipative setting as it does in the presence of Hamiltonian structure. (That said, there are special interactions between slow manifold theory and Hamiltonian systems theory that disappear in the presence of dissipation. C.f. Section \ref{hamiltonian_SM}.)

When dissipation plays a dominant dynamical role slow manifold theory provides both an asymptotic closure and an explanation for the closure's establishment; the invariant manifold embodying the closure attracts nearby trajectories at a rate set by the strength of dissipation. Referring once again to the kinetic-to-fluid reduction, the invariant manifold corresponds approximately to local Maxwellian distribution functions, while the attraction phenomenon physically corresponds to the Second Law of Thermodynamics. More generally slow manifolds arise in the presence of competition between dissipative and non-dissipative dynamical processes, similar in spirit to the competition between collisions and cyclotron dynamics at the heart of the well-known Braginskii closure.\,\citep{Braginskii_1965} In these cases the theory instead reveals that closure establishes itself by way of damped oscillations.

When dissipation is either weak or irrelevant, slow manifold theory exposes a closure's structural properties as well as its inherent delicacy. For instance, because slow manifold theory identifies a closure with a submanifold in phase space the task of demonstrating that the closure inherits Hamiltonian structure from its parent model becomes strikingly simple; the closure's Hamiltonian structure arises by merely ``pulling-back" the Hamiltonian structure of the parent model to the slow manifold. A much more subtle question about the reduced model therefore jumps to the fore --- in the absence of dissipation how long will the reduced model remain an accurate picture of the true system dynamics? Rapid oscillations about the slow manifold, even if initialized with small amplitudes, generally pass in and out of resonance with one another, leading to a possibly-diffusive breakdown of the reduction. \citep{Neishtadt_1996} Such resonant interaction brings to mind the ``kicks'' experienced by a charged particle in a magnetic bottle subject to radio-frequency waves upon crossing the resonant surfaces where the cyclotron frequency divides the wave frequency.\,\citep{Jaeger_1972} If oscillations about the slow manifold avoid resonance, adiabatic invariance serves as a mechanism by which a system can stick to a slow manifold over much larger time intervals. \citep{Nekhoroshev_1971}

In order to ease readers into the theory of slow manifolds we have written this Review using a pedagogical style that assumes only a modest level of mathematical sophistication. Linear algebra is a must, as is multivariable calculus, especially the chain rule for smooth functions between vector spaces of arbitrary dimension. Differential topology is only required for a careful reading of the advanced topics in Section \ref{hamiltonian_SM}, where we assume familiarity with differential forms on manifolds.  In order to illustrate various facets of the theory we repeatedly refer to a small set of finite-dimensional examples that facilitate manual calculations. We hope that this approach will rapidly build concrete intuition for otherwise abstract concepts that pervade the theory of fast-slow dynamical systems.

While much of the material we will cover may also be found elsewhere in the literature, especially in the Review articles \citep{MacKay_2004, Gorban_2004}, a modest amount of material appears here for the first time. Our proof of the zero derivative principle in Section \ref{zdp_sec} appears to be new, as does our discussion of discrete-time fast-slow systems, slow manifold integrators, and their relationship with implicit numerical integrators in Section \ref{SMI_sec}. We also extend MacKay's treatment \citep{MacKay_2004} of slow manifolds in symplectic and Poisson Hamiltonian systems to include presymplectic Hamiltonian systems with possibly nontrivial Lie symmetry groups in Section \ref{hamiltonian_SM}. Finally many of our examples, especially the lengthy discussion of the slow manifold contained in the collisional Vlasov-Maxwell system with quasineutral scaling in Section \ref{QN_application}, contain new observations.

\section{Training Wheels: Abraham-Lorentz Dynamics\label{AL_section}}
Due in part to the potentially disastrous consequences of runaway electrons that may be produced in the long-awaited ITER magnetic fusion experiment, an old classical physics problem recently reemerged as an important component of contemporary plasma physics deliberation --- the problem of computing the drag on an electron produced by the electron's own emitted radiation. This problem deviled physicists prior to the quantum revolution due to the vexing infinities associated with a point-charge's self field. During the development of quantum mechanics, however, the problem took a back seat to seemingly more grandiose concerns. Nowadays while quantum field theory sheds light on some aspects of radiation drag, open questions remain. Especially relevant to the plasma community is the question of how to properly and efficiently account for radiation drag in many-body systems; one must sometimes contend with the interference of the individual radiation fields produced by each particle. \citep{Howe_2014,Kimel_1995} (The role of this coherent radiation phenomenon in runaway electron physics is unclear.)

Early on in the $20^{\text{th}}$ century Abraham proposed the following system of equations to model radiation reaction in the case of a single electron with (negative) charge $e$ moving through a static external magnetic field $\bm{B}(\bm{x})$. (The restriction to static magnetic fields is merely convenient for our discussion; Abraham's results allow for time-dependent external electric and magnetic fields.)
\begin{align}
 \frac{2}{3}\frac{e^2}{c^3}\dot{\bm{a}} &=m\bm{a}- \frac{e}{c} \bm{v}\times\bm{B}(\bm{x})\label{AL_one} \\
\dot{\bm{v}}& = \bm{a}\\
\dot{\bm{x}}& = \bm{v}.\label{AL_three}
\end{align}
This form of the radiation drag force, $\bm{F}_R = \frac{2}{3}\frac{e^2}{c^3}\dot{\bm{a}} $, emerges from a renormalization argument applied to the dynamics of a finite-sized non-relativistic charged particle with vanishingly-small radius. The problematic infinities produced by point-like charges still appear as the particle size tends to zero, but Abraham tamed them by tuning the ``bare" particle mass. The parameter $m$ that appears in Eq.\,\eqref{AL_one} may be interpreted as the (finite) sum of the (infinite) electromagnetic mass and the (infinite) bare particle mass. Mass renormalization ``cancels" the infinity associated with the point-particle self-energy.

While Abraham's proposal only applies to a single electron, and it falls short of supplying an accurate evolution equation in the relativistic regime appropriate for runaway electrons, we will use the Abraham-Lorentz model repeatedly in this Review to illustrate many features of slow manifold reduction theory. There are several good reasons for doing so. (a) From the perspective of multi-scale dynamical systems theory the Abraham-Lorentz equation is qualitatively similar to the more-accurate and manifestly Lorentz-covariant Lorentz-Abraham-Dirac equation. (b) Manual computations associated with the Abraham-Lorentz model tend to be algebraically simpler than corresponding calculations involving the Lorentz-Abraham-Dirac model. (c) Different parameter regimes for the Abraham-Lorentz equation serve admirably to illustrate qualitatively distinct features of slow manifold reduction theory.
%
% its multi-scale nature together with its balance between analytic simplicity and qualitative similarity to the manifestly-covariant Lorentz-Abraham-Dirac equation render Eqs.\,\eqref{AL_one}-\eqref{AL_three} most suitable for illustrating many of the basic features inherent to slow manifold reduction theory.

When using the Abraham-Lorentz model to illustrate elements of slow manifold reduction theory, we will refer to different parameter regimes in order to highlight different aspects of the theory. These different regimes are described most efficiently through the introduction of rescaled dependent and independent variables for Eqs.\,\eqref{AL_one}-\eqref{AL_three} that explicitly reveal a pair of fundamental dimensionless parameters. Scale time by the observer timescale $T$ as $t = T\overline{t}$, space by the observer length scale $L$ as $\bm{x} = L\overline{\bm{x}}$, velocity as $\bm{v} = (L/T)\overline{\bm{v}}$, magnetic field as $\bm{B} = B_0 \overline{\bm{B}}$, and acceleration as $\bm{a} = (L/T) (|e| B_0) (m c)^{-1} \overline{\bm{a}}$. (This scaling for the acceleration corresponds to a particle executing cyclotron motion.) In terms of the dimensionless barred-variables the Abraham-Lorentz equations may be written
\begin{align}
\left(\frac{r_0}{cT}\right)\frac{2}{3} \frac{d\overline{\bm{a}}}{d\overline{t}} &=\overline{\bm{a}} - \zeta\, \overline{\bm{v}}\times\overline{\bm{B}}(\overline{\bm{x}}) \\
 \frac{d\overline{\bm{v}}}{d\overline{t}} &= (|\omega_c| T)\overline{\bm{a}}\\
\frac{d \overline{\bm{x}}}{d\overline{t}} &= \overline{\bm{v}},
\end{align}
where we have introduced the cyclotron frequency $\omega_c = e B_0 /(m c)$, the classical electron radius $r_0 = e^2/(mc^2)$, and the sign of the charge $\zeta = \pm 1$. ($\zeta = -1$ corresponds to electrons.) The dimensionless parameters
\begin{align}
\epsilon_R &= \frac{r_0}{cT}\\
\epsilon_B & =\frac{1}{|\omega_c| T}
\end{align}
represent the ratio of the electron size to the distance light travels during a time interval $T$, and the ratio of the cyclotron period to $T$, respectively. 

Assuming the background magnetic field exceeds about $10^4$ Gauss (in particular that $|\bm{B}|$ is nowhere vanishing), and the observation timescale we care about significantly exceeds the cyclotron period,  $\epsilon_R$ and $\epsilon_B$ tend to be quite small individually. The ratio of $\epsilon_R$ to $\epsilon_B$, 
\begin{align}
\frac{\epsilon_R}{\epsilon_B} = \frac{|\omega_c|r_0}{c} %= \frac{e B_0}{mc^2}\frac{e^2}{mc^2},
\end{align}
also tends to be small, although numerically it approaches unity if the strength of the magnetic field $B_0$ is comparable to that of a magnetar. (The assumptions underlying the derivation of the Abraham-Lorentz force surely break down in such extreme conditions.) Therefore physically-interesting parameter regimes in the magnetized setting correspond to the relative ordering
\begin{align}
\epsilon_R &= \epsilon^\gamma\\
\epsilon_B &= \epsilon,
\end{align}
where $\gamma\geq 0$ is a constant and  $1\gg \epsilon > 0$. 

In this Review we will be concerned with three values of the parameter $\gamma$: $\gamma = \infty$, $\gamma = 2$, and $\gamma = 0$. These $\gamma$ values correspond to the following three physically-distinct regimes of Abraham-Lorentz dynamics.

\subsection{Zero-Drag Regime: $\gamma = \infty$\label{zero_drag_sec}}

When $\gamma = \infty$ the radiation drag vanishes altogether. This zero-drag regime therefore corresponds to the usual Lorentz force equation
\begin{align}
\dot{\bm{v}}& = \frac{1}{\epsilon}\,\zeta\,\bm{v}\times \bm{B}\label{zero_drag_one}\\
\dot{\bm{x}}& = \bm{v}.\label{zero_drag_two}
\end{align}
Because these equations arise from the variational principle
\begin{align}
\delta \int_{t_1}^{t_2} \bigg( \bm{v}\cdot \dot{\bm{x}} + \frac{1}{\epsilon}\zeta \bm{A}(\bm{x})\cdot\dot{\bm{x}} - \frac{1}{2} |\bm{v}|^2\bigg)\,dt = 0,
\end{align}
where $\nabla\times\bm{A} = \bm{B}$, this form of Abraham-Lorentz dynamics will be useful when discussing slow manifolds that arise in non-dissipative systems in Section \ref{interpretation_SM}. We will also refer to the zero-drag regime when discussing the inheritance of Hamiltonian structure by slow manifolds in Section \ref{hamiltonian_SM}.

\subsection{Maximal Drag Regime: $\gamma = 0$\label{maximal_drag_sec}}

At the other extreme where $\gamma = 0$ the Abraham-Lorentz equation becomes
\begin{align}
\frac{2}{3}\dot{\bm{a}} &= \bm{a} -\zeta\,\bm{v}\times\bm{B}(\bm{x})\label{maximally_damped_one}\\
\epsilon \dot{\bm{v}} & = \bm{a}\\
\dot{\bm{x}} & = \bm{v}.\label{maximally_damped_three}
\end{align}
We will refer to this value of $\gamma$ as the maximal drag regime because $\gamma < 0$ corresponds to the physically uninteresting scenario of an observation timescale much shorter than the time for light to traverse a classical electron radius. The especially pronounced influence of radiation drag in the maximal drag regime illuminates a peculiar corner of slow manifold theory that often goes unnoticed in the literature. We will revisit this point, along with the maximal drag limit, when discussing degenerate fast-slow systems in Section \ref{when_split}.

\subsection{Weak Drag Regime: $\gamma = 2$\label{weak_drag_sec}}
The case $\gamma = 2$ corresponds to a generic middle ground between $\gamma = \infty$ and $\gamma = 0$ that we will refer to as the weak drag regime. In this regime the Abraham-Lorentz equations take the form
\begin{align}
\epsilon^2\frac{2}{3}\dot{\bm{a}} &= \bm{a} -\zeta\,\bm{v}\times\bm{B}(\bm{x})\label{weakly_damped_one}\\
\epsilon \dot{\bm{v}} & = \bm{a}\\
\dot{\bm{x}} & = \bm{v}.\label{weakly_damped_three}
\end{align}
Of the three regimes we will discuss in this Review, the weak drag regime makes closest contact with the seminal work of Spohn in \citep{Spohn_2000}.  We will use this regime to illustrate features of non-degenerate slow manifolds in dissipative systems in Section \ref{interpretation_SM}.

\section{Key Concept: Invariant Manifolds\label{IM_section}}
Invariant manifolds play a major role in modern dynamical systems theory, and also comprise the primordial concept underlying the notion of a slow manifold. Before embarking on a proper discussion of slow manifolds, we are therefore compelled  to define and discuss invariant manifolds in a broader context. This Section will introduce the notion of an invariant manifold associated with a dynamical system of the form 
\begin{align}\label{gen_ode}
\dot{z} = U(z),
\end{align}
where $z\in\mathbb{R}^d\equiv Z$ is a $d$-dimensional vector and $U:\mathbb{R}^d\rightarrow\mathbb{R}^d$ is a $d$-component vector field on $z$-space. For simplicity's sake we will not allow $U$ to depend on time, and we will not allow $d\rightarrow\infty$. This will keep the discussion within the realm of autonomous ordinary differential equations on finite-dimensional state spaces. (Note however that the example of quasineutral plasma dynamics discussed in Section \ref{QN_application} requires $d=\infty$.)

Broadly speaking, invariant manifolds may be thought of as ``effortless constraint sets," where ``effortless" refers to the fact that no external forcing is required to maintain the constraint. Effortless constraints stand in stark contrast to the usual holonomic constraints encountered in introductory mechanics texts such as \citep{Landau_1976}. For example, a pendulum may be realized as a particle in a gravitational field subject to the constraint that its distance to some fixed point remains constant in time. This constraint is not effortless however because persistence of the constant-length condition requires work to be done by some external agent. In contrast, restriction to a level set of the energy in a conservative system comprises a constraint that requires zero external effort to maintain; all that is required of an external agent to maintain the constant-energy constraint is to prepare the system in question to lie on a particular energy surface at a single instant of time.

From the mathematical perspective an invariant manifold is defined as a subset $S\subset Z$ of $z$-space with the following properties.
\begin{definition}[invariant manifold]\label{invariant_manifold_def}
A subset $S\subset Z$ is an invariant manifold associated with the ODE \eqref{gen_ode} if the following conditions are satisfied by $S$.
\begin{itemize}
\item[(a)] $S$ is a smooth submanifold of $Z$.
\item[(b)] $S$ is invariant under the flow of the ODE \eqref{gen_ode}.
\end{itemize}
\end{definition} 
\noindent Because each of the conditions (a) and (b) is somewhat technical it is worth dwelling on their meaning  here.  

Property (a) means that there is some smooth invertible coordinate transformation $z\mapsto (x,y)$ with $x\in\mathbb{R}^k$ and $y\in\mathbb{R}^{d-k}$ such that each point $s\in S$ has the form
\begin{align}
s = (x,y^*(x)),
\end{align}
with $y^*(x)\in\mathbb{R}^{d-k}$ a smooth vector-valued function of $x$. In other words $S$ is the \emph{graph} of a smooth function $y^*$. The integer $0\leq k \leq d $ is called the dimension of $S$. Strictly speaking, submanifolds in general are only required to be graphs locally in order to accommodate folds and points where $y^*(x)$ cannot be defined. However, except for the following example, all submanifolds that appear in this Review will arise as graphs of well-defined single-valued functions $y^*$.
%(In actual fact this explanation of (a) is inadequate because some submanifolds may fold over themselves, leading to regions where $y^*(x)$ is either not defined or multi-valued. However, all submanifolds in this Review do not exhibit such behavior -- they are globally graphs.) 
\begin{example}
The upper half of the circle in $\mathbb{R}^2$ defined by the equation $x^2+y^2 = 1$, $y>0$, is a $1$-dimensional submanifold of $\mathbb{R}^2$ because it may be written as the graph of $y^*(x) = \sqrt{1-x^2}$ with $x$ restricted to the interval $[-1,1]$. The lower half of the circle is also a submanifold with $y^*(x) = - \sqrt{1 - x^2}$ and $x$ restricted in the same way. In fact the whole circle $x^2+y^2 = 1$ is technically a submanifold that may be thought of as the graph of the multi-valued function $y^*_{\pm}(x) = \pm\sqrt{1-x^2}$ whose domain is restricted to the interval $x\in[-1,1]$. The left and right ends of $[-1,1]$ correspond to \emph{folds}, the presence of which always implies a multi-valued $y^*$. Although the square root is well-defined as a complex-valued function for $x$ outside the interval $[-1,1]$, the definition of a submanifold requires $y^*(x)$ to be real-valued. This explains why the domain of $x$ must be restricted. In contrast the graph of $y^*_1(x) = \sqrt{1+x^2}$ is a $1$-dimensional submanifold that requires neither a multi-valued $y^*$ nor a restriction on $x$.
\end{example}

Property (b) means that if a solution of Eq.\,\eqref{gen_ode} is contained in $S$ initially then it remains in $S$ for all time. In symbols we may write
\begin{align}
z(0)\in S\Rightarrow z(t)\in S\,\,\forall t\in\mathbb{R}.\label{invariance_condition_abstract}
\end{align}
As a useful shorthand we also say that $S$ is \emph{invariant}.

\begin{example}
Any point $z_0$ where $U(z_0) = 0$ is invariant under the flow of \eqref{gen_ode} because the unique solution of \eqref{gen_ode} with $z(0) = z_0$ is $z(t) = z_0$. Likewise, if \eqref{gen_ode} obeys the conservation law $\frac{d}{dt} Q(z(t)) = 0$ then any level set of the conserved quantity $Q$ is invariant. Note however that not all invariant objects arise as either fixed points or level sets of conserved quantities. Consider, for example, the ODE
\begin{align}
\dot{y} =& A(y)\label{linear_IM}\\
\dot{x} =& g(x,y),\label{nonlinear_IM}
\end{align}
where $A(y)$ is a smooth nonlinear function with $A(0) = 0$ and $g(x,y)$ is an arbitrary function of $x$ and $y$. In this case the set $y=0$ is invariant because $y=0$ is a fixed point for $\dot{y} = A(y)$ even though $(x,y) = (x,0)$ need not be a fixed point for the entire ODE for any $x$. In particular the invariance of the set $y=0$ does not require the existence of a conservation law for Eqs.\,\eqref{linear_IM}-\eqref{nonlinear_IM}.
\end{example}

Upon combining properties (a) and (b) we may construct the following analytic picture of invariant manifolds. By property (a) there are coordinates $(x,y)$ on $z$-space such that $S$ is the graph of a smooth function $y^*(x)$. In these coordinates the ODE \eqref{gen_ode} takes the general form
\begin{align}
\dot{y} &= f(x,y)\label{gen_ode_split_one}\\
\dot{x} &= g(x,y),\label{gen_ode_split_two}
\end{align}
for some pair of functions $f(x,y)$, $g(x,y)$. By property (b) if $(x(t),y(t))$ is a solution of Eqs.\,\eqref{gen_ode_split_one}-\eqref{gen_ode_split_two} that begins in $S$ then $(x(t),y(t))$ must also be in $S$ for all $t\in \mathbb{R}$. Since all points $s\in S$ have the form $s=(x,y^*(x))$ a solution $(x(t),y(t))$ that begins in $S$ must therefore satisfy
\begin{align}
\begin{pmatrix}
x(t)\\
y(t)
\end{pmatrix}
=
\begin{pmatrix}
x(t)\\
y^*(x(t))
\end{pmatrix},\label{basic_slaving}
\end{align}
for each $t\in\mathbb{R}$.
In other words the variable $y$ is \emph{slaved} to the variable $x$. Moreover, because $(x(t),y(t))$ is a solution of Eqs.\,\eqref{gen_ode_split_one}-\eqref{gen_ode_split_two}, substitution of the slaving relation \eqref{basic_slaving} into the ODE \eqref{gen_ode_split_one}-\eqref{gen_ode_split_two} leads to the pair of equations
\begin{gather}
g^i(x,y^*)\partial_{i}(y^*)^j(x) = f^j(x,y^*)\label{invariance_equation_gen_index}\\ 
\dot{x}^i = g^i(x,y^*),\label{reduced_dynamics_gen_index}
\end{gather}
where the index $i$ ranges from $i=1$ to $i = k$ and $j$ ranges from $j=1$ to $j = d-k$. Introducing the Fr\'echet derivative notation, \citep{Abraham_2008,Marsden_Ratiu_1999}
\begin{align}
Dy^*(x)[\delta x] \equiv \frac{d}{d\epsilon}\bigg|_0 y^*(x+\epsilon \delta x),
\end{align}
these equations may also be written without indices as
\begin{gather}
Dy^*(x)[g(x,y^*(x))] = f(x,y^*(x))\label{invariance_equation_gen}\\
\dot{x} = g(x,y^*(x)).\label{reduced_dynamics_gen}
\end{gather}
Equations \eqref{invariance_equation_gen}-\eqref{reduced_dynamics_gen} provide a powerful way to construct invariant manifolds and understand the reduced dynamics that occur on such manifolds.
%
% Conversely, it is not difficult to show that the graph of any $y^*(x)$ that solves Eq.\,\eqref{invariance_equation_gen} is an invariant manifold. 

Equation \eqref{reduced_dynamics_gen} shows that solutions of Eqs.\,\eqref{gen_ode_split_one}-\eqref{gen_ode_split_two} that begin on the invariant manifold $S$ have the remarkable property that the evolution equation for $x$ closes on itself. Equivalently the function $y^*$ whose graph is equal to the invariant manifold provides a \emph{closure} for the $x$-dynamics. It follows that invariant manifolds comprise a paradigm for dimension reduction. %While dynamics in $z$-space are $d$-dimensional, dynamics on an invariant manifold $S$ in $z$-space are $k$-dimensional, where $k$ may very well be less than $d$. 
Within the framework of Haken's theory of synergetics \cite{Haken_1977}, the variable $x$ comprises an \emph{order parameter}, which determines the evolution of the remaining modes in the system $y$ in accordance with Haken's enslaving principle.

Equation \eqref{invariance_equation_gen}, which arose from substituting the slaving relation \eqref{basic_slaving} into the evolution equation for $y$, plays a role complementary to that of Eq.\,\eqref{reduced_dynamics_gen}. Instead of providing an evolution equation, it provides a (time-independent) first-order nonlinear partial differential equation (PDE) that must be satisfied by the function $y^*(x)$. We will refer to this PDE as \emph{the invariance equation}, as it is a basic tool in the study of invariant manifolds. When we eventually introduce slow manifolds in Section \ref{basic_SM_theory} we will take up the problem of finding asymptotic solutions of the invariance equation.

Apparently all invariant manifolds satisfy the invariance equation. 
The converse is also true. Every solution $y^*$ of the invariance equation gives rise to an invariant manifold equal to the graph of $y^*$. We may therefore summarize the analytic perspective on invariant manifolds as follows.

\begin{proposition}[PDE characterization of invariant manifolds]\label{PDE_char}
The graph of a function $y^*(x)$ is an invariant manifold for the ODE
\begin{align}
\dot{y} =& f(x,y)\\
\dot{x} = & g(x,y)
\end{align}
if and only if $y^*$ is a solution of the invariance equation:
\begin{align}\label{invariance_eqn_gen_prop}
Dy^*(x)[g(x,y^*(x))] = f(x,y^*(x)).
\end{align}
\end{proposition}

The PDE perspective on invariant manifolds provided by Proposition \ref{PDE_char} is most useful when addressing questions such as (1) how smooth should one expect an invariant manifold to be? or (2) can invariant manifolds be constructed using functional fixed-point methods? (See for example the thesis \citep{Riley_2012}.) Asymptotic solution methods like those that will be discussed in this review also fit naturally within the PDE approach. On the other hand the PDE perspective has the disadvantage of obscuring the more geometrically-intuitive picture provided by the definition \ref{invariant_manifold_def}. 

In the remainder of this Review, we will often use the analytic approach to invariant manifold theory when performing explicit calculations, or when describing estimates for the validity time of certain constructions. However, we will refer to the geometric perspective offered by Definition \ref{invariant_manifold_def} whenever possible to build the reader's geometric intuition. Much of this geometric intuition will be extracted from the following ``local" geometric characterization of invariant manifolds.

\begin{proposition}[Invariance as Tangency]\label{tangency_prop}
A submanifold $S\subset Z$ is an invariant manifold for the ODE \eqref{gen_ode} if and only if for each $s\in S$ the vector $U(s)$ is tangent to $S$.
\end{proposition}

\begin{proof}
For those familiar with flows, this follows immediately from the definition \ref{invariant_manifold_def}. However, we will give an explicit proof starting from the PDE picture of invariant manifolds provided by Proposition \ref{PDE_char}.

Move into coordinates $(x,y)$ where the submanifold $S$ is given as the graph of a function $y^*(x)$ and the ODE \eqref{gen_ode} takes the form $\dot{y} = f(x,y)$, $\dot{x} = g(x,y)$. In these coordinates we may characterize vectors tangent to $S$ as follows. A vector $V = (\delta x,\delta y)$ is tangent to $S$ at $s = (x,y^*(x))$ if and only if it is the initial velocity of a curve $(x(t),y(t))$ contained in $S$ and passing through $s$ at $t=0$. That is $V$ must be of the form $\frac{d}{dt}\mid_{t=0}(x(t),y(t))\equiv (\dot{x},\dot{y})$ where $y(t) = y^*(x(t))$ for all $t$ and $(x(0),y(0)) = (x,y^*(x))$. By the chain rule this means
\begin{align}
\delta y =& Dy^*(x)[\dot{x}]\\
\delta x =& \dot{x},
\end{align}
and that the general form of a vector tangent to $S$ at $s$ is given by $V = (\dot{x},Dy^*(x)[\dot{x}])$, where $\dot{x}$ is any $k$-component vector.

If $S$ is invariant then $y^*$ must satisfy the invariance equation \eqref{invariance_eqn_gen_prop}. Therefore if $s=(x,y^*(x))$ is any point on $S$, the value of $U = (g,f)$ at $s$ is given by $(g(x,y^*(x),f(x,y^*(x)))$ where
\begin{align}
f(x,y^*(x)) = Dy^*(x)[g(x,y^*(x))].
\end{align}
This vector has the form $(\dot{x},Dy^*(x)[\dot{x}])$ with $\dot{x} = g(x,y^*(x))$. By the above characterization of vectors tangent to $S$, this implies $U(s)$ is tangent to $S$.

Conversely, suppose that $U(s)$ is tangent to $S$ for all $s\in S$. Then $U(s)$ must be of the form $(\dot{x},Dy^*(x)[\dot{x}])$ for some $\dot{x}$. But we know that $U(s) = (g(x,y^*(x)),f(x,y^*(x)))$ as well. Therefore it must be the case that $\dot{x} = g(x,y^*(x))$. But this means
\begin{align}
f(x,y^*(x)) = Dy^*(x)[\dot{x}] = Dy^*(x)[g(x,y^*(x))],
\end{align}
for all $x$, as claimed.

\end{proof}

Further elements of the theory of invariant manifolds will not be necessary in what follows. Nevertheless, the theory is much richer than this Section's discussion suggests. We refer the reader to \citep{Kirchgraber_1986} for an application to multi-step numerical integrators, \citep{Robinson_1977} for a use of invariant manifold theory to establish so-called ``shadowing theorems," \citep{Capinsky_2016} for a fantastic application of invariant manifold theory in the context of Arnold diffusion, and \citep{Llave_2001} for a tutorial on the venerable Kolmogorov-Arnold-Moser (KAM) theory of persistent invariant tori in nearly-integrable Hamiltonian systems.

\section{Fast-Slow Systems\label{FS_systems_sec}}
%Fast-slow dynamical systems provide the basic framework for slow manifold theory. 
If slow manifold reduction were a game then fast-slow systems would comprise the arena where the game is (usually) played. The purpose of this Section is to build a working understanding of that arena. Fast-slow systems will be defined, and techniques will be described for detecting fast-slow systems ``in the wild." We will argue that these techniques are important because a system's fast-slow identity may easily be hidden by obvious choices of a model's dependent variables. The Abraham-Lorentz model described in Section \ref{AL_section} will be used as a prototypical plasma-relevant example illustrating the theory.

\subsection{What is a fast-slow system?}\label{what_is_sec}
The notion of a fast-slow system is a refinement of the more primordial notion of a singularly-perturbed ODE. Our discussion of fast-slow systems therefore begins with generic singularly-perturbed ODEs of the type
\begin{align}
\epsilon\, \dot{y} =& f_\epsilon(x,y)\label{spODE_one}\\
\dot{x} =& g_\epsilon(x,y),\label{spODE_two}
\end{align}
where $\epsilon\ll 1$ is a small positive parameter and $f_\epsilon,g_\epsilon$ depend smoothly on $\epsilon$ in a neighborhood of $\epsilon = 0$. We assume that $x$ and $y$ live in vector spaces $X$ and $Y$ whose dimensions may differ, or even be infinite. If either of $X$ or $Y$ is infinite-dimensional we assume that space is normed and complete with respect to that norm. (In other words we require $X$ and $Y$ to be Banach spaces.)

Because $f_\epsilon$ and $g_\epsilon$ depend smoothly on $\epsilon$ they admit the formal power series expansions
\begin{align}
f_\epsilon =& f_0 + \epsilon f_1 + \epsilon^2 f_2 + \dots\\
g_\epsilon =& g_0 + \epsilon g_1 + \epsilon ^2 g_2 + \dots, 
\end{align}
where the $f_k,g_k$ are, up to constant multipliers, Taylor coefficients. When $f_0\neq 0$ the singularly-perturbed ODE \eqref{spODE_one}-\eqref{spODE_two} therefore exhibits an extremely-short $O(\epsilon)$ timescale. On the other hand, if there are regions in $(x,y)$-space where $f_0(x,y) = 0 $ then the dynamical timescale is $O(1)$ in those regions. It follows that Eqs.\,\eqref{spODE_one}-\eqref{spODE_two} typically comprise an example of a system with multiple timescales. Accordingly, in systems of the form \eqref{spODE_one}-\eqref{spODE_two} $y$ is referred to as the fast variable and $x$ is referred to as the slow variable.

%The point of introducing the more refined notion of a fast-slow system is to focus on systems that exhibit this multi-scale behavior  be Without some additional assumptions on $f_\epsilon$ and $g_\epsilon$ it is difficult to say much about the qualitative behavior of such a general dynamical system.

We will now bring into focus a special class of singularly-perturbed ODEs for which the aforementioned multi-scale behavior manifests in a particularly organized fashion.

\begin{definition}[Fast-slow system]\label{FS_def}
A \emph{fast-slow system} (c.f. \citep{MacKay_2004,Fenichel_1979}) is an ODE of the form 
\begin{align}
\epsilon \,\dot{y} =& f_\epsilon(x,y)\\
\dot{x} = & g_\epsilon(x,y),
\end{align}
where $f_\epsilon,g_\epsilon$ are smooth functions of the parameter $\epsilon$ and $f_0$ satisfies the condition
\begin{align}
D_yf_0(x,y)\,\text{is invertible whenever }f_0(x,y)=0.\label{FS_condition}
\end{align}
Here $D_yf_0(x,y):Y\rightarrow Y$ is the linear map given by $D_yf_0(x,y)[\delta y] = \frac{d}{d\lambda}\big|_0 f_0(x,y+\lambda \delta y)$.
\end{definition}
While the condition \eqref{FS_condition} appears to be quite technical it is simple to motivate by considering the $\epsilon\rightarrow 0$ limit of Eqs.\,\eqref{spODE_one}-\eqref{spODE_two}. Suppose those equations admit an $\epsilon$-dependent family of solutions $(x_\epsilon(t),y_\epsilon(t))$ that \emph{does not} sample the $O(\epsilon)$ timescale. Heuristically we are assuming the existence of a \emph{slow solution}. Then because $\dot{y}_\epsilon = O(1)$ as $\epsilon\rightarrow 0$ the limiting solution $(x_0(t),y_0(t))$ must satisfy
\begin{align}
0 &= f_0(x_0(t),y_0(t))\label{limit_FS_one}\\
\dot{x}_0(t) & = g_0(x_0(t),y_0(t)).\label{limit_FS_two}
\end{align}
These limiting differential-algebraic equations (DAEs) may be quite complicated in general due to hidden differential constraints (\emph{secondary constraints} in the language of Dirac constraint theory, c.f. \citep{Gotay_1978} ) implied by Eq.\,\eqref{limit_FS_one}. (See \citep{Gear_1988} for a discussion of the \emph{differentiation index} for DAEs.) However, they are simple to understand in the special case where
\begin{align}
f_0(x,y) = 0\label{limit_constraint}
\end{align}
can be solved to give $y$ as a unique function of $x$, i.e. $y = y^*_0(x)$. Indeed, if $f_0 = 0$ implies $y = y_0^*(x)$ then Eqs.\,\eqref{limit_FS_one}-\eqref{limit_FS_two} reduce to the ordinary differential equation
\begin{align}
\dot{x}_0(t) = g_0(x_0(t),y_0^*(x_0(t))).\label{limit_ODE}
\end{align}
The condition \eqref{FS_condition} is a natural way to ensure that this reduction occurs; the implicit function theorem states that if condition \eqref{FS_condition} is satisfied then, at least locally, Eq.\,\eqref{limit_constraint} can be solved uniquely to give $y = y_0^*(x)$. In this Review we will always tacitly assume that the equation $f_0(x,y) =0$ has a global unique solution $y=y_0^*(x)$; examples of fast-slow systems with multi-branched solutions of $f_0=0$ may arise in some plasma physical contexts, but these examples lie beyond the scope of our discussion.

Another way to motivate condition \eqref{FS_condition} is to consider generic singularly-perturbed ODEs. In generic systems the function $f_0$ will be a generic $Y$-valued function of $(x,y)$. Therefore the derivative $D_yf_0(x,y):Y\rightarrow Y$ will be a generic linear map for almost all $x$. Because the space of invertible linear maps $Y\rightarrow Y$ is open the generic map $D_yf_0(x,y)$ should fall into this space for almost all $x$. Alternatively we may say that if $f_0$ does not satisfy \eqref{FS_condition} then we may ``fix" $f_0$ by subjecting it to an arbitrarily-small perturbation. Therefore ``most" singularly-perturbed systems of the form \eqref{spODE_one}-\eqref{spODE_two} will satisfy \eqref{FS_condition}.

Each of the previous arguments motivating condition \eqref{FS_condition} has its own weakness. The first argument invoking the implicit function theorem begs the question ``why should we demand that the DAE \eqref{limit_FS_one}-\eqref{limit_FS_two} be as simple as possible?" The second argument invoking genericity begs the question ``why should \emph{my} model be generic?" The flaws in these arguments indicate that restricting our attention to fast-slow systems may rule out certain practically interesting multi-scale phenomena. We will return to this point in Section \ref{QN_application} where we will give a detailed treatment of a significant plasma-physical example that lies outside the scope of fast-slow systems theory. Nevertheless, in this Section and the next several Sections we shall remain steadfast and devote significant attention to fast-slow systems because such systems do arise frequently in plasma physics. (c.f. Table \ref{table1}.)

\begin{example}\label{FS_example_easy}
The Abraham-Lorentz equations in the weak-drag regime, i.e. Eqs.\, \eqref{weakly_damped_one}-\eqref{weakly_damped_three}, comprise a simple and fundamental example of a fast-slow system relevant to plasma physics. As written this system does not comprise a singularly-perturbed ODE of the type \eqref{spODE_one}-\eqref{spODE_two} because there are effectively \emph{two} short timescales, one $O(\epsilon)$, the other $O(\epsilon^2)$. This issue can be remedied easily, however, by \emph{zooming in} on the $O(\epsilon)$ timescale. Let $\tau = t/\epsilon$ be the fast time variable associated with this magnification. In terms of $\tau$ the weak-drag Abraham-Lorentz equations become
\begin{align}
\epsilon \frac{2}{3} \frac{d\bm{a}}{d\tau} =& \bm{a} - \zeta\,\bm{v}\times\bm{B}(\bm{x})\label{rescaled_AL_one}\\
\frac{d\bm{v}}{d\tau} = &\bm{a}\\
\frac{d\bm{x}}{d\tau} = & \epsilon \bm{v}.\label{rescaled_AL_two}
\end{align}
Equations \eqref{rescaled_AL_one}-\eqref{rescaled_AL_two} comprise a singularly-perturbed ODE of type \eqref{spODE_one}-\eqref{spODE_two} with $y = \bm{a}$, $x = (\bm{x},\bm{v})$, and 
\begin{align}
f_\epsilon(x,y) =& \frac{3}{2}\bm{a} -\frac{3}{2} \zeta\,\bm{v}\times\bm{B}(\bm{x})\\
g_\epsilon(x,y) = & (\epsilon \bm{v}, \bm{a}).
\end{align}
Moreover, the $y$-derivative $D_yf_0(x,y)$ is readily computed as
\begin{align}
D_yf_0(x,y)[\delta y] = \frac{3}{2}\delta\bm{a},
%\begin{pmatrix}
%\frac{3}{2} & 0 & 0 \\
%0 & \frac{3}{2} & 0\\
%0 & 0 & \frac{3}{2}
%\end{pmatrix}
%\begin{pmatrix}
%\delta a_x\\
%\delta a_y\\
%\delta a_z
%\end{pmatrix},
\end{align}
which shows that $D_yf_0(x,y)$ is invertible for all $(x,y)$, in particular for those $(x,y)$ where $f_0(x,y) = 0$. It follows that condition \eqref{FS_condition} is satisfied and that the Abraham-Lorentz equations in the weak-drag regime comprise a fast-slow system.

\end{example}

\begin{example}\label{FS_example_hard}
Abraham-Lorentz dynamics in the zero drag regime ($\gamma = \infty$, c.f. Section \ref{zero_drag_sec}) provide a more typical and interesting example of the manner in which fast-slow systems arise in plasma physics. Equations \eqref{zero_drag_one}-\eqref{zero_drag_two}, in contrast to the weak drag regime equations, immediately take the form \eqref{spODE_one}-\eqref{spODE_two} of a singularly-perturbed ODE with $x = \bm{x}$, $y = \bm{v}$, and 
\begin{align}
f_\epsilon(x,y) =& \zeta\,\bm{v}\times\bm{B}(\bm{x})\\
g_\epsilon(x,y) = & \bm{v}.
\end{align} 
However the zero-drag equations \emph{do not} comprise a fast-slow system because
\begin{align}
D_yf_0(x,y)[\delta y] = & \zeta \,\delta\bm{v}\times\bm{B}(\bm{x}),
\end{align}
which clearly vanishes whenever $\delta \bm{v} = \lambda \bm{B}(\bm{x})$ for any $\lambda\in\mathbb{R}$. Thus $D_yf_0(x,y)$ is not invertible \emph{for any} $(x,y)$, in violation of condition \eqref{FS_condition}. It would therefore seem that fast-slow system theory cannot be applied to Abraham-Lorentz dynamics in the zero drag regime. 

On the other hand suppose we introduce new coordinates on $(\bm{x},\bm{v})$-space, $(\bm{x},v_1,v_2,v_\parallel)$, defined by the formula
\begin{align}
\bm{v} = v_\parallel \bm{v}(\bm{x}) + v_1 \,\bm{e}_1(\bm{x}) + v_2\,\bm{e}_2(\bm{x}),\label{moving_frame_v}
\end{align}
where $\bm{b} = \bm{B}/|\bm{B}|$ is the unit vector along the magnetic field, and $(\bm{e}_1,\bm{e}_2,\bm{b})$ form a space-dependent right-handed orthonormal frame, or \emph{moving frame}. 
%This coordinate transformation amounts to decomposing the velocity vector $\bm{v}$ into components along the moving frame. 
In these new coordinates the zero-drag Abraham-Lorentz equations become
\begin{align}
\epsilon \,\dot{v}_1 =& \phantom{-}\zeta \,|\bm{B}(\bm{x})|\,v_2 + \epsilon \bm{v}\cdot\bm{R}(\bm{x})\,v_2 - \epsilon v_\parallel \,\bm{v}\cdot\nabla\bm{b}\cdot\bm{e}_1\\
\epsilon\, \dot{v}_2 = &\hspace*{.45em}\text{-}\zeta\,|\bm{B}(\bm{x})|\,v_1 - \epsilon\bm{v}\cdot \bm{R}(\bm{x})\,v_1-\epsilon v_\parallel \,\bm{v}\cdot\nabla\bm{b}\cdot\bm{e}_2\\
\dot{v}_\parallel = &\bm{v}\cdot \nabla\bm{b}\cdot\bm{v}\\
\dot{\bm{x}} = & \bm{v},
\end{align}
where $\bm{R} = (\nabla\bm{e}_1)\cdot\bm{e}_2$ is the so-called ``gyrogauge vector" elucidated by Littlejohn in \citep{Littlejohn_1984}, and $\bm{v}$ is now being used as shorthand notation for formula \eqref{moving_frame_v}.
Again these equations take the form \eqref{spODE_one}-\eqref{spODE_two} of a singularly-perturbed ODE with slow variable $x = (\bm{x},v_\parallel)$, fast variable $y = (v_1,v_2)$, $g_\epsilon = (g_\epsilon^{\bm{x}},g_\epsilon^{v_\parallel})$, $f_\epsilon = (f_\epsilon^{v_1},f_\epsilon^{v_2})$, and
\begin{align}
f_\epsilon^{v_1} = & \phantom{-}\zeta \,|\bm{B}(\bm{x})|\,v_2 + \epsilon \bm{v}\cdot\bm{R}(\bm{x})\,v_2 - \epsilon v_\parallel \,\bm{v}\cdot\nabla\bm{b}\cdot\bm{e}_1 \label{zero_drag_new_Feps_one}\\
f_\epsilon^{v_2} = &\hspace*{.45em}\text{-}\zeta\,|\bm{B}(\bm{x})|\,v_1 - \epsilon\bm{v}\cdot \bm{R}(\bm{x})\,v_1-\epsilon v_\parallel \,\bm{v}\cdot\nabla\bm{b}\cdot\bm{e}_2 \label{zero_drag_new_Feps_two}\\
g_\epsilon^{v_\parallel} = &\bm{v}\cdot \nabla\bm{b}\cdot\bm{v}  \\
g_\epsilon^{\bm{x}} = &\bm{v}.\label{zero_drag_new_Geps_two}
\end{align}
Apparently the coordinate transformation has changed the number of fast and slow variables. Moreover, the new form \eqref{zero_drag_new_Feps_one}-\eqref{zero_drag_new_Feps_two} for $f_\epsilon$ satisfies
\begin{align}
D_yf_0(x,y)[\delta y] = & \begin{pmatrix}
0 & \zeta\,|\bm{B}(\bm{x})|\\
-\zeta\,|\bm{B}(\bm{x})| & 0 
\end{pmatrix}\begin{pmatrix}
\delta v_1\\
\delta v_2
\end{pmatrix}
\end{align}
which vanishes if and only if $\delta y \equiv (\delta v_1,\delta v_2) = 0$. In terms of the dependent variables $(\bm{x},v_\parallel, v_1, v_2)$ the zero-drag Abraham-Lorentz equations are fast-slow!
\end{example}

Example \ref{FS_example_hard} illustrates an important practical point regarding detection of fast-slow systems when working with models from physics. The technical condition \eqref{FS_condition} defining a fast-slow system, if satisfied for some set of dependent variables, \emph{is not satisfied for every other choice of dependent variables.} Moreover, if a particular system does \emph{not} admit such a special set of dependent variables one's deep-held desire to put singularly-perturbed systems in fast-slow form could easily lead to an arduous, fruitless search. 
We are therefore forcibly lead to study a broader class of singularly-perturbed systems --- those that admit \emph{fast-slow splits}.

\begin{definition}[fast-slow split]\label{fs_split_def}
A singularly-perturbed ODE of form
\begin{align}
\epsilon\,\dot{z} = U_\epsilon(z)\label{fs_split_ode}
\end{align}
with $U_\epsilon(z)$ a smooth function of $\epsilon$ and $z$ admits a \emph{fast-slow split} if there is an $\epsilon$-dependent invertible change of dependent variables $z\mapsto (\overline{x}_\epsilon,\overline{y}_\epsilon)$ that (a) depends smoothly on $\epsilon$ along with its inverse,  and (b) transforms the ODE \eqref{fs_split_ode} into a fast-slow system.
\end{definition}

\noindent Apparently the Abraham-Lorentz equations in the weak-drag regime both comprise a fast-slow system and admit a fast-slow split. (The coordinate transformation is just the identity map.) On the other hand the same equations in the zero-drag regime do not comprise a fast-slow system but do admit a non-trivial fast-slow split. 

The crucial difficulty associated with the notion of a fast-slow split is assessing whether a given singularly-perturbed system admits a fast-slow split, and then finding the associated coordinate change when appropriate. The following subsection will provide practically useful tools for addressing this difficulty.

\subsection{When does a system admit a fast-slow split?\label{when_split}}
Having established that some singularly-perturbed ODEs can be fast-slow systems in disguise we now turn to the question of how to detect such systems ``in the wild." While this is not an easy question to answer in general, a satisfying partial resolution is given by the following coordinate-independent test. \citep{Noethen_2011}

\begin{proposition}[The transversality test]\label{FS_split_test}
Consider the ODE 
\begin{align}
\epsilon \,\dot{z} = U_\epsilon(z),\label{thm_ode}
\end{align}
with $U_\epsilon(z)$ a smooth function of $\epsilon$ and $z$. If Eq.\,\eqref{thm_ode} admits a fast-slow split, then the the following transversality condition must be satisfied:
\begin{itemize}
\item[(T)] The image $\mathrm{im}\, DU_0(z)\subset Z$ and the kernel $\mathrm{ker}\,DU_0(z)\subset Z$ are complementary subspaces for all $z$ where $U_0(z) = 0$.
\end{itemize}
\end{proposition}

\begin{proof}

Since $\epsilon\,\dot{z} = U_\epsilon(z)$ admits a fast-slow split there are coordinates $(x,y)$ on $z$-space in which the ODE takes the form $\epsilon\,\dot{y} = f_\epsilon(x,y)$, $\dot{x} = g_\epsilon(x,y)$ with $f_\epsilon$ satisfying condition \eqref{FS_condition}. Let 
\begin{align}
\overline{U}_\epsilon(x,y) = \begin{pmatrix}
\epsilon\,g_\epsilon(x,y)\\
f_\epsilon(x,y)
\end{pmatrix},
\end{align}
and write the derivative of $\overline{U}_\epsilon(x,y)$ at $\epsilon = 0$ in the block-matrix form
\begin{align}
D\overline{U}_0(x,y)= \begin{pmatrix}
0 & 0\\
D_xf_0(x,y) & D_yf_0(x,y)
\end{pmatrix}.
\end{align}

First we will establish that the intersection of $\text{im }D\overline{U}_0(x,y)$ with $\text{ker }\overline{U}_0(x,y)$ is trivial whenever $\overline{U}_0(x,y) = 0$.
If $\overline{U}_0(x,y) = 0$ and $(\delta x,\delta y)$ is simultaneously in the image and the kernel of $D\overline{U}_0(x,y)$, then
\begin{align}
\begin{pmatrix}
0 & 0\\
D_xf_0(x,y) & D_yf_0(x,y)
\end{pmatrix}\begin{pmatrix}
\delta x\\
\delta y
\end{pmatrix} &= 0\\
\begin{pmatrix}
0 & 0\\
D_xf_0(x,y) & D_yf_0(x,y)
\end{pmatrix}\begin{pmatrix}
\delta x^\prime\\
\delta y^\prime
\end{pmatrix} &= \begin{pmatrix}
\delta x\\
\delta y
\end{pmatrix},\label{image_condition}
\end{align}
for some $(\delta x^\prime,\delta y^\prime)$. In particular, because $\overline{U}_0(x,y) = 0$ if and only if $f_0(x,y) = 0$, $D_yf_0(x,y)$ is invertible and
\begin{align}
\delta y = - [D_yf_0(x,y)]^{-1}D_xf_0(x,y)[\delta x].\label{graph_thm_cond}
\end{align}
But Eq.\,\eqref{image_condition} implies that $\delta x=0$. Therefore the formula \eqref{graph_thm_cond} implies that both $\delta x$ and $\delta y$ must be zero. This means the intersection of $\text{im }D\overline{U}_0(x,y)$ with $\text{ker }\overline{U}_0(x,y)$ is trivial.

Next we will show that if $\Phi_\epsilon:\overline{z}\mapsto z$ is any smooth $\epsilon$-dependent change of coordinates and property (T) is satisfied in the coordinates $\overline{z}$ then (T) must also be satisfied in the original coordinates $z$. The usual change of variables formula for ODEs implies that $\epsilon\,\dot{z} = U_\epsilon(z)$ becomes $\epsilon\, \dot{\overline{z}} = \overline{U}_\epsilon(\overline{z})$ in $\overline{z}$-coordinates, where the vector fields $U_\epsilon$ and $\overline{U}_\epsilon$ are related by
\begin{align}
U_\epsilon(z) = D\Phi_\epsilon(\overline{z}_\epsilon^*)[\overline{U}_\epsilon(\overline{z}_\epsilon^*)],
\end{align}
with $\overline{z}_\epsilon^* = \overline{z}_\epsilon^*(z) = \Phi_\epsilon^{-1}(z)$. Consequently if $z$ is in the zero level set of $U_0$, then $z_0^* = \Phi_0^{-1}(z)$ must be in the zero level set of $\overline{U}_0$. It follows that the derivative of $U_0$ evaluated at $z$ in the zero level of $U_0$ must be given by
\begin{align}
DU_0(z)[\delta z] = D\Phi_0(\overline{z}_0^*)[D\overline{U}_0(z_0^*)[Dz_0^*(z)[\delta z]]],
\end{align}
where the usual terms involving second derivatives of $\Phi_0$ vanish because $\overline{U}_0(z_0^*) = 0$.
Because $z_0^*(z) = \Phi_0^{-1}(z)$, the derivative matrix $Dz_0^*(z)$ is the inverse of $M = D\Phi_0(\overline{z}_0^*)$. This means the matrix $DU_0(z)$ is related to the matrix $D\overline{U}_0(z_0^*)$ by the similarity transformation
\begin{align}
DU_0(z) = M\, D\overline{U}_0(z_0^*)\,M^{-1}.
\end{align}
This establishes property (T) in the $z$-coordinates for the following reason. Suppose $\delta z$ is simultaneously in the image and the kernel of $DU_0(z)$. (Remember that we are currently assuming $U_0(z)=0$.) Because $0 = DU_0(z)[\delta z] = M D\overline{U}_0(z_0^*)M^{-1}\delta z$ and $M$ is invertible the vector $\delta\overline{z} = M^{-1}\delta z$ must be in the kernel of $D\overline{U}_0(z_0^*)$. Because $\delta z = DU_0(z)[\delta z^\prime] = MD\overline{U}_0(z_0^*)M^{-1}\delta z^\prime$ for some $\delta z^\prime$, $\delta \overline{z}$ must also be of the form $\delta\overline{z} = D\overline{U}_0(z^*_0)M^{-1}\delta z^\prime$. In other words, $\delta \overline{z} = M^{-1}\delta z$ must simultaneously be in the image and kernel of $D\overline{U}_0(z_0^*)$. Therefore if $\delta z$ were non-zero there would be a non-zero $\delta \overline{z}$ in the intersection of the image and kernel of $D\overline{U}_0(z_0^*)$, in violation of property (T) in $\overline{z}$-coordinates. 

Finally we remark that the argument in the previous paragraph applied to the coordinate transformation $\Phi_\epsilon: (x,y)\mapsto z$ implies the desired result.
\end{proof}

Heuristically, the transversality test justifies searching for a hidden fast-slow split when appropriate, and prevents needless searching when inappropriate. What the test \emph{does not} do is identify the coordinate change comprising a fast-slow split. In fact the transversality test is merely a \emph{necessary} condition for a system to admit a fast-slow split. As discussed in \citep{Noethen_2011}, formulation of simple necessary and sufficient conditions for existence of a fast-slow split is a subtle matter.

Readers interested in technical statements of such necessary and sufficient conditions are encouraged to consult Section 2.1 of \citep{Noethen_2011}. There is, however, a useful and less technical rule of thumb for how to go about searching for a fast-slow split when a system passes the transversality test. Consider the ODE \eqref{thm_ode} on the fast timescale $\tau=t/\epsilon$, i.e. $dz/d\tau = U_\epsilon(z)$. Assuming \eqref{thm_ode} passes the transversality test the key to identifying a fast-slow split is to find sufficiently many independent conservation laws for the limiting fast-time dynamical system
\begin{align}
\frac{dz}{d\tau} = U_0(z).\label{limiting_fast_time}
\end{align}

It is easy to understand why finding conservation laws for \eqref{limiting_fast_time} should be a necessary condition for existence of a fast-slow split. If there is a hidden fast-slow split given by a coordinate transformation $\Phi_\epsilon: z\mapsto (x,y)$ then the limiting fast-time dynamics in $(x,y)$-coordinates must be given by
\begin{align}
\frac{dy}{d\tau} =& f_0(x,y)\\
\frac{dx}{d\tau} = &0.
\end{align}
Apparently each component of the slow variable $x$ is a conserved quantity in these coordinates. This implies that the components of $x$, when expressed in $z$-coordinates, must be conserved quantities for \eqref{limiting_fast_time}. We say that the slow variable $x$ must be a \emph{limiting conserved quantity}. In other words, if a fast-slow split exists then there must be as many conserved quantities for \eqref{limiting_fast_time} as there are components of the slow variable $x$ in the fast-slow split.

A practical procedure for detecting and finding fast-slow splits may therefore be summarized as follows.
\begin{enumerate}
\item Write the system under consideration in the form \eqref{thm_ode} and then apply the transversality test. (c.f. Proposition \ref{FS_split_test}.)
\item If the system fails the transversality test because condition (T) is not satisfied then stop. No fast-slow split can be found.
\item If the system passes the transversality test look for independent conserved quantities for the limiting fast-time dynamics \eqref{limiting_fast_time}. 
\item If the number of independent limiting conserved quantities is as large as the kernel of $DU_0(z)$ along the level set $U_0(z)=0$ then a fast-slow split likely exists with the slow variable given by the limiting conserved quantities.
\end{enumerate}

\begin{example}\label{FS_test_example_pass}
Let us apply this procedure to Abraham-Lorentz dynamics in the zero-drag regime. First we write Eqs.\,\eqref{zero_drag_one}-\eqref{zero_drag_two} in the form \eqref{thm_ode} by setting 
\[
z = \begin{pmatrix}
\bm{x}\\
\bm{v}
\end{pmatrix}
\]
and
\begin{align}
U_\epsilon(z) = \begin{pmatrix}
\epsilon\,\bm{v}\\
\zeta\,\bm{v}\times\bm{B}(\bm{x})
\end{pmatrix}.
\end{align}

Next we apply the transversality test as follows. First we note that $U_0(z)=0$ if and only if $\bm{v} = v_\parallel\bm{b}(\bm{x})$ for some $v_\parallel\in\mathbb{R}$. Therefore we may explicitly compute the derivative $DU_0(z)$ along the zero level of $U_0$ as
\begin{align}
DU_0(z)= \begin{pmatrix}
0 & 0\\
\zeta\,v_\parallel\bm{b}_\times\cdot\nabla\bm{B}^T  & - \zeta\,|\bm{B}|\bm{b}_\times
\end{pmatrix},
\end{align}
where the dyad $\bm{b}_\times$ is defined according to
\begin{align}
\bm{b}_\times\cdot\bm{w} & = \bm{b}\times\bm{w}.
\end{align}
It follows that the image of $DU_0(z)$ comprises all vectors of the form  
\begin{align}
\delta z_{\text{im}} = \begin{pmatrix}
0\\
\bm{w}_\perp
\end{pmatrix},
\end{align}
where $\bm{w}_\perp$ is any $3$-component vector perpendicular to $\bm{b}(\bm{x})$.
An image vector $\delta z_{\text{im}}$ will also be in the kernel of $DU_0(z)$ if and only if
\begin{align}
0=-\zeta |\bm{B}|\bm{b}_\times\cdot\bm{w}_\perp,
\end{align}
which is equivalent to $\bm{w}_\perp = 0$.
Thus the intersection of the image and kernel of $DU_0(z)$ contains only the zero vector. It follows that Abraham-Lorentz dynamics in the zero-drag regime pass the transversality test.

Because the transversality test returned a positive result we now look for conserved quantities of the limiting fast-time dynamics \eqref{limiting_fast_time}, which in this case are given explicitly by
\begin{align}
\frac{d\bm{v}}{d\tau} &= \zeta\,\bm{v}\times\bm{B}(\bm{x})\label{limit_v_zero_drag}\\
\frac{d\bm{x}}{d\tau}& = 0.\label{limit_x_zero_drag}
\end{align}
By Eq.\,\eqref{limit_x_zero_drag} the components of the position vector $\bm{x}$ comprise three independent limiting conserved quantities. These conserved quantities are insufficient to construct a fast-slow split, however, because the kernel of $DU_0(z)$ has dimension greater than three. Indeed, it is straightforward to verify that the general vector $\delta z_{\text{ker}} = (\delta\bm{x}_{\text{ker}},\delta\bm{v}_{\text{ker}})$ in the kernel of $DU_0(z)$ takes the form
\begin{align}
\delta \bm{x}_{\text{ker}} &= \delta\bm{x}\\
\delta\bm{v}_{\text{ker}} &= \delta v_\parallel \bm{b} -|\bm{B}|^{-1}\,v_\parallel\,\bm{b}_{\times}\cdot\bm{b}_\times\cdot (\delta\bm{x}\cdot\nabla\bm{B}),
\end{align}
which exhibits four independent free parameters $\delta\bm{x}$ and $\delta v_\parallel$. We therefore need one additional conserved quantity for Eqs.\,\eqref{limit_v_zero_drag}-\eqref{limit_x_zero_drag} before a fast-slow split can be found. The ``missing'' limiting conserved quantity emerges by computing the scalar product of Eq.\,\eqref{limit_v_zero_drag} with $\bm{b}(\bm{x})$ as
\begin{align}
\frac{d}{d\tau}(\bm{b}\cdot \bm{v}) = \bm{b}(\bm{x})\cdot \frac{d\bm{v}}{d\tau} = &\zeta\,\bm{b}(\bm{x})\cdot (\bm{v}\times\bm{B}(\bm{x})) = 0.
\end{align}
Thus $v_\parallel \equiv \bm{v}\cdot \bm{b}(\bm{x})$ is a fourth independent limiting conserved quantity.

As a final step we complete the partially-defined coordinate system $(\bm{x},v_\parallel)$ comprising the limiting conserved quantities by identifying a pair of additional dependent variables that are independent of $\bm{x}$ and $v_\parallel$. The perpendicular components of $\bm{v}$ relative to a moving frame $(\bm{e}_1,\bm{e}_2,\bm{b})$, i.e. $v_1 = \bm{v}\cdot \bm{e}_1$ and $v_2 = \bm{v}\cdot\bm{e}_2$, comprise such a pair. We now reasonably expect that the zero-drag Abraham-Lorentz equations will comprise a fast-slow system when expressed in terms of the dependent variables $(\bm{x},v_\parallel,v_1,v_2)$ with the slow variable $\bm{x} = (\bm{x},v_\parallel)$ and the fast variable $y = (v_1,v_2)$. This was verified explicitly in Example \ref{FS_example_hard}. We therefore conclude that the general procedure described in this subsection is sufficient to recover the fast-slow split for the zero-drag Abraham-Lorentz equations.
\end{example}

\begin{example}\label{FS_test_example_fail}
Let us also apply our procedure to Abraham-Lorentz dynamics in the maximal drag regime. ($\gamma = 0$, c.f. Section \ref{maximal_drag_sec}.) As a preparatory step we re-scale the acceleration variable $\bm{a}$ according to $\bm{a} = \epsilon \bm{\alpha}$. (Note that this scaling transformation is not invertible when $\epsilon=0$, and therefore cannot possibly be subsumed into the coordinate transformation $\Phi_\epsilon$ defining a fast-slow split.) In terms of $\bm{\alpha}$ the maximal drag equations may be written
\begin{align}
\epsilon\,\frac{2}{3} \dot{\bm{\alpha}} =& \epsilon\,\bm{\alpha} - \zeta\,\bm{v}\times\bm{B}\\
\dot{\bm{v}} =& \bm{\alpha}\\
\dot{\bm{x}}=&\bm{v}. 
\end{align}
As in Example \ref{FS_test_example_pass}, we begin by writing Eqs.\,\eqref{maximally_damped_one}-\eqref{maximally_damped_three} in the form \eqref{thm_ode} by setting 
\begin{align}
z = \begin{pmatrix}
\bm{x}\\
\bm{v}\\
\bm{\alpha}
\end{pmatrix}
\end{align}
and 
\begin{align}
U_\epsilon(z) = \begin{pmatrix}
 \epsilon\bm{v}\\
\epsilon\bm{\alpha}\\
- \frac{3}{2}\,\zeta\,\bm{v}\times\bm{B}(\bm{x}) + \frac{3}{2}\epsilon\bm{\alpha}
\end{pmatrix}.
\end{align}

In order to apply the transversality test we note that $U_0(z)=0$ if and only if
\begin{align}
\bm{v} = & v_\parallel \bm{b}(\bm{x}),
\end{align}
for some $v_\parallel\in\mathbb{R}$.
We may therefore compute the derivative $DU_0(z)$ along the zero level of $U_0$ explicitly as
\begin{align}
DU_0(z) = \begin{pmatrix}
0 & 0& 0\\
0 & 0& 0\\
-\frac{3}{2}\,\zeta\,v_\parallel\,\bm{b}_\times\cdot\nabla\bm{B}^T & \frac{3}{2}\,\zeta\,|\bm{B}|\bm{b}_\times & 0
\end{pmatrix}.
\end{align}
The image of this matrix comprises the collection of vectors of the form
\begin{align}
\delta z_{\text{im}} = \begin{pmatrix}
0\\
0\\
\bm{w}_\perp
\end{pmatrix}
\end{align}
where $\bm{w}_\perp$ is any $3$-component vector perpendicular to $\bm{b}(\bm{x})$. Such image vectors $\delta z_{\text{im}}$ are also in the kernel of $DU_0(z)$ because the right-most column of $DU_0(z)$ contains only zeros. It follows that Abraham-Lorentz dynamics in the maximal drag regime \emph{fail} the transversality test. Therefore this system \emph{cannot} admit a fast-slow split. 
\end{example}

The previous example illustrates that one's ability to apply fast-slow systems theory to a given system exhibiting a singular perturbation structure depends sensitively on one's chosen asymptotic scaling. In one scaling regime the transversality test may pass (e.g. the weak-drag Abrahm-Lorentz equations) while in another scaling regime for the same system (e.g. the maximal-drag Abraham-Lorentz equations) the transversality test may fail. As is usually the case, careful physical, and sometimes numerical considerations are necessary to determine which scaling regime matters most in any given scenario.

In any case the fact that some singularly-perturbed systems do not admit fast-slow splits motivates the introduction of yet another class of singularly-perturbed system --- those systems that fail the transversality test.

\begin{definition}[Degenerate fast-slow system]
A degenerate fast-slow system is an ODE of the form
\begin{align}
\epsilon \dot{z} = U_\epsilon(z),
\end{align}
where $U_\epsilon$ depends smoothly on $\epsilon$ and
\begin{align}
\text{im }DU_0(z)\,\cap\,\text{ker }DU_0(z)\neq \emptyset
\end{align}
for some $z$ where $U_0(z) = 0$.
\end{definition}

\noindent We will postpone further discussion of the largely unexplored topic of degenerate fast-slow systems to Section \ref{QN_application}.

\section{Basic Slow Manifold Theory\label{basic_SM_theory}}
We will now take up the topic of slow manifold reduction in earnest. Going forward we will assume the reader is familiar with the background material presented in Sections \ref{AL_section}-\ref{FS_systems_sec}, which will be drawn upon frequently without hesitation or elaboration. We also advise readers to compare and contrast our discussion with those of \citep{MacKay_2004} and \citep{Gorban_2004}, which comprise reviews of slow manifold theory from a mathematical perspective and from a kinetic theory perspective, respectively. 

Because slow manifolds arise frequently in plasma physics that bear no relation to kinetic-to-fluid reductions the emphasis of our presentation will be broader than that in \citep{Gorban_2004}. Because most plasma physicists lack good training in differential topology our presentation will also be more pedestrian than that contained in \citep{MacKay_2004}. (We will, however, try to keep the geometry of slow manifold reduction in focus.) Moreover, because we believe computational plasma physics will benefit substantially from embracing ideas from slow manifold reduction, our discussion will highlight a perspective on slow manifolds that is especially compatible with numerical investigations.

Slow manifold reduction theory amounts to a rigorous, systematic, and geometric approach to the common notion of closure. Because of its associated level of rigor slow manifold reduction distinguishes itself from closure theories based on uncontrolled approximations like the direct interaction approximation due to Kraichnan \citep{Kraichnan_1959} or the far-reaching generalization thereof due to Martin, Siggia, and Rose.\,\citep{Martin_1973} As mentioned in the Introduction, the theory aligns instead more closely with the Chapman-Enskog formalism, which is fundamentally based on an asymptotic separation of timescales controlled by some small parameter $\epsilon$. While slow manifolds are not tied up directly with the presence of multiple spatial scales they are by no means incompatible with spacetime scale separation. We remark that a spacetime covariant approach to slow manifold reduction seems to be missing from the literature at the present time, even though there appears to be no essential impediment to developing such a theory. 

\subsection{Origins of slow manifolds\label{origins_sec}}

The simplest and most direct way to describe the theory of slow manifolds is to work within the context of fast-slow systems developed in Section \ref{FS_systems_sec}. Therefore suppose that we have been given a plasma model that is formulated as a fast-slow system as in Definition \ref{FS_def}. Such a system generally exhibits a pair of disparate timescales, a short $O(\epsilon)$ scale owing to the fact that generically $\dot{y} = O(1/\epsilon)$, and a longer $O(1)$ scale that characterizes trajectory segments passing within the region where $f_\epsilon(x,y) = O(\epsilon)$. Due to this intermingling of fast and slow dynamics a natural question to ask is whether there are special solutions where the fast timescale is completely inactive, or else only excited with a very small amplitude. We will refer to such solutions as slow solutions.

A suggestive mechanical analogy to keep in mind throughout this discussion is a pendulum whose massive end (with mass $m_1$) is attached to a second mass ($m_2$) by a light, stiff spring. While generic motions of such a system involve both swinging of the pendulum and relatively rapid oscillations of the stiff spring, physical experience suggests that there should be system motions during which the stiff spring is not excited, so that the overall dynamics resembles that of an ordinary pendulum with mass $M = m_1+ m_2$. We leave it as an exercise for the reader to verify that this two-mass system may be formulated as a fast-slow system, where the small parameter $\epsilon$ is proportional to the ratio of the pendulum frequency to the spring frequency.

Perhaps the most remarkable feature of fast-slow systems is that slow solutions of such systems organize themselves in a geometrically simple manner. Due to the typical complexity of the phase portrait for a dynamical system one might instead expect that such solutions tend to disperse themselves around phase space in a haphazard fashion. But due to the special structural properties of fast-slow systems these solutions in fact lie along special submanifolds in phase space called \emph{slow manifolds} that are readily computable using asymptotic methods. We may bring these slow manifolds to light as follows.

First observe, as we did earlier when motivating the definition of fast-slow systems, that if $(x_\epsilon(t),y_\epsilon(t))$ is a slow solution of a fast-slow system for each $\epsilon$ then the limiting solution $(x_0(t),y_0(t))$ must satisfy Eqs.\,\eqref{limit_FS_one}-\eqref{limit_FS_two}. In particular the limiting fast variable $y_0(t)$ must be \emph{slaved} to the limiting slow variable $x_0(t)$ according to $y_0(t) = y_0^*(x_0(t))$, where $y_0^*(x)$ is implicitly defined by the formula $f_0(x,y_0^*(x)) = 0$. (Recall that the definition of fast-slow systems ensures this equation uniquely determines $y_0^*$, at least locally.) Let us refer to the submanifold $S_0$ given by the graph of $y_0^*$, i.e.
\begin{align}
S_0 = \{(x,y)\mid y = y_0^*(x)\},
\end{align}
as the \emph{limiting slow manifold}.

Apparently each limiting slow trajectory must be contained in the limiting slow manifold. In fact the limiting slow manifold is the union of all the limiting slow trajectories. If we suppose that slow trajectories exist for finite $\epsilon$ this observation leads to the  seemingly-reasonable hypothesis that the qualitative organization of slow trajectories when $\epsilon =0$ persists when $\epsilon$ is small but finite. In particular it suggests that the collection of finite-$\epsilon$ slow trajectories form a small deformation of the limiting slow manifold. This deformed submanifold $S_\epsilon$ must have the form
\begin{align}
S_\epsilon = \{(x,y)\mid y = y_\epsilon^*(x)\},
\end{align}
where $y_\epsilon^*(x)$ is some function that tends to $y_0^*(x)$ as $\epsilon\rightarrow 0$. Moreover, because we expect $S_\epsilon$ to be the union of finite-$\epsilon$ slow trajectories it is reasonable to suspect that $S_\epsilon$ is an invariant manifold for the fast-slow system. (See Section \ref{IM_section} for background on invariant manifolds.) Indeed, if $S_\epsilon$ were not an invariant manifold then there would be slow solutions that start in $S_\epsilon$ and then eventually leave, contradicting the hypothesis that $S_\epsilon$ contains entire slow trajectories. 

Let us now test our hypothesis that there is a deformation $S_\epsilon$ of the limiting slow manifold $S_0$ that (a) contains finite-$\epsilon$ slow trajectories, and (b) is invariant under the flow of the (finite-$\epsilon$) fast-slow system. By Proposition \ref{PDE_char} from Section \ref{IM_section} with $f(x,y) = f_\epsilon(x,y)/\epsilon$ and $g(x,y) = g_\epsilon(x,y)$, invariance of $S_\epsilon$ implies that the slaving function $y_\epsilon^*(x)$ must satisfy the (scaled) invariance equation
\begin{align}
\epsilon\,Dy_\epsilon^*(x)[g_\epsilon(x,y_\epsilon^*(x))] = f_\epsilon(x,y_\epsilon^*(x)).\label{scaled_invariance_eq}
\end{align}
If there is a $y_\epsilon^*$ that solves this equation then the graph of $y_\epsilon^*$ will necessarily be an invariant manifold. Remarkably, solutions contained in such an invariant manifold are \emph{automatically} free of the $O(\epsilon)$ timescale because for such solutions (a) the timescale of the fast variable $y(t) = y_\epsilon^*(x(t))$ is determined by $x(t)$, and (b) $x(t)$ is a solution of the equation $\dot{x} = g_\epsilon(x,y^*_\epsilon(x)) = O(1)$. The task of testing our hypothesis therefore reduces to showing that the invariance equation \eqref{scaled_invariance_eq} admits a solution that is asymptotic to $y_0^*$.

Proving that \eqref{scaled_invariance_eq} admits a solution is generally a highly-nontrivial task because first-order systems of partial differential equations usually cannot be solved by hand, and sometimes do not admit solutions at all. Let us therefore attempt to develop a preliminary understanding of solutions of \eqref{scaled_invariance_eq} using asymptotic expansions. Specifically let us suppose $y_\epsilon^*(x)$ has the asymptotic expansion
\begin{align}
y^*_\epsilon(x) = y_0^*(x) + \epsilon\,y_1^*(x) + \epsilon^2\,y_2^*(x) + \dots,\label{formal_SM_ansatz_proto}
\end{align}
and investigate conditions imposed on the coefficient functions $y_k^*$ by the invariance equation \eqref{scaled_invariance_eq}.

Substituting the ansatz \eqref{formal_SM_ansatz_proto} into \eqref{scaled_invariance_eq} and then collecting $O(1)$ terms leads to
\begin{align}
0 = f_0(x,y_0^*(x)),
\end{align}
which is consistent with our assumption that $S_\epsilon$ is a small deformation of the limiting slow manifold $S_0$. Collecting the $O(\epsilon)$ terms leads to 
\begin{align}
Dy_0^*[g_0(x,y_0^*)] = D_yf_0(x,y_0^*)[y_1^*]+f_1(x,y_0^*),
\end{align}
which, by the definition of fast-slow systems, can be used to solve for $y_1^*$ in terms of quantities already computed at zero'th order. In particular, 
\begin{align}
y_1^* = [D_yf_0(x,y_0^*)]^{-1}\bigg(Dy_0^*[g_0(x,y_0^*)]-f_1(x,y_0^*)\bigg).
\end{align}
Going to yet-higher orders the following pattern emerges. Within the collection of $O(\epsilon^k)$ terms generated by substituting \eqref{formal_SM_ansatz_proto} into \eqref{scaled_invariance_eq} the only term involving the coefficient $y_k^*$ is $D_yf_0(x,y_0^*)[y_k^*]$. All other terms in the collection involve only the coefficients $y_l^*$ with $l<k$. It follows that the coefficients $y_k^*$ may be explicitly computed recursively for all $0 \leq k <\infty$.

The conclusion that we draw from this asymptotic analysis is that if a solution $y_\epsilon^*$ of the invariance equation \eqref{scaled_invariance_eq} exists and is smooth in $\epsilon$ then that solution has a unique asymptotic expansion in terms of the $y_k^*$ that may be computed recursively for all $k$. Moreover, even if a true solution does not exist, the asymptotic expansion of such a solution \emph{always exists}. We summarize this curious result by introducing the notion of a formal slow manifold,

\begin{definition}[formal slow manifold]\label{fSM_def}
Given a fast-slow system $\epsilon\,\dot{y} = f_\epsilon(x,y)$, $\dot{x} = g_\epsilon(x,y)$, a \emph{formal slow manifold} is a formal power series solution $y_\epsilon^*$ of the invariance equation
\begin{align}
\epsilon\,Dy_\epsilon^*(x)[g_\epsilon(x,y_\epsilon^*(x))] = f_\epsilon(x,y_\epsilon^*(x)).\label{scaled_invariance_eq_def}
\end{align}
\end{definition}
\noindent We then state the following
\begin{theorem}[Existence and uniqueness of formal slow manifolds]\label{SM_eau}
Associated with each fast-slow system is a unique formal slow manifold $y_\epsilon^* = y_0^* + \epsilon\,y_1^* + \epsilon^2\,y_2^* + \dots$. The coefficients $y_k^*$ may be explicitly computed recursively. In particular we have the following low-order formulas:
\begin{align}
0 = & f_0(x,y_0^*(x))\label{limit_sm_eqn}\\
y_1^* =& [D_yf_0(x,y_0^*)]^{-1}\bigg(Dy_0^*[g_0(x,y_0^*)]-f_1(x,y_0^*)\bigg)\label{first_order_sm_eqn}
\end{align}
\end{theorem}

\begin{example}\label{weak_drag_SM_example}
Theorem \ref{SM_eau} may be applied directly to compute the formal slow manifold associated with Abraham-Lorentz dynamics in the weak drag regime. (c.f. Section \ref{weak_drag_sec}.) As explained in Example \ref{FS_example_easy} from Section \ref{what_is_sec}, when this model is expressed in terms of the fast time $\tau$ it becomes a fast-slow system with $x = (\bm{x},\bm{v})$, $y = \bm{a}$, and 
\begin{align}
f_0(x,y) =& \frac{3}{2}\bm{a} - \frac{3}{2}\zeta\,\bm{v}\times\bm{B}(\bm{x})\\
g_0(x,y) = & (0,\bm{a})\label{weak_drag_g0}\\
g_1(x,y) = &(\bm{v},0).\label{weak_drag_g1}
\end{align}
All higher-order coefficients in the power series expansions of $f_\epsilon$ and $g_\epsilon$ are zero. Equation \eqref{limit_sm_eqn} in the statement of Theorem \ref{SM_eau} therefore implies that the limiting slow manifold is given as the graph of $y_0^* = \bm{a}_0^*$, where
\begin{align}
\bm{a}_0^*(\bm{x},\bm{v}) = \zeta\,\bm{v}\times\bm{B}(\bm{x}).\label{weak_drag_y0}
\end{align}
In order to compute the first-order correction to the limiting slow manifold from \eqref{first_order_sm_eqn} we first note that the derivatives $D_yf_0$ and $Dy_0^*$ are given by
\begin{align}
D_yf_0[\delta y ] = & \frac{3}{2}\delta\bm{a}\\
Dy_0^*(x)[\delta x] = & \zeta\,\delta\bm{v}\times\bm{B} + \zeta\,\bm{v}\times(\delta\bm{x}\cdot\nabla\bm{B}).
\end{align}
Therefore we have
\begin{align}
Dy_0^*[g_0(x,y_0^*)]-f_1(x,y_0^*) =& \zeta\,\bm{a}_0^*\times\bm{B}\nonumber\\
=& (\bm{v}\times\bm{B})\times\bm{B},
\end{align}
and $y_1^* = \bm{a}_1^*$ with
\begin{align}
\bm{a}_1^* = \frac{2}{3} (\bm{v}\times\bm{B})\times\bm{B}.\label{weak_drag_y1}
\end{align}
In summary, the formal slow manifold for Abraham-Lorentz dynamics in the weak-drag regime is given by $y_\epsilon^* = \bm{a}_\epsilon^*$, where
\begin{align}
\bm{a}_\epsilon^* = \zeta\,\bm{v}\times\bm{B} +\epsilon \,\frac{2}{3}(\bm{v}\times\bm{B})\times\bm{B}+O(\epsilon^2).
\end{align}
The first term in this series is nothing more than the usual Lorentz force on a charged particle. The second term is the small-velocity form of the so-called Landau-Lifshitz form of the radiation drag force. We will connect this result to the work of Spohn after introducing the notion of slow manifold reduction in the next subsection.
\end{example}

\begin{example}\label{zero_drag_sm_example}
While Theorem \ref{SM_eau} cannot be applied directly to Abraham-Lorentz dynamics in the zero-drag regime, it can be applied after changing dependent variables from $(\bm{x},\bm{v})$ to $(\bm{x},v_\parallel,v_1,v_2)$ as in Example \ref{FS_example_hard}; recall that $x = (\bm{x},v_\parallel)$ and $y = (v_1,v_2)$ comprise a fast-slow split for this system. As for the ingredients required to apply Theorem \ref{SM_eau}, we have $f_\epsilon = (f_\epsilon^{v_1},f_\epsilon^{v_2})$, $g_\epsilon = (g_\epsilon^{\bm{x}},g_\epsilon^{v_\parallel})$, 
\begin{align}
f_0^{v_1}(x,y) =&\hspace*{1.2em}\zeta |\bm{B}| \,v_2 \\
f_0^{v_2}(x,y) = &-\zeta |\bm{B}|\,v_1\\
g_0^{v_\parallel}(x,y) = & \bm{v}\cdot\nabla\bm{b}\cdot\bm{v}\\
g_0^{\bm{x}}(x,y) = & \bm{v}
\end{align}
at $0^{\text{th}}$ order, and
\begin{align}
f_1^{v_1}(x,y) =&\hspace*{1.2em}\bm{v}\cdot\bm{R}\,v_2 - v_\parallel \bm{v}\cdot\nabla\bm{b}\cdot \bm{e}_1 \\
f_1^{v_2}(x,y) = &-\bm{v}\cdot \bm{R}\,v_1 - v_\parallel \bm{v}\cdot\nabla\bm{b}\cdot\bm{e}_2\\
g_1^{v_\parallel}(x,y) = & 0\\
g_1^{\bm{x}}(x,y) = & 0
\end{align}
at first order. The formal slow manifold associated with this system is therefore of the form $y_\epsilon^* = ((v_{1})^*_0,(v_{2})^*_0)$, i.e. the perpendicular velocity variables are slaved to the particle position $\bm{x}$ and the parallel velocity $v_\parallel$. Because $f_0(x,y) =0$ if and only if $v_1 = v_2 = 0$ the limiting slaving function $y_0^*$ vanishes:
\begin{align}
\begin{pmatrix}
(v_{1})_0^*\\
(v_{2})_0^*
\end{pmatrix} = 
\begin{pmatrix}
0\\
0
\end{pmatrix} .\label{zero_drag_y0}
\end{align}
It follows that the general formula \eqref{first_order_sm_eqn} for the first-order slaving function simplifies to
\begin{align}
y_1^* = &-[D_yf_0(x,y_0^*)]^{-1}\bigg(f_1(x,y_0^*)\bigg).\label{simplifed_y1}
\end{align}
And since the derivative matrix $D_yf_0(x,y_0^*)$ is given by
\begin{align}
D_yf_0(x,y_0^*)[\delta y] =
\begin{pmatrix}
0 & \zeta\,|\bm{B}|\\
-\zeta\,|\bm{B}| & 0
\end{pmatrix} 
\begin{pmatrix}
\delta v_1\\
\delta v_2
\end{pmatrix},
\end{align}
equation \eqref{simplifed_y1} implies
\begin{align}
\begin{pmatrix}
(v_{1})_1^*\\
(v_{2})_1^*
\end{pmatrix}=
\frac{1}{\zeta\,|\bm{B}|} \begin{pmatrix}
0 & -1\\
1 & 0
\end{pmatrix} \begin{pmatrix}
v_\parallel^2\,\bm{b}\cdot\nabla\bm{b}\cdot \bm{e}_1\\
v_\parallel^2\,\bm{b}\cdot\nabla\bm{b}\cdot\bm{e}_2
\end{pmatrix}.\label{curv_drift_slave}
\end{align}
Equation \eqref{curv_drift_slave} says that the first-order correction to the zero-drag formal slow manifold is given by the so-called curvature drift velocity from guiding center theory. This connection to guiding center theory may be emphasized further by introducing the notation 
\begin{align}
\bm{v}_{\perp\epsilon}^* = (v_1)^*_\epsilon\, \bm{e}_1 + (v_2)^*_\epsilon\,\bm{e}_2,\label{perp_shorthand}
\end{align}
in terms of which \eqref{curv_drift_slave} becomes
\begin{align}
\bm{v}_{\perp 1 }^* =- \frac{v_\parallel^2\,(\bm{b}\cdot\nabla\bm{b})\times\bm{b}}{\zeta\,|\bm{B}|}\equiv \bm{v}_c,\label{zero_drag_y1}
\end{align}
which is the familiar expression for the curvature drift velocity. %We will discuss the interpretation of this result after introducing the notion of formal slow manifold reduction.

We leave it as an exercise for the reader to verify that the second-order slaving function $\bm{v}_{\perp 2}^*$ is given by
\begin{align}
\bm{v}_{\perp 2}^* = -\frac{v_\parallel ( \bm{b}\cdot \nabla \bm{v}_c + \bm{v}_c\cdot \nabla\bm{b} )\times\bm{b}}{\zeta\,|\bm{B}|}.\label{zero_drag_y2}
\end{align}

\end{example}

The argument used above to deduce Theorem \ref{SM_eau} actually says nothing at all about the existence of a solution $y_\epsilon^*$ of the invariance equation. We therefore cannot infer from that argument that finite-$\epsilon$ slow solutions of a fast-slow system comprise an invariant manifold, or even that finite-$\epsilon$ slow solutions exist! In fact the formal power series $y_\epsilon^*$ may easily fail to be convergent \citep{Lorenz_1987,Vanneste_2004,MacKay_2004}. What then, if anything, does a formal slow manifold have to tell us about solutions of a fast-slow system?

While this important question will be taken up in greater detail in Section \ref{interpretation_SM} we may preempt the more complete discussion with the following heuristic picture. Consider the finite truncation 
\begin{align}
y_\epsilon^{*N} = y_0^* + \epsilon\,y_1^* +\dots + \epsilon^N\,y_N^*
\end{align}
of a formal slow manifold, where $N$ is some fixed non-negative integer. While the function $y_\epsilon^{*N}$ does not solve the invariance equation \eqref{scaled_invariance_eq_def} exactly, it does solve it in an approximate sense: when $y_\epsilon^{*N}$ is substituted into Eq.\,\eqref{scaled_invariance_eq_def} the difference between the left- and right-hand sides is $O(\epsilon^{N+1})$. Because $N$ is arbitrary, we say that formal slow manifolds comprise invariant manifolds to all orders in perturbation theory. Referring then to the interpretation of invariant manifolds in terms of tangency given in Proposition \ref{tangency_prop} the angle between the vector field $U_\epsilon = (g_\epsilon,f_\epsilon/\epsilon)$ and the graph of the truncation $y_\epsilon^{*N}$ becomes arbitrarily small as $N$ increases. In this sense finite truncations of formal slow manifolds are \emph{almost} invariant sets. (See for example \citep{Kristiansen_2016} for a precise statement of the sense in which a particular class of slow manifolds comprises almost invariant sets.)

Where true invariant sets contain solutions for all time, almost invariant sets generally only keep solutions nearby over some finite time interval. The term ``sticky set" is an appropriate descriptor for these objects. This ``sticking" time interval is generally $O(1)$ as $\epsilon$ tends to zero for truncated formal slow manifolds that do not exhibit normal instabilities. While such an interval certainly significantly exceeds the short $O(\epsilon)$ timescale, its length often cannot be increased by increasing $N$ or decreasing $\epsilon$. However, under certain special circumstances, some of which will be explained in Section \ref{interpretation_SM}, the sticking time may be much longer, or even infinite. Whatever the case, truncated formal slow manifolds certainly play an important dynamical role on the $O(1)$ timescale. We are therefore motivated to set aside special terminology for truncations of formal slow manifolds.

\begin{definition}[slow manifold]\label{sm_defined_properly}
Given a fast-slow system with formal slow manifold $y_\epsilon^* = y_0^* + \epsilon\,y_1^*+\dots$, a \emph{slow manifold} of order $N$ is a submanifold $\widetilde{S}_\epsilon$ in $(x,y)$-space of the form
\begin{align}
\widetilde{S}_\epsilon = \{(x,y)\mid y = \widetilde{y}_\epsilon(x)\},
\end{align}
where $\widetilde{y}_\epsilon(x)$ depends smoothly on $x$ and $\epsilon$ and
\begin{align}
y_\epsilon^* - \widetilde{y}_\epsilon = O(\epsilon^{N+1}),
\end{align}
where the difference is understood in the sense of formal power series.
\end{definition}

\begin{remark}
This definition should be compared with MacKay's definition of slow manifolds from \citep{MacKay_2004}.
\end{remark}

It is enlightening to compare the qualitative description of slow manifolds just presented with qualitative features of the classical averaging theory due to Kruskal.\,\citep{Kruskal_1962} Recall that this theory forms the mathematical basis for the guiding center description of charged particle motion in strong magnetic fields, amongst other reduced models. Like fast-slow systems theory, Kruskal's theory deals with dynamical systems that exhibit a short $O(\epsilon)$ timescale and a longer $O(1)$ timescale. In contrast to the theory of slow manifolds, however, Kruskal's formalism aims to describe system motions during which both timescales are active simultaneously. In order to make progress in such an endeavor Kruskal assumes that the limiting short-timescale dynamics comprise strictly periodic motion and then proceeds to ``average" over those rapid periodic oscillations. The averaged equations produced by Kruskal's method generally have a validity time that is $O(1)$ as $\epsilon$ tends to zero. This is true in spite of the fact that the guiding center equations of motion are often tacitly assumed to remain valid over arbitrarily-long time intervals. We will argue in Section \ref{interpretation_SM} that understanding the link between many popular reduced models in plasma physics and slow manifold theory exposes an uncomfortable truth: the fallacy of overestimating the validity time of reduced models occurs frequently within plasma physics, even outside of discussions of guiding center theory. 

\subsection{Slow manifold reduction\label{reduction_sec}}
Having established that slow manifolds possess a certain stickiness, we now turn to describing approximately the slow dynamics of solutions that are stuck to a slow manifold. This discussion will lead to the notion of \emph{slow manifold reduction}.

Recall from Section \ref{IM_section} that if a dynamical system $\dot{x} = g(x,y)$, $\dot{y} = f(x,y)$ admits an invariant manifold given by the graph of a function $y^*(x)$ then for any solution $(x(t),y(t))$ contained in the invariant manifold Eq.\,\eqref{reduced_dynamics_gen} must be satisfied by $x(t)$. That is, the dynamics of the variable $x$ close on themselves when initial conditions for $(x,y)$ are chosen to lie on the invariant manifold. In this scenario $y^*(x)$ may be interpreted as a closure function. Because formal slow manifolds may be thought of as invariant manifolds to all orders in perturbation theory it is therefore natural to expect that the equation
\begin{align}
\dot{x} = g_\epsilon(x,y_\epsilon^*(x)),\label{slow_reduced_proto}
\end{align}
describes the dynamics of solutions that are stuck to a slow manifold of any order. This point is driven home by the fact that when $y_\epsilon^*$ is a genuine solution of the invariance equation (meaning $y_\epsilon^*$ is a slow manifold of order $\infty$) Eq.\,\eqref{slow_reduced_proto} is exactly satisfied for solutions initialized on the graph of $y_\epsilon^*$. 

We will return to the issue of describing the precise sense in which Eq.\,\eqref{slow_reduced_proto} governs the dynamics of slow trajectories in Section \ref{interpretation_SM}. Nevertheless, the rough picture of the dynamical role played by \eqref{slow_reduced_proto} is worth describing here. While Eq.\,\eqref{slow_reduced_proto} is not technically an ordinary differential equation as it is written (the right-hand-side is not a well-defined function of $x$ because $y_\epsilon^*$ is defined only as a formal power series), we can extract a genuine ordinary differential equation by expanding the right-hand-side as a formal power series in $\epsilon$ and then truncating at some order $N$. Due to the truncation this ordinary differential equation will not be satisfied exactly by any particular solution of the original fast-slow system. However on any of the $O(1)$ time intervals over which a trajectory sticks to a slow manifold of order $N$ the difference between (a) a solution of Eq.\,\eqref{slow_reduced_proto} truncated at $N^{\text{th}}$ order and (b) the $x$-component of a solution of the original fast-slow system initialized within $O(\epsilon^N)$ of the slow manifold can be made arbitrarily small by increasing $N$.

The preceding interpretation of Eq.\,\eqref{slow_reduced_proto} may be summarized by the assertion that \eqref{slow_reduced_proto} describes the dynamics of solutions initialized on a formal slow manifold to all orders in perturbation theory. Any statement of properties possessed by solutions stuck to a slow manifold that is valid to all orders in $\epsilon$ may therefore be conveniently formulated in terms of the notion of a fast-slow system's \emph{formal slow manifold reduction}.

\begin{definition}[formal slow manifold reduction]\label{formal_SMR_definition}
Given a fast-slow system $\dot{x} = g_\epsilon(x,y)$, $\epsilon\,\dot{y} = f_\epsilon(x,y)$ with formal slow manifold $y_\epsilon^*$ the system's \emph{formal slow manifold reduction} is the formal ODE on $x$-space given by
\begin{align}
\dot{x} = g_\epsilon^*(x),
\end{align}
where
\begin{align}
g_\epsilon^*(x)= g_\epsilon(x,y_\epsilon^*(x)),\label{formal_SMR_def}
\end{align}
and $g_\epsilon^*$ is interpreted as a formal power series in $\epsilon$.
\end{definition}

\noindent In contrast, properties of stuck solutions that pertain to a specific order of accuracy are more conveniently formulated in terms of a fast-slow system's \emph{slow manifold reduction} of order $N$.

\begin{definition}[slow manifold reduction]\label{smr_defined}
Given a fast-slow system $\dot{x} = g_\epsilon(x,y)$, $\epsilon\,\dot{y} = f_\epsilon(x,y)$ with formal slow manifold $y_\epsilon^*$, a \emph{slow manifold reduction} of order $N$ is any ODE $
\dot{x} = \widetilde{g}_\epsilon(x)$ with the property
\begin{align}
\widetilde{g}_\epsilon(x) - g_\epsilon(x,y_\epsilon^*(x)) = O(\epsilon^{N+1}).\label{SMR_def}
\end{align}
\end{definition}

\begin{example}
A first-order slow manifold reduction of Abraham-Lorentz dynamics in the weak-drag regime may be constructed as follows. As discussed in Example \ref{weak_drag_SM_example}, the $0^{\text{th}}$- and $1^{\text{st}}$-order terms in this system's formal slow manifold $y_\epsilon^*(x) = \bm{a}_\epsilon^*(\bm{x},\bm{v})$ are given by Eqs.\,\eqref{weak_drag_y0} and \eqref{weak_drag_y1}. This means that the system's formal slow manifold reduction $g_\epsilon^*(x)=(\dot{\bm{x}}_\epsilon^*,\dot{\bm{v}}_\epsilon^*)$, which is generally given by
\begin{align}
g_\epsilon^*(x) &= g_0(x,y_0^*)\nonumber\\
 &+ \epsilon\,(g_1(x,y_0^*) +D_yg_0(x,y_0^*)[y_1^*]) + O(\epsilon^2),
\end{align}
may be written explicitly
\begin{align}
\dot{\bm{v}}_\epsilon^* &= \zeta\,\bm{v}\times\bm{B} + \epsilon\, \frac{2}{3}(\bm{v}\times\bm{B})\times\bm{B}+O(\epsilon^2)\\
\dot{\bm{x}}_\epsilon^* & = \epsilon\,\bm{v} 
\end{align}
where we have used the definitions of $g_0(x,y)$ and $g_1(x,y)$ from Eqs.\,\eqref{weak_drag_g0} and \eqref{weak_drag_g1}. Now we may merely omit all of the terms in the formal slow manifold reduction that are $O(\epsilon^2)$ or higher to obtain a slow manifold reduction of order $N=1$:
\begin{align}
\dot{\widetilde{\bm{v}}}_\epsilon &= \zeta\,\bm{v}\times\bm{B} + \epsilon\, \frac{2}{3}(\bm{v}\times\bm{B})\times\bm{B}\label{LL_one}\\
\dot{\widetilde{\bm{x}}}_\epsilon & = \epsilon\,\bm{v}.\label{LL_two}
\end{align}
Notice that according to Definition \ref{smr_defined} the form of the $O(\epsilon^2)$ terms in any first-order slow manifold reduction may be modified at will without changing the order of the reduction. While we have here made the seemingly-obvious choice of setting those terms equal to zero, there are often good reasons to make other choices. (c.f. Section \ref{hamiltonian_SM}.)

The first-order slow manifold reduction \eqref{LL_one}-\eqref{LL_two} coincides with the small (but non-zero) velocity limit of Eq.\,(76.3) in \citep{Landau_fields_1975}. (See \citep{Vranic_2016} for a convenient laboratory-frame formula giving the Landau-Lifshitz form of the radiation drag force for arbitrary velocities.) It therefore represents Landau's and Lifshitz's proposal for the dynamical equations governing a single charged particle that is experiencing radiation drag. In contrast to the Abraham-Lorentz equations, which evolve in the nine-dimensional phase space of triples $(\bm{x},\bm{v},\bm{a})$, the Landau-Lifshitz equations evolve in the six-dimensional phase space of pairs $(\bm{x},\bm{v})$. We have explained this reduction in dimensionality using the notion of slow manifolds. This should be compared with the argument of Spohn,\,\cite{Spohn_2000} who speaks of a \emph{critical manifold} in the phase space for Abraham-Lorentz dynamics. (To be precise, Spohn deals with the more accurate, but qualitatively similar Lorentz-Abraham-Dirac equation.) Spohn's critical manifold is really the same thing as our slow manifold. This coincidence comes as no surprise because in \citep{Spohn_2000} Spohn explicitly makes the link between his argument and Fenichel's geometric singular perturbation theory, which we have mentioned forms one of the pillars of slow manifold theory. We will return again to Spohn's work on radiation drag when discussing the rigorous validity limits of slow manifold reduction in Section \ref{interpretation_SM}.

\end{example}

\begin{example}\label{second_order_smr_zero_drag}
Let us also examine a second-order slow manifold reduction of the zero-drag Abraham-Lorentz equations. Recall from Example \ref{FS_example_hard} that this system admits a fast-slow split, and that the first few terms of the associated formal slow manifold were computed in Example \ref{zero_drag_sm_example}. In terms of the notation introduced near the end of the latter example the formal slow manifold reduction for the zero-drag equations is given by
\begin{align}
\dot{v}_{\parallel\epsilon}^* &=(v_\parallel \bm{b} + \bm{v}_{\perp\epsilon}^*)\cdot\nabla\bm{b}\cdot  \bm{v}_{\perp\epsilon}^* \\
\dot{\bm{x}}_{\epsilon}^* &=v_\parallel\,\bm{b} +  \bm{v}_{\perp\epsilon}^*.
\end{align}
Using Eqs.\,\eqref{zero_drag_y0}, \eqref{zero_drag_y1}, and \eqref{zero_drag_y2} for the first three coefficients in $y_\epsilon^*$ the previous equations may be simplified and truncated at $O(\epsilon^3)$ to produce the following slow manifold reduction of order $N=2$:
\begin{align}
\dot{\widetilde{v}}_{\parallel\epsilon} &=  - \epsilon^2 \bm{b}\cdot\nabla\frac{1}{2}|\bm{v}_c|^2 \label{2nd_order_smr_1}\\
\dot{\widetilde{\bm{x}}}_{\epsilon}  &= v_\parallel\,\bm{b}  +\epsilon\bm{v}_c\nonumber\\
&-\epsilon^2\frac{v_\parallel ( \bm{b}\cdot \nabla \bm{v}_c + \bm{v}_c\cdot \nabla\bm{b} )\times\bm{b}}{\zeta\,|\bm{B}|}.\label{2nd_order_smr_2}
\end{align}
We remind the reader that the curvature drift velocity $\bm{v}_c$ was defined in Eq.\,\eqref{zero_drag_y1}.

Equations \eqref{2nd_order_smr_1}-\eqref{2nd_order_smr_2} bear a striking resemblance to the well-known guiding center equations of motion for a charged particle moving through a strong magnetic field. There are however notable differences: the $\nabla B$-drift is absent from Eq.\,\eqref{2nd_order_smr_2}, and while Eq.\,\eqref{2nd_order_smr_1} resembles the mirror force $-\mu \bm{b}\cdot\nabla|\bm{B}|$ in form, $\frac{1}{2}|\bm{v}_c|^2$ appears in place of the effective potential $\mu |\bm{B}|$. Perhaps even more curiously the magnetic moment $\mu$ does not appear at all! What then is the interpretation of these equations?

The key to answering this question is recognizing that relative to the asymptotic scaling used for the zero-drag equations the cyclotron period is $O(\epsilon)$. In keeping with our earlier qualitative description of the dynamical role played by slow manifolds, the zero-drag slow manifold reduction therefore cannot sample the cyclotron period. On the other hand the $O(1)$ timescale that characterizes motion on the slow manifold is also described correctly by the usual guiding center theory. The only way to reconcile these two observations is to conclude that \eqref{2nd_order_smr_1}-\eqref{2nd_order_smr_2} describe particle dynamics \emph{on the zero level set of the magnetic moment}. This perspective immediately explains the ``missing" terms in the second-order slow manifold reduction --- those terms are all proportional to $\mu$. The reconciliation also implies that the $O(\epsilon^2)$ terms in the slow manifold reduction must appear in the guiding center equations of motion expanded to a comparable order in $\epsilon$.
\end{example}

The passage from a fast-slow system to its slow manifold reduction happens frequently within plasma physics. In fact very few reduced models for plasma behavior based on controlled approximations do not arise in this manner. (The authors would be excited to learn of exceptions to this rule!) For example, the famous magnetohydrodynamic (MHD) model for strongly-magnetized plasmas arises as the lowest-order slow manifold reduction of the two-fluid-Maxwell model. In the non-dissipative case this slow manifold reduction was described in detail in Ref.\,\citep{Burby_two_fluid_2017}, where order $0$, $1$, and $2$ reductions were computed. The kinetic version of MHD known variously as kinetic MHD or the guiding center plasma model also arises via slow manifold reduction, this time applied to the Vlasov-Maxwell model. This observation was used to show for the first time that non-dissipative kinetic MHD comprises an infinite-dimensional Hamiltonian system in Ref.\,\citep{Burby_Sengupta_2017_pop}. Other examples abound, some of which are summarized in Table \ref{table1}.

\begin{table}[htb]
\caption{Some slow manifolds from plasma physics\label{table1}}
\begin{center}
\begin{tabular}{c|c}
\hline
\textbf{parent model} & \textbf{slow manifold}\\
\hline
Vlasov-Maxwell & Navier-Stokes-Maxwell\\ 
(unmagnetized, collisional) & \\
\hline
Vlasov-Maxwell & Braginskii\\
 (magnetized, collisional) & \\
 \hline
Vlasov-Maxwell & Kinetic MHD ($k_\perp\rho \ll 1$)\\
(magnetized, weakly-collisional) & Gyrokinetics ($k_\perp\rho \sim 1$)\\
\hline
Vlasov-Maxwell & Vlasov-Poisson ($0^{\text{th}}$ order)\\
(weakly-relativistic) & Vlasov-Darwin ($1^{\text{st}}$ order)\\
\hline
Two-fluid-Maxwell & MHD ($0^{\text{th}}$ order)\\
(magnetized) & Hall MHD ($1^{\text{st}}$ order)\\ 
\hline
Lorentz force dynamics & $\mu = 0$ guiding center\\
(magnetized) & \\
\hline
Lorentz loop dynamics & finite-$\mu$ guiding center\\
(magnetized) & \\
\end{tabular}
\end{center}
\end{table}

A simple corollary of the prevalence of slow manifold reductions in plasma physics is that any valid result for general slow manifold reductions applies immediately to a broad class of reduced plasma models. For instance, if a reduced model may be shown to arise as the $0^{\text{th}}$-order slow manifold reduction of a fast-slow system then a simple consequence of Theorem \ref{SM_eau} is that the (formal) accuracy of such a model may be increased by constructing higher-order slow manifold reductions. Such higher-order reduced models may be deduced systematically, and allow inclusion of important physical effects absent from the leading-order reduction. The ideal MHD model for instance obeys the frozen-in Law, whereas the higher-order slow manifold reductions in Ref.\,\onlinecite{Burby_two_fluid_2017} incorporate (non-dissipative, fluid-based) non-ideal effects responsible for magnetic reconnection. Outside of MHD interesting outstanding applications of this observation include deriving weakly-relativistic reduced kinetic models that improve on the Vlasov-Darwin model, incorporating the effects of harmonic generation in wave-mean-flow models such as those discussed in Refs.\,\citep{Dewar_1970,Gjaja_Holm_1996,Burby_Ruiz_2019}, and computing the effects of small deviations from neutrality in nominally quasineutral plasmas. (The last application will be discussed in Section \ref{QN_application}.)

% if needed, discuss high-order terms and generalizations of Burnett here

%\begin{example}\label{zero_drag_reduction_high_order_example}
%test
%\end{example}

A second general principle obeyed by a fast-slow system's slow manifold reductions, and therefore by many reduced models in plasma physics, concerns conservation laws. If a fast-slow system obeys a conservation law 
\begin{align}
\frac{d}{dt}Q(x(t),y(t)) = 0,
\end{align}
then the formal slow manifold reduction of that system, $\dot{x} = g_\epsilon(x,y_\epsilon^*(x))$, \emph{inherits} the formal conservation law
\begin{align}
\frac{d}{dt}Q^*_\epsilon(x(t)) = \frac{d}{dt}Q(x(t),y_\epsilon^*(x(t))) = 0.
\end{align}
In particular when $\epsilon=0$ we obtain the conservation law
\begin{align}
\frac{d}{dt}Q^*_0(x(t)) = \frac{d}{dt}Q(x(t),y_0^*(x(t))) = 0
\end{align}
for the limiting slow manifold reduction $\dot{x}=\widetilde{g}_0(x) = g_0(x,y_0^*(x))$. The precise meaning of this inheritance result is described in the following

\begin{theorem}[inheritance of first integrals]\label{scalar_inheritance_thm}
Let $\dot{x} = g_\epsilon(x,y)$, $\epsilon\,\dot{y} = f_\epsilon(x,y)$ be a fast slow system with formal slow manifold $y_\epsilon^*$. If this systems possesses an $\epsilon$-dependent scalar conserved quantity $Q_\epsilon(x,y)$ that is smooth in $\epsilon$ and $(x,y)$ then the quantity 
\begin{align}
Q_\epsilon^*(x) = Q_\epsilon(x,y_\epsilon^*(x)),\label{inherited_scalar}
\end{align}
interpreted as a formal power series in $\epsilon$, is a conserved quantity for the system's formal slow manifold reduction $\dot{x} = g_\epsilon^*(x)$. In other words,
\begin{align}
DQ_\epsilon^*(x)[g_\epsilon^*(x)] = 0,\label{inherited_conservation_of_scalar}
\end{align}
where equality is understood in the sense of formal power series. In particular the lowest-order consequences of \eqref{inherited_conservation_of_scalar} are
\begin{align}
0&=DQ_0^*(x)[g_0^*(x)]\label{limit_cons_thm}\\
0&=DQ_1^*(x)[g_0^*(x)] + DQ_0^*(x)[g_1^*(x)]\\
0&=DQ_2^*(x)[g_0^*(x)] + DQ_1^*(x)[g_1^*(x)] + DQ_0^*(x)[g_2^*(x)]\label{cons_break_thm}
\end{align}
\end{theorem}

\begin{remark}
If a fast-slow systems possesses structural properties other than conservation laws those properties tend to be inherited by the system's formal slow manifold reduction as well. For instance in Section \ref{hamiltonian_SM} we will demonstrate that slow manifold reductions of Hamiltonian fast-slow systems may be chosen to be Hamiltonian systems in their own right.
\end{remark}

\begin{proof}
This is a simple corollary of two fundamental properties of composition of functions that are formal power series in $\epsilon$. First, if $q(\epsilon)$ is a smooth function of $\epsilon$ and we define $f_\epsilon(x) = q(\epsilon)$ as a function of the variable $x$, then $f(x^*_\epsilon) = q(\epsilon)$ whenever $x_\epsilon^* = x_0 + \epsilon\,x_1+\epsilon^2\,x_2 + \dots$ is an $x$-valued formal power series. Second, if $F_\epsilon:X\rightarrow Y$ and $G_\epsilon: Y\rightarrow Z$ are smooth formal power series maps between vector spaces then the chain rule
\begin{align}
D(G_\epsilon\circ F_\epsilon)(x) = DG_\epsilon(F_\epsilon(x))\circ DF_\epsilon(x)\label{formal_chain_rule}
\end{align}
is satisfied to all orders in $\epsilon$.

We apply these simple properties as follows. First note that because $Q_\epsilon(x,y)$ is a conserved quantity $DQ_\epsilon(x,y)[(g_\epsilon(x,y),f_\epsilon(x,y)/\epsilon)] = 0$
%\begin{align}
%DQ_\epsilon(x,y)[(g_\epsilon(x,y),f_\epsilon(x,y)/\epsilon)] = 0
%\end{align}
for all $(x,y)$. By the first property above we therefore have the power series identity $DQ_\epsilon(x,y^*_\epsilon)[(g_\epsilon(x,y^*_\epsilon),f_\epsilon(x,y^*_\epsilon)/\epsilon)] = 0$.
%\begin{align}
%DQ_\epsilon(x,y^*_\epsilon)[(g_\epsilon(x,y^*_\epsilon),f_\epsilon(x,y^*_\epsilon)/\epsilon)] = 0,
%\end{align}
Using the definition of the formal slow manifold the last identity may also be written as
\begin{align}
DQ_\epsilon(x,y^*_\epsilon)[(g_\epsilon(x,y^*_\epsilon),Dy_\epsilon^*(x)[g_\epsilon(x,y^*_\epsilon)])] = 0.\label{proto_cons_prf}
\end{align}
But by the chain rule \eqref{formal_chain_rule} for formal power series,
\begin{align}
D{Q}_\epsilon^*(x)[\delta x] = DQ_\epsilon(x,y_\epsilon^*)[(\delta x,Dy_\epsilon^*(x)[\delta x])]
\end{align}
for all $\delta x$, in particular for $\delta x = g_\epsilon(x,y^*_\epsilon)$. Therefore Eq.\,\eqref{proto_cons_prf} implies the desired result.
\end{proof}

\noindent The limiting conservation law \eqref{limit_cons_thm} explains a number of interesting conservation laws across a variety of reduced plasma models. When barotropic electron and ion fluids are coupled to an electromagnetic field \emph{via} Maxwell's equations particle relabeling symmetry implies the presence of a circulation invariant for each fluid species of the form
\begin{align}
\frac{d}{dt}\oint_{C(t)} (q_\sigma \bm{A} +\epsilon\, m_\sigma \bm{u}_\sigma)\cdot d\bm{x} = 0,\label{circulation_invariant_tfm}
\end{align}
where the integrand may be interpreted as the canonical momentum carried by species $\sigma$, and the closed curve $C(t)$ is dragged along by the flow of species $\sigma$. The dimensionless parameter $\epsilon$ appears naturally when formulating the two-fluid-Maxwell equations as a fast-slow system as in Ref.\,\citep{Burby_two_fluid_2017}. The $0^{\text{th}}$-order slow manifold for the two-fluid-Maxwell system is (barotropic) ideal magnetohydrodynamics. Therefore Theorem \ref{scalar_inheritance_thm} implies the $\epsilon\rightarrow 0$ limit of the circulation law \eqref{circulation_invariant_tfm},
\begin{align}
q_\sigma\frac{d}{dt}\oint_{C(t)}\bm{A}\cdot d\bm{x} = 0,\label{limit_circulation_invariant}
\end{align} 
should be a conservation law for ideal MHD. But Eq.\,\eqref{limit_circulation_invariant} is (up to an unimportant factor of $q_\sigma$) just an integral statement of the well-known frozen-in law for ideal magnetic fields. The same limiting conservation law \eqref{limit_cons_thm} explains quadratic free-energy conservation in gyrokinetics \citep{Schekochihin_2009}, the phase-space circulation invariant for kinetic MHD discovered in \citep{Burby_Sengupta_2017_pop}, and the non-obvious form of the energy flux in the Vlasov-Poisson system's local energy conservation law. \citep{Similon_1981,Qin_weak_FT_2014}
 
Another general reduced model phenomenon highlighted by Theorem \ref{scalar_inheritance_thm} pertains to \emph{breakdown} of conservation laws. One might hope naively that since $Q_0^*$ is a conserved quantity for the $0^{\text{th}}$-order slow manifold reduction of a fast slow system the ostensibly more-accurate truncation
\begin{align}
Q_\epsilon^{*N}(x) = Q^*_0(x) + \epsilon\,Q_1^*(x) + \dots +\epsilon^N\,Q_N^*(x)
\end{align}
would be conserved by an $N^{\text{th}}$ order slow manifold reduction. However Eq.\,\eqref{cons_break_thm} may be used to show that this is not necessarily true. Suppose we set $N=1$ and use the simplest first-order truncation of $g_\epsilon^*$,
\begin{align}
\widetilde{g}_\epsilon(x)& = g_0(x,y_0^*)\nonumber\\
& + \epsilon\bigg(g_1(x,y_0^*) + D_yg_0(x,y_0^*)[y_1^*]\bigg),\label{test_first_order_reduction}
\end{align}
to define a slow manifold reduction of order $1$. Note that since $g_1^*$ in Eq.\,\eqref{cons_break_thm} is given by $g_1^*(x) = g_1(x,y_0^*) + D_yg_0(x,y_0^*)[y_1^*]$
%\begin{align}
%g_1^*(x) = g_1(x,y_0^*) + D_yg_0(x,y_0^*)[y_1^*],
%\end{align}
another way to write the slow manifold reduction \eqref{test_first_order_reduction} is
\begin{align}
\widetilde{g}_\epsilon(x) = g^*_0(x) + \epsilon \,g_1^*(x).\label{simple_first_order_trunc}
\end{align}
The time derivative of $Q_\epsilon^{*1} = Q_0^* + \epsilon\, Q_1^*$ along a solution of the ODE defined by \eqref{simple_first_order_trunc} is given by
\begin{align}
DQ_\epsilon^{*1}(x)[\widetilde{g}_\epsilon] &= DQ_0^*(x)[g_0^*(x)]\nonumber\\
&+\epsilon\bigg(DQ_1^*(x)[g_0^*(x)] + DQ_0^*(x)[g_1^*(x)]\bigg)\nonumber\\
&+\epsilon^2\bigg(DQ_1^*(x)[g_1^*(x)]\bigg).
\end{align}
By the second-order consequence \eqref{cons_break_thm} of the inheritance Theorem \eqref{scalar_inheritance_thm} the previous expression simplifies to
\begin{align}
DQ_\epsilon^{*1}(x)[\widetilde{g}_\epsilon]  = -\epsilon^2\bigg( DQ_2^*(x)[g_0^*(x)]  + DQ_0^*(x)[g_2^*(x)]  \bigg),
\end{align}
which is by no means zero in general. We may say that the first-order slow manifold reduction \eqref{simple_first_order_trunc} conserves $Q_\epsilon^{*1}$ up to order $\epsilon^2$, but we cannot definitely say whether or not this first-order reduced model satisfies a conservation law associated with $Q_\epsilon$. Given an \emph{a priori} conserved quantity $Q_\epsilon$ and a slow manifold reduction $\widetilde{g}_\epsilon$ for a fast-slow system there are apparently additional conditions on $\widetilde{g}_\epsilon$ beyond the basic ordering condition \eqref{SMR_def} imposed by requiring $\widetilde{g}_\epsilon$ conserves some approximation to $Q_\epsilon^*$.

\begin{example}\label{energy_breakdown_example}
The generic breakdown of conservation laws in higher-order slow manifold reductions may be seen explicitly in the context of Abraham-Lorentz dynamics in the zero-drag regime. (c.f. Section \ref{zero_drag_sec}.) The zero-drag equations conserve the kinetic energy
\begin{align}
E =& \frac{1}{2}|\bm{v}|^2\nonumber\\
=& \frac{1}{2}v_\parallel^2 + \frac{1}{2}v_1^2 + \frac{1}{2} v_2^2,\label{conserved_kinetic_ex}
\end{align}
where the second line of \eqref{conserved_kinetic_ex} gives the conserved quantity in terms of the system's fast-slow split $x =(\bm{x},v_\parallel)$, $y = (v_1,v_2)$. In accordance with Eq.\,\eqref{limit_cons_thm} from Theorem \ref{scalar_inheritance_thm}, the $0^{\text{th}}$-order slow manifold reduction of the zero-drag equations,
\begin{align}
\dot{v}_\parallel =& 0\\
\dot{\bm{x}} = & v_\parallel \bm{b}(\bm{x}),
\end{align}
conserves the limiting kinetic energy
\begin{align}
E_0^*(\bm{x},v_\parallel) = \frac{1}{2}v_\parallel^2.
\end{align}
It is also true (somewhat serindipitously) that the first-order slow manifold reduction given by
\begin{align}
\dot{v_\parallel}& = 0\\
\dot{\bm{x}} &= v_\parallel\,\bm{b}  +\epsilon\,\bm{v}_c
\end{align}
conserves the first-order-accurate expression for the kinetic energy
\begin{align}
E_0^* + \epsilon\,E_1^* = \frac{1}{2}v_\parallel^2.
\end{align}
However, the second-order slow manifold reduction defined in Eqs.\,\eqref{2nd_order_smr_1}-\eqref{2nd_order_smr_2}
does not generally conserve the second-order-accurate energy,
\begin{align}
E_0^* + \epsilon E_1^* + \epsilon^2 E_2^* = \frac{1}{2}v_\parallel^2 + \epsilon^2\frac{1}{2}|\bm{v}_c|^2,\label{example_2nd_order_energy}
\end{align}
because 
\begin{align}
&\frac{d}{dt}(E_0^* + \epsilon E_1^* + \epsilon^2 E_2^*)= - \epsilon^2\,v_\parallel \bm{b}\cdot\nabla\frac{1}{2}|\bm{v}_c|^2\nonumber\\
& + \epsilon^2\,v_\parallel\bm{b}\cdot\nabla \frac{1}{2}|\bm{v}_c|^2+ \epsilon^3\,\bm{v}_c\cdot\nabla \frac{1}{2}|\bm{v}_c|^2\nonumber\\
&- \epsilon^4  \,\bigg(\frac{v_\parallel ( \bm{b}\cdot \nabla \bm{v}_c + \bm{v}_c\cdot \nabla\bm{b} )\times\bm{b}}{\zeta\,|\bm{B}|}\bigg) \cdot\nabla \frac{1}{2}|\bm{v}_c|^2\nonumber\\
 &= O(\epsilon^3).
\end{align}
This failure of the second-order slow manifold reduction \eqref{2nd_order_smr_1}-\eqref{2nd_order_smr_1} to conserve the second-order energy \eqref{example_2nd_order_energy} does \emph{not} imply that energy conservation breaks down at second order on the slow manifold. Instead it illustrates the more general point that not all slow manifold reductions of a given order are created equal. Some of those reductions need not be compatible with the underlying fast-slow system's exact conservation laws. In applications where conservation laws ought to be respected slow manifold reductions that preserve approximations of all known conservation laws should be preferred over naive power series truncations.

In Section \ref{hamiltonian_SM} we will explain a powerful strategy for constructing slow manifold reductions of any order that do preserve conservation laws, as well as other more nuanced structural properties, when the underlying fast-slow system admits a Hamiltonian structure. Other invariant-preserving slow manifold truncations for dissipative systems exist as well. For example, if some fraction of a fast-slow system's dependent variables are chosen to be equal to the conserved quantities of interest then it is straightforward to ensure that any slow manifold reduction preserves the associated conservation laws. To wit, in the zero-drag equations the parallel velocity $v_\parallel$ may be exchanged with the kinetic energy $\varepsilon$, giving the system of equations
\begin{align}
\epsilon\,\dot{v}_1&=\hspace*{.8em}\zeta\,|\bm{B}|\,v_2  \nonumber\\
&\hspace*{1.2em}+\epsilon\,(v_\parallel^{\star}\,\bm{b}+\bm{v}_\perp)\cdot(\bm{R}\,v_2 - v_\parallel^\star\,\nabla\bm{b}\cdot \bm{e}_1)\\
\epsilon\,\dot{v}_2&=-\zeta\,|\bm{B}|\,v_1\nonumber\\
&\hspace*{1.2em}-\epsilon\,(v_\parallel^{\star}\,\bm{b}+\bm{v}_\perp)\cdot(\bm{R}\,v_1 + v_\parallel^\star\,\nabla\bm{b}\cdot \bm{e}_2)\\
\dot{\varepsilon} &= 0\\
\dot{\bm{x}} &=  v_\parallel^{\star}\,\bm{b} + \bm{v}_\perp,
\end{align}
where
\begin{align}
v_\parallel^{\star} = \pm \sqrt{2\varepsilon - |\bm{v}_\perp|^2}.
\end{align}
In spite of the change of dependent variables these equations still comprise a fast-slow system with slow variable $x = (\bm{x},\varepsilon)$ and fast variable $y = (v_1,v_2)$. Theorem \ref{SM_eau} therefore ensures there is a unique formal slow manifold with slaving functions $v_{1\epsilon}^*(\bm{x},\varepsilon),v_{2\epsilon}^*(\bm{x},\varepsilon)$ and associated formal slow manifold reduction
\begin{align}
\dot{\varepsilon}_\epsilon^*& = 0\label{energy_surface_one}\\
\dot{\bm{x}}_\epsilon^* & = \pm\sqrt{2\varepsilon - |\bm{v}_{\perp\epsilon}^*|^2} \,\bm{b}(\bm{x})+ \bm{v}_{\perp\epsilon}^*.\label{energy_surface_two}
\end{align}
If a slow manifold reduction of order $N$ is extracted from Eqs.\,\eqref{energy_surface_one}-\eqref{energy_surface_two} by setting all $O(\epsilon^{N+1})$ terms equal to zero the kinetic energy $\varepsilon$ will be conserved regardless of the $N$ chosen. Note that the slaving functions $v_{1\epsilon}^*(\bm{x},\varepsilon),v_{2\epsilon}^*(\bm{x},\varepsilon)$ are not the same as the slaving functions found earlier in the $x=(\bm{x},v_\parallel)$ representation of the fast-slow system. Also note that the $x = (\bm{x},\varepsilon)$ representation has the disadvantage of necessarily coping with the two branches of the square root function.
\end{example}

\subsection{The zero-derivative principle\label{zdp_sec}}
The preceding discussion suggests that working with slow manifolds (Section \ref{origins_sec}) and their associated slow manifold reductions (Section \ref{reduction_sec}) seems to require the laborious intermediate step of finding perturbative solutions of the invariance equation \eqref{scaled_invariance_eq_def} using pencil and paper. In each of the examples we have used to illustrate basic elements of the theory we have performed the necessary calculations in precisely this manner. However, manual calculations, especially those involving perturbation expansions, may only be carried so far. It is often true that computing second- or third-order terms of a formal slow manifold requires only modest bookkeeping and error checking. But fourth-order calculations easily test the limits of human fortitude, and may even push computer algebra systems (CAS) beyond their capacities. At higher orders still, neither the sweat of scientists nor the crunch of CAS will serve to reliably complete the computations.

In this Section we will present an alternative approach to finding approximate solutions of the invariance equation that is amenable to numerical computation. (Here we distinguish between numerical computations and symbolic computations.) The approach is based on the so-called \emph{zero-derivative principle}, and has been applied by Gear, Kaper, Kevrekidis, and Zagaris \citep{Gear_2006} to numerically compute approximations to slow manifolds for various practically-relevant problems. With the zero-derivative principle in hand points near a fast-slow system's slow manifold $\{(x,y_\epsilon^*(x))\}$ may be computed given (a) an accurate numerical integration scheme for full fast-slow dynamics, and (b) a choice for the $x$-component of the desired point $(x,y_\epsilon^*(x))$ on the slow manifold. Manual calculations are not required. In principle the formal slow manifold may be approximated with any desired accuracy. In practice this accuracy is limited by the truncation and roundoff errors that affect all numerical recipes. 

Our discussion will focus on a theoretical demonstration of the zero-derivative principle. For details on numerical implementation of the theory we refer to \citep{Gear_2005,Gear_2006,Zagaris_2009}. Other numerical approaches to approximating slow manifolds include the intrinsic low-dimensional manifold approach of \citep{Maas_1992}, the Computational Singular Perturbation (CSP) technique of \citep{Lam_1994}, the extended zero-derivative principle of \citep{Benoit_2015}, the bounded derivative principle developed in \citep{Kreiss_1985} and applied in \citep{Ariel_2012}, and Fraser's method from \citep{fraser_1988}.

The basic idea underlying the zero-derivative principle is that solutions of a fast-slow system that are stuck to a slow manifold do not sample the $O(\epsilon)$ timescale. Therefore the time derivatives of such a solution should not grow as $\epsilon$ tends to $0$. Moreover if a solution \emph{does} contain a fast component with small amplitude then a sufficiently high-order time derivative of that solution should amplify that fast component to an appreciable size. We are therefore motivated to investigate the relationship between high-order time derivatives of solutions to fast-slow systems and high-order approximations of slow manifolds.

It will be convenient for the purposes of this discussion to study fast-slow systems in fast time $\tau = t/\epsilon$. In terms of $\tau$ a general fast-slow system takes the form
\begin{align}
\frac{dy}{d\tau} &= f_\epsilon(x,y)\\
\frac{dx}{d\tau}& = \epsilon\,g_\epsilon(x,y),
\end{align}
and the invariance equation \eqref{scaled_invariance_eq_def}  that determines the system's formal slow manifold does not change. Note that in terms of $\tau$ the system remains in the category of ordinary differential equations as $\epsilon\rightarrow0$.

On the formal slow manifold the first $\tau$-derivative of the fast variable $y$ is nearly zero. This can be seen by appealing directly to the invariance equation \eqref{scaled_invariance_eq_def}. Alternatively we may come to the same conclusion by comparing the formal slow manifold $y_\epsilon^*(x)$ with the function $\widetilde{y}_\epsilon^{1}(x)$ whose graph is the zero level set of $dy/d\tau =f_\epsilon(x,y)$ as follows. (This may seem like an inefficient demonstration, but it serves to motivate our general proof of the zero derivative principle.) Because 
\begin{align}
0 =  f_\epsilon(x,\widetilde{y}_\epsilon^{1}(x)),
\end{align}
the formal power series expansion of  $\widetilde{y}_\epsilon^{1}(x)$,
\begin{align}
\widetilde{y}_\epsilon^{1} = \widetilde{y}_0^{1} + \epsilon\,\widetilde{y}_1^{1} + \epsilon^2\,\widetilde{y}_2^{1}+\dots,  
\end{align}
satisfies the sequence of equations
\begin{align}
0&=f_0(x,\widetilde{y}_0^{1})\\
0&=f_1(x,\widetilde{y}_0^{1}) + D_yf_0(x,\widetilde{y}_0^{1})[\widetilde{y}_1^{1}]\\
0&=\dots.
\end{align}
These equations show in particular that 
\begin{align}
\widetilde{y}_0^{1}(x)  = y_0^*(x),
\end{align}
where $y_0^*$ is the leading-order term in the formal slow manifold $y_\epsilon^*$. Therefore the $\tau$-derivative of $y$ on the slow manifold may be estimated as follows. If $\delta y = y_\epsilon^*-\widetilde{y}^{1}_\epsilon$, then 
\begin{align}
\bigg(\frac{dy}{d\tau}\bigg)^*_\epsilon =& f_\epsilon(x,y_\epsilon^*)\nonumber\\
=&f_\epsilon(x,\widetilde{y}^{1}_\epsilon + \delta y)\nonumber\\
=&f_\epsilon(x,\widetilde{y}^{1}_\epsilon)\nonumber\\
& + \int_0^1D_yf_\epsilon(x,\widetilde{y}^{1}_\epsilon + \lambda\,\delta y)[\delta y]\,d\lambda\nonumber\\
=&\bigg(\int_0^1D_yf_\epsilon(x,\widetilde{y}^{1}_\epsilon + \lambda\,\delta y)\,d\lambda\bigg)[\delta y].\label{first_estimate_zdp}
\end{align}
Because 
\begin{align}
\delta y =& (y_0^*+\epsilon\,y_1^* + \dots) - (\widetilde{y}_0^{1} + \epsilon\,\widetilde{y}_1^{1} +\dots)\nonumber\\
=&\epsilon (y_1^* - \widetilde{y}_1^{1}) +\dots\nonumber\\
=&O(\epsilon),
\end{align}
and the linear operator in parentheses in Eq.\,\eqref{first_estimate_zdp} is $O(1)$,  we conclude from Eq.\,\eqref{first_estimate_zdp} that $\left(\frac{dy}{d\tau}\right)^*_\epsilon =O(\epsilon)$.

Aside from demonstrating the smallness of the first $\tau$-derivative on the formal slow manifold the preceding argument also shows that we may obtain a $0^{\text{th}}$-order accurate approximation $\widetilde{y}_\epsilon^{1}$ of the formal slow manifold $y_\epsilon^*$ by solving the algebraic equation
\begin{align}
0 = \mathfrak{D}^1_\epsilon(x,\widetilde{y}_\epsilon^{1}(x)) = f_\epsilon(x,\widetilde{y}_\epsilon^{1}(x)).
\end{align}
(This is not a trivial observation because the limiting slow manifold is defined by the equation $\mathfrak{D}^1_0 = 0$, not $\mathfrak{D}^1_\epsilon = 0$.)
Here $0^{\text{th}}$-order accurate means $y_\epsilon^* - \widetilde{y}_\epsilon^{1} = O(\epsilon)$. Referring then to the basic idea underlying the zero-derivative principle explained earlier we are therefore lead to consider the zero level sets of the higher derivative functions
\begin{align}
\mathfrak{D}^N_\epsilon(x,y) \equiv \frac{d^Ny}{d\tau^N},\label{frakD_def}
\end{align}
and how well solutions $\widetilde{y}^N_\epsilon$ of the associated vanishing derivative conditions
\begin{align}
\mathfrak{D}^N_\epsilon(x,\widetilde{y}^N_\epsilon(x)) = 0,
\end{align}
approximate the formal slow manifold $y_\epsilon^*$. In particular we are motivated to prove the following
\begin{theorem}[The zero-derivative principle]\label{zdp}
Given a fast-slow system with formal slow manifold $y_\epsilon^*$ define the functions $\mathfrak{D}^N_\epsilon(x,y)$ according to Eq.\,\eqref{frakD_def}. If $\widetilde{y}^N_\epsilon(x)$ is the unique solution of 
\begin{align}
\mathfrak{D}^N_\epsilon(x,\widetilde{y}^N_\epsilon(x)) = 0\label{zero_nth_deriv}
\end{align}
with $\widetilde{y}^N_0 = y_0^*$ then 
\begin{align}
y_\epsilon^* - \widetilde{y}^N_\epsilon = O(\epsilon^{N}).
\end{align}
In other words, the graph of $\widetilde{y}^N_\epsilon$ defines a slow manifold of order $N-1$. (c.f. Definition \ref{sm_defined_properly}.)
\end{theorem}

Theorem \ref{zdp} is powerful for two reasons. First, whereas a slow manifold was defined in Def.\,\ref{sm_defined_properly} as an approximate solution of a nonlinear system of partial differential equations, Theorem \ref{zdp} shows that the same objects may be computed as solutions of nonlinear systems of algebraic equations. Second, because the level set functions $\mathfrak{D}_\epsilon^N$ are equal to time derivatives of the fast variable $y$ evaluation of $\mathfrak{D}_\epsilon^N$ may be performed by running simulations of the underlying fast-slow system and using the simulation results to approximate the appropriate derivative. In particular there is no need to compute the $\mathfrak{D}_\epsilon^N$ for large $N$ by hand. Put slightly differently, Theorem \ref{zdp} reduces the computation of accurate slow manifolds to the combination of two problems (1) finding the zeros of a function, and (2) estimating high-order derivatives of solutions of ordinary differential equations. Each of these problems is well suited to numerical computation; neither requires any amount of pen-and-paper manipulations. (Of course, the computations \emph{could} be done by hand given sufficient motivation.)

In order to prove Theorem \ref{zdp} it is useful to first assemble four supporting Lemmas so that the proof's general structure will be more apparent. After stating and proving the Lemmas we will immediately give the desired proof.

 First we take note of the following basic relationship between the functions $\mathfrak{D}_\epsilon^N$ for different values of $N$ that reflects the simple recursive property of ordinary derivatives $(d/d\tau)^N = (d/d\tau)(d/d\tau)^{N-1}$.

\begin{lemma}\label{recursive_derivative_formula}
For each $N> 1$ the function $\mathfrak{D}_\epsilon^N(x,y)$ defined in Eq.\,\eqref{frakD_def} is related to the function $\mathfrak{D}_\epsilon^{N-1}(x,y)$ by the formula
\begin{align}
\mathfrak{D}_\epsilon^N(x,y) = &\epsilon\,D_x\mathfrak{D}_\epsilon^{N-1}(x,y)[g_\epsilon(x,y)]\nonumber\\
 &+ D_y\mathfrak{D}_\epsilon^{N-1}(x,y)[f_\epsilon(x,y)].
\end{align}
\end{lemma}

\begin{proof}
This is merely an application of the chain rule to formula
\begin{align}
\mathfrak{D}_\epsilon^N(x,y)  = \frac{d}{d\tau}\mathfrak{D}_\epsilon^{N-1}(x,y).
\end{align}
\end{proof}

Next we note that the non-degeneracy property \eqref{FS_condition} satisfied by all fast-slow systems, which may be read as a property of the function $\mathfrak{D}_0^1$, naturally ``propagates" across $N$ values.
\begin{lemma}\label{invertibility_lemma}
For each $N\geq 1$ the function $\mathfrak{D}_\epsilon^N(x,y)$ defined in Eq.\,\eqref{frakD_def} satisfies the following non-degeneracy condition. The partial derivative $D_y\mathfrak{D}_\epsilon^N(x,y)$ is invertible with $O(1)$ inverse for $(\epsilon,x,y)$ in a neighborhood of $\Gamma_0 = \{(\epsilon,x,y)\mid \epsilon = 0,\,y= y_0^*(x)\}$.

\end{lemma}  

\begin{proof}
First consider $N=1$. By property \eqref{FS_condition} $D_y\mathfrak{D}_0^1(x,y_0^*(x))$ is invertible for each $x$. Because the space of invertible operators is open, $D_y\mathfrak{D}_\epsilon^1(x,y)$ must also be invertible for $(\epsilon,x,y)$ sufficiently close to $(0,x,y_0^*(x))$. This proves the non-degeneracy property for $N=1$.

Now suppose that the non-degeneracy property is satisfied for $N = m-1$. We will prove that it must also be satisfied for $N = m$, and therefore by induction that it is satisfied for all $N\geq 1$. By Lemma \ref{recursive_derivative_formula} we may express the function $\mathfrak{D}_\epsilon^{m}$ in terms of $\mathfrak{D}_\epsilon^{m-1}$ as
\begin{align}
\mathfrak{D}_\epsilon^m(x,y) = &\epsilon\,D_x\mathfrak{D}_\epsilon^{m-1}(x,y)[g_\epsilon(x,y)]\nonumber\\
 &+ D_y\mathfrak{D}_\epsilon^{m-1}(x,y)[f_\epsilon(x,y)].\label{m_eqn}
\end{align}
Therefore the partial derivative $D_y\mathfrak{D}_\epsilon^m(x,y)$ evaluated for $(\epsilon,x,y)=(0,x,y_0^*(x))\in\Gamma_0$ must be given by
\begin{align}
D_y\mathfrak{D}_0^m(x,y_0^*) = D_y\mathfrak{D}_0^{m-1}(x,y_0^*)[D_yf_0(x,y_0^*)].
\end{align}
This formula expresses $D_y\mathfrak{D}_0^m(x,y_0^*)$ as the composition of the operator $D_y\mathfrak{D}_0^{m-1}(x,y_0^*)$, which is invertible by the inductive hypothesis, and the operator $D_yf_0(x,y_0^*)$, which is invertible by property \eqref{FS_condition}. Because the composition of invertible operators is also invertible, we conclude that $D_y\mathfrak{D}_0^m(x,y_0^*)$ itself is invertible. The invertibility of $D_y\mathfrak{D}_\epsilon^m(x,y)$ for $(\epsilon,x,y)$ \emph{near} $\Gamma_0$ now follows from the openness of the space of invertible linear operators, just as in the case $N=1$.
\end{proof}

Next we will justify the implicit claim made in the statement of Theorem \ref{zdp} that there is a unique solution $\widetilde{y}^N_\epsilon$ of $\mathfrak{D}^N_\epsilon = 0$ that is close to the limiting slow manifold $y_0^*$.

\begin{lemma}\label{uniqueness_lemma}
For each fixed $x$ there exists a unique solution $y = \widetilde{y}_\epsilon^N$ of the algebraic equation 
\begin{align}
\mathfrak{D}_\epsilon^N(x,y) = 0\label{def_eqs}
\end{align}
that satisfies $\widetilde{y}_0^N = y_0^*$.
\end{lemma}

\begin{proof}
The proof is an application of the implicit function theorem. Define the three-variable function $\mathfrak{D}^N(\epsilon,x,y) = \mathfrak{D}_\epsilon^N(x,y)$. By Lemma \ref{recursive_derivative_formula}, this function may also be written
\begin{align}
\mathfrak{D}^N(\epsilon,x,y)=&\epsilon\,D_x\mathfrak{D}_\epsilon^{N-1}(x,y)[g_\epsilon(x,y)]\nonumber\\
 &+ D_y\mathfrak{D}_\epsilon^{N-1}(x,y)[f_\epsilon(x,y)].
\end{align}
For each $x_0$ the value of $\mathfrak{D}^N(0,x_0,y_0^*(x_0))$ is therefore
\begin{align}
\mathfrak{D}^N(0,x_0,y_0^*(x_0)) = D_y\mathfrak{D}_0^{N-1}(x_0,y_0^*)[f_0(x_0,y_0^*)]=0,
\end{align}
where we have used the definition of the limiting slow manifold. In addition the partial derivative $D_y\mathfrak{D}^N(0,x_0,y_0^*(x_0))$ is given by
\begin{align}
D_y\mathfrak{D}^N(0,x_0,y_0^*(x_0)) =D_y\mathfrak{D}_\epsilon^{N-1}(x_0,y_0^*)[D_yf_\epsilon(x_0,y_0^*)],
\end{align}
which is invertible by Lemma \ref{invertibility_lemma}. The implicit function theorem therefore says that there is a neighborhood $W$ of $(0,x_0)$ and a unique function $\widetilde{y}^N:W\rightarrow Y$ defined on that neighborhood such that $\widetilde{y}^N(0,x_0) = y_0^*(x_0)$ and
\begin{align}
\mathfrak{D}^N(\epsilon,x,\widetilde{y}^N(\epsilon,x)) = 0
\end{align}
for each $(\epsilon,x)\in W$.
If two of these neighborhoods intersect nontrivially the functions $\widetilde{y}^N(\epsilon,x)$ must agree on the overlaps. Therefore the locally-defined $\widetilde{y}^N(\epsilon,x)$ stitch together to globally define a unique function
\begin{align}
\widetilde{y}^N_\epsilon(x) \equiv \widetilde{y}^N(\epsilon,x),
\end{align}
that satisfies $\mathfrak{D}_\epsilon^N=0$ and $\widetilde{y}^N_0 = y_0^*$.
\end{proof}

As a final step in preparation for the proof of Theorem \ref{zdp} we will establish the basic relationship between the formal slow manifold $y_\epsilon^*$ and the derivative functions $\mathfrak{D}_\epsilon^N(x,y)$.
\begin{lemma}\label{fsm_zdp}
If $y_\epsilon^*$ is the formal slow manifold associated with a fast-slow system then 
\begin{align}
\mathfrak{D}_\epsilon^N(x,y_\epsilon^*(x)) = O(\epsilon^N),\label{fsm_basic_estimate_zdp}
\end{align}
where $\mathfrak{D}_\epsilon^N(x,y)$ is defined in Eq.\,\eqref{frakD_def}.
\end{lemma}

\begin{proof}
This will be another proof by induction. Because the earlier discussion in the text already establishes the result \eqref{fsm_basic_estimate_zdp} for $N=1$, all we must do is prove \eqref{fsm_basic_estimate_zdp} for $N = m$ assuming \eqref{fsm_basic_estimate_zdp} is true when $N=m-1$.

As in the proof of the Lemma \eqref{invertibility_lemma} we begin by expressing $\mathfrak{D}_\epsilon^m(x,y)$ in terms of $\mathfrak{D}_\epsilon^{m-1}(x,y)$ using Eq.\,\eqref{m_eqn}. We then substitute $y = y_\epsilon^*$, use the invariance equation \eqref{scaled_invariance_eq_def}, and apply the chain rule as follows.
\begin{align}
\mathfrak{D}_\epsilon^m(x,y_\epsilon^*) = &\epsilon\,D_x\mathfrak{D}_\epsilon^{m-1}(x,y)[g_\epsilon(x,y_\epsilon^*)]\nonumber\\
 &+ D_y\mathfrak{D}_\epsilon^{m-1}(x,y_\epsilon^*)[f_\epsilon(x,y_\epsilon^*)]\nonumber\\
(\text{by \eqref{scaled_invariance_eq_def}}) =&\epsilon\,D_x\mathfrak{D}_\epsilon^{m-1}(x,y)[g_\epsilon(x,y_\epsilon^*)]\nonumber\\
 &+ \epsilon\,D_y\mathfrak{D}_\epsilon^{m-1}(x,y_\epsilon^*)[Dy_0^*(x)[g_\epsilon(x,y_\epsilon^*)]]\nonumber\\
 (\text{by \eqref{formal_chain_rule}})=& \epsilon\, D\alpha_\epsilon^{m-1}(x)[g_\epsilon(x,y_\epsilon^*)],\label{estimate_formal_zdp}
\end{align}
where
\begin{align}
\alpha_\epsilon^{m-1}(x) = \mathfrak{D}_\epsilon^{m-1}(x,y_\epsilon^*(x)).
\end{align}
By the inductive hypothesis, the function $\alpha_\epsilon^{m-1}(x) = O(\epsilon^{m-1})$. Substituting this estimate into Eq.\,\eqref{estimate_formal_zdp} therefore gives
\begin{align}
\mathfrak{D}_\epsilon^m(x,y_\epsilon^*)  = \epsilon\,O(\epsilon^{m-1}) = O(\epsilon^m).
\end{align}
\end{proof}

The proof of Theorem \ref{zdp} now proceeds as follows.
\begin{proof}[proof of Theorem \ref{zdp}]
Set $\delta y = y_\epsilon^* - \widetilde{y}_\epsilon^N$. By Lemma \ref{fsm_zdp} we have the simple estimate
\begin{align}
\mathfrak{D}_\epsilon^N(x,y_\epsilon^*) = & \mathfrak{D}_\epsilon^N(x,\widetilde{y}_\epsilon^N + \delta y)\nonumber\\
=&\mathfrak{D}_\epsilon^N(x,\widetilde{y}_\epsilon^N) \nonumber\\
&+ \int_0^1 D_y \mathfrak{D}_\epsilon^N(x,\widetilde{y}_\epsilon^N +\lambda \delta y)[\delta y]\,d\lambda\nonumber\\
=&\bigg(\int_0^1 D_y \mathfrak{D}_\epsilon^N(x,\widetilde{y}_\epsilon^N +\lambda \delta y)\,d\lambda\bigg)[\delta y]\nonumber\\
=&O(\epsilon^N).
\end{align}
Therefore the proof will be complete if we can show that the linear operator
\begin{align}
L_\epsilon = \int_0^1 D_y \mathfrak{D}_\epsilon^N(x,\widetilde{y}_\epsilon +\lambda \delta y)\,d\lambda
\end{align}
is invertible with $O(1)$ inverse. (Note that $L_\epsilon$ must be interpreted as a formal power series whose coefficients are linear operators.)

Lemma \ref{uniqueness_lemma} demonstrates that $\widetilde{y}^N_\epsilon = y_0^* + O(\epsilon)$ and that $\delta y = O(\epsilon)$. Therefore when $\epsilon = 0$, the operator $L_\epsilon$ is given by
\begin{align}
L_0 = D_y \mathfrak{D}_0^N(x,y_0^* ),
\end{align}
which is invertible by Lemma \ref{invertibility_lemma}. It follows that if we write $L_\epsilon = L_0 + \epsilon\,\delta L_\epsilon$, where $\delta L_\epsilon = O(1)$ accounts for all of the higher-order terms in the series $L_\epsilon$, we may write the formal inverse (i.e. the inverse to all orders in $\epsilon$) of $L_\epsilon$ as
\begin{align}
(L_\epsilon)^{-1} = L_0^{-1}\bigg(1 - \epsilon\,[L_0^{-1}\delta L_\epsilon]+\epsilon^2[L_0^{-1}\delta L_\epsilon]^2 + O(\epsilon^3)\bigg),
\end{align}
which is clearly $O(1)$.
\end{proof}

\section{Interpretation of Formal Slow Manifolds\label{interpretation_SM}}
As explained in Section \ref{basic_SM_theory} in the context of fast-slow systems, every slow manifold of order $N$ is given as the graph of a function $\widetilde{y}_\epsilon$ that approximately solves the invariance equation. Since true solutions of the invariance equation give rise to invariant manifolds it is natural to wonder if slow manifolds have properties that approximate those of invariant manifolds. This Section aims to study this question with greater precision than was used in Section \ref{basic_SM_theory}. In particular it aims to describe approaches to assessing the \emph{normal stability} of a slow manifold. If a trajectory starts on a slow manifold will it remain close to the slow manifold for a substantial period of time, or will it rapidly wander away? We will discuss both the magnitude of this normal wandering and the timescale over which the wandering remains small. 

This question is interesting from the perspective of plasma physics because it cuts to the heart of an important and often overlooked aspect of reduced plasma models: the reduced model's time of validity. If a reduced model arises as the slow manifold reduction of a fast-slow system then the validity time of that model can be no longer than the timescale for normal stability of the associated slow manifold. Over longer timescales the role of physical effects not captured by the reduced model becomes amplified, possibly leading to a dramatic breakdown of the reduced model's predictions. This breakdown could be signaled from within the reduced model, as when the firehose instability develops in the kinetic MHD model \citep{Schekochihin_2010} or when high-$\bm{k}$ whistler waves develop in Hall MHD \citep{Farmer_2019,Huba_1991}, or may be undetectable based on predictions of the reduced model alone, as in the case of applying averaging to nearly-integrable systems exhibiting Arnold diffusion. \citep{Capinsky_2016} (The action variables in such systems can wander appreciably on timescales that are longer than $1/\epsilon^N$ for each $N$, but the averaged dynamics make no such prediction.)

As a first step we will formulate a general approach to assessing normal stability that mimics the familiar process of analyzing the linear stability of equilibria. In particular we will identify a non-autonomous linear system associated with a given slow manifold whose stability properties are tied to the slow manifold's nonlinear normal stability. Being non-autonomous, this linear system is in general more difficult to analyze than the sort of linear systems used to study the stability of equilibria or periodic orbits. However the goal of this system's analysis is also more ambitious than equilibrium stability analyses --- the aim is to assess the stability of an entire reduced model rather than an individual solution.

Next we discuss two important classes of slow manifolds that arise frequently in applications and for which specialized techniques exist for the analysis of normal stability --- normally-hyperbolic and normally-elliptic slow manifolds. These two classes correspond to slow manifolds that exponentially attract or repel nearby trajectories, and slow manifolds around which nearby trajectories oscillate, respectively. In the normally hyperbolic case we will summarize how Fenichel's geometric singular perturbation theory \citep{Fenichel_1979} establishes the existence of a true invariant manifold whose asymptotic expansion agrees with the formal slow manifold to all orders in perturbation theory. In the normally elliptic case we will show that system trajectories starting on a slow manifold will typically remain nearby on an $O(1)$ time interval. While this interval significantly exceeds the $O(\epsilon)$ timescale characterizing the fast dynamics, it generally cannot be enlarged by decreasing $\epsilon$ due to potential resonances between the normal oscillations. We will also explain how the theory of adiabatic invariance for Hamiltonian system can sometimes be used to significantly improve the normal stability timescale when resonance phenomena only weakly effect the slow manifold.

Our aim in presenting this material is not to provide definitive estimates for the normal stability of slow manifolds in any generality. Instead we hope to establish useful heuristics that may be applied without excessive difficulty by practicing plasma physicists. Because the stability of reduced plasma models is so frequently overlooked, we believe that a working knowledge of these heuristics is just as important as a working knowledge of the more formal aspects of slow manifold reduction theory outlined in Section \ref{basic_SM_theory}. This Section should also be compared with MacKay's discussion of slow manifold stability in \citep{MacKay_2004}.

\subsection{Normal deviations: general case}
Suppose we are given a fast-slow system and a slow manifold $\widetilde{y}_\epsilon$ for that system of order $N$. We say that a solution $(x(t),y(t))$ of this system \emph{starts on the slow manifold} if
\begin{align}
y(0)= \widetilde{y}_\epsilon(x(0)).
\end{align}
If the graph of $\widetilde{y}_\epsilon$ were an invariant manifold then solutions that start on the slow manifold would stay on the slow manifold for all time, i.e. we would have $y(t) = \widetilde{y}_\epsilon(x(t)) $ for all $t\in\mathbb{R}$. However because $\widetilde{y}_\epsilon$ only solves the invariance equation approximately the true value of $y(t)$ must in general differ somewhat from $\widetilde{y}_\epsilon(x(t))$. We would like to estimate the size of this deviation,
\begin{align}
\xi(t) = y(t) - \widetilde{y}_\epsilon(x(t)),\label{xi_introduced}
\end{align}
under the assumption that $\xi(0) = 0$. This is what we will refer to as the \emph{normal stability problem}.

To that end interpret Eq.\,\eqref{xi_introduced} as a change of dependent variables $(x,y)\mapsto (x,\xi)$. In the new variables $(x,\xi)$ the fast-slow system becomes
\begin{align}
\epsilon\,\dot{\xi} =& f_\epsilon(x,\widetilde{y}_\epsilon + \xi) - \epsilon\,D\widetilde{y}_\epsilon(x)[g_\epsilon(x,\widetilde{y}_\epsilon + \xi)]\label{x_xi_one}\\
\dot{x} = & g_\epsilon(x,\widetilde{y}_\epsilon + \xi).\label{x_xi_two}
\end{align}
The most direct approach to assessing the normal stability of the slow manifold $\widetilde{y}_\epsilon$ is to solve Eqs.\,\eqref{x_xi_one}-\eqref{x_xi_two} with the initial condition $(x(0),\xi(0)) = (x,0)$ and then check if the norm $||\xi(t)||$ remains small for $t$ away from $0$, or if it grows when $t$ is sufficiently large. (We will not discuss the choice of norm here; in finite dimensions the choice is often not so important, while in infinite dimensions the choice may be extremely subtle.) Of course such a procedure is overly ambitious because exact solutions of the nonlinear differential equations \eqref{x_xi_one}-\eqref{x_xi_two} usually cannot be computed by hand. However we \emph{do} know that $\widetilde{y}_\epsilon$ approximately solves the invariance equation, which suggests we might be able to make partial progress without knowledge of the exact solution. In fact we will show that our knowledge of $\widetilde{y}_\epsilon$ allows us to put Eqs.\,\eqref{x_xi_one}-\eqref{x_xi_two} into a \emph{normal form} that is often amenable to analysis of $||\xi(t)||$.

As a first step in identifying the desired normal form we recapitulate the sense in which $\widetilde{y}_\epsilon$ approximately solves the invariance equation as follows. Because $\delta y = \widetilde{y}_\epsilon - y_\epsilon^* = O(\epsilon^{N+1})$ the residual of the invariance equation,
\begin{align}
R_\epsilon(x) = f_\epsilon(x,\widetilde{y}_\epsilon) - \epsilon\,D\widetilde{y}_\epsilon(x)[g_\epsilon(x,\widetilde{y}_\epsilon)],
\end{align}
must satisfy
\begin{align}
R_\epsilon(x) &=  f_\epsilon(x,y_\epsilon^*+\delta y) \nonumber\\
&- \epsilon\,D(y_\epsilon^*+\delta y)(x)[g_\epsilon(x,y_\epsilon^*+\delta y)]\nonumber\\
& = f_\epsilon(x,y_\epsilon^*) - \epsilon\,Dy_\epsilon^*(x)[g_\epsilon(x,y_\epsilon^*)]\nonumber\\
& +\bigg(\int_0^1 D_yf_\epsilon(x,y_\epsilon^*+\lambda\delta y)\,d\lambda\bigg)[\delta y]\nonumber\\
& -\epsilon\,D\delta y(x)\left[\int_0^1g_\epsilon(x,y_\epsilon^*+\lambda\delta y)\,d\lambda\right]\nonumber\\
& - \epsilon\,Dy_\epsilon^*(x)\left[\int_0^1D_yg_\epsilon(x,y_\epsilon^*+\lambda\delta y)[\delta y]\,d\lambda\right]\nonumber\\
& = O(\epsilon^{N+1}).
\end{align}
Therefore there must be an $O(1)$ function $r_\epsilon(x)$ such that $R_\epsilon(x) = \epsilon^{N+1}r_\epsilon(x)$. 

Next we use $R_\epsilon = \epsilon^{N+1}r_\epsilon$ to re-write the dynamical equations \eqref{x_xi_one}-\eqref{x_xi_two} for $(x,\xi)$ as
\begin{align}
\epsilon\,\dot{\xi} = & F_\epsilon(x,\xi)[\xi] - \epsilon\,D\widetilde{y}_\epsilon(x)[G_\epsilon(x,\xi)[\xi]] \nonumber\\
&+ \epsilon^{N+1}\,r_\epsilon(x)\\
\dot{x} = & g_\epsilon(x,\widetilde{y}_\epsilon) + G_\epsilon(x,\xi)[\xi],
\end{align}
where the $(x,\xi)$-dependent linear operators $F_\epsilon(x,\xi),G_\epsilon(x,\xi)$ are given by
\begin{align}
F_\epsilon(x,\xi) =& \int_0^1 D_yf_\epsilon(x,\widetilde{y}_\epsilon + \lambda\xi)\,d\lambda\\
G_\epsilon(x,\xi) = & \int_0^1 D_yg_\epsilon(x,\widetilde{y}_\epsilon + \lambda\xi)\,d\lambda.
\end{align}

Finally we use the Taylor identities
\begin{align}
&F_\epsilon(x,\xi) =  D_yf_\epsilon(x,\widetilde{y}_\epsilon)\nonumber\\
 +& \int_0^1\int_0^1\lambda_0\,D^2_yf_\epsilon(x,\widetilde{y}_\epsilon + \lambda_0\lambda_1\xi)[\xi,\xi]\,d\lambda_0\,d\lambda_1\\
 &G_\epsilon(x,\xi) =  D_yG_\epsilon(x,\widetilde{y}_\epsilon)\nonumber\\
 +& \int_0^1\int_0^1\lambda_0\,D^2_yg_\epsilon(x,\widetilde{y}_\epsilon + \lambda_0\lambda_1\xi)[\xi,\xi]\,d\lambda_0\,d\lambda_1,
\end{align}
to write the evolution equations for $(x,\xi)$ in a form that is weakly-nonlinear in $\xi$:
\begin{align}
\epsilon\,\dot{\xi} &= D_yf_\epsilon(x,\widetilde{y}_\epsilon)[\xi] - \epsilon\,D\widetilde{y}_\epsilon(x)[D_yg_\epsilon(x,\widetilde{y}_\epsilon)[\xi]]\nonumber\\
&+\mathcal{N}_\epsilon(x,\xi) + \epsilon^{N+1}\,r_\epsilon(x)\label{nf_one}\\
\dot{x} &= g_\epsilon(x,\widetilde{y}_\epsilon) + D_yg_\epsilon(x,\widetilde{y}_\epsilon)[\xi] + \mathcal{M}_\epsilon(x,\xi). \label{nf_two}
\end{align}
Here $\mathcal{N}_\epsilon(x,\xi) = O(||\xi||^2)$ and $\mathcal{M}_\epsilon(x,\xi)=O(||\xi||^2)$ are each nonlinear functions of $\xi$. This proves the following. (C.f. the discussion of normal forms in \citep{Cox_2003}) 
\begin{proposition}[slow manifold normal form]\label{nf_prop}
Suppose the graph of $\widetilde{y}_\epsilon$ is a slow manifold of order $N$ for a fast-slow system $\epsilon\,\dot{y} = f_\epsilon(x,y)$, $\dot{x} = g_\epsilon(x,y)$. Then there is a change of dependent variables $(x,y)\mapsto (x,\xi)$ that transforms for the fast-slow system into the form
\begin{align}
\epsilon\,\dot{\xi} &= D_yf_\epsilon(x,\widetilde{y}_\epsilon)[\xi] - \epsilon\,D\widetilde{y}_\epsilon(x)[D_yg_\epsilon(x,\widetilde{y}_\epsilon)[\xi]]\nonumber\\
&+\mathcal{N}_\epsilon(x,\xi) + \epsilon^{N+1}\,r_\epsilon(x)\label{nf_thm_one}\\
\dot{x} &= g_\epsilon(x,\widetilde{y}_\epsilon) + D_yg_\epsilon(x,\widetilde{y}_\epsilon)[\xi] + \mathcal{M}_\epsilon(x,\xi),\label{nf_thm_two}
\end{align}
where $\mathcal{N}_\epsilon(x,\xi),\mathcal{M}_\epsilon(x,\xi)$ are nonlinear functions of $\xi$, and $r_\epsilon = O(1)$ as $\epsilon\rightarrow 0$.
\end{proposition}

The power of Proposition \ref{nf_prop} comes from the weakly-nonlinear evolution equation for $\xi$. This equation illuminates the basic role played by the non-autonomous linear system
\begin{align}
\epsilon\,\dot{\widetilde{\xi}} &= D_yf_\epsilon(x,\widetilde{y}_\epsilon)[\widetilde{\xi}] - \epsilon\,D\widetilde{y}_\epsilon(x)[D_yg_\epsilon(x,\widetilde{y}_\epsilon)[\widetilde{\xi}]],\label{stability_eqn}
\end{align}
where $x = x(t)$ is the $x$-component of a solution $(x(t),\xi(t))$ of Eqs.\,\eqref{nf_thm_one}-\eqref{nf_thm_two}. As the following Theorem will show, establishing stability of the linear system \eqref{stability_eqn} is sometimes sufficient for establishing corresponding nonlinear normal stability estimates for the underlying fast-slow system. In particular it will show that if the linear equation \eqref{stability_eqn} exhibits stability on an $O(1/\epsilon^k)$ timescale then under certain technical hypotheses the corresponding slow manifold will exhibit normal stability on a similar timescale.

\begin{theorem}[a normal stability estimate]\label{nse}
Let $\epsilon\,\dot{y} = f_\epsilon(x,y)$, $\dot{x} = g_\epsilon(x,y)$ be a fast-slow system and suppose the graph of $\widetilde{y}_\epsilon$ is a slow manifold of order $N$. Also suppose that the second derivatives $D^2_yf_\epsilon(x,y), D^2_yg_\epsilon(x,y)$ are bounded uniformly in $(\epsilon,x,y)$, and the functions $\widetilde{y}_\epsilon(x),D\widetilde{y}_\epsilon(x),r_\epsilon(x)$ are bounded uniformly in $(\epsilon,x)$. 

Let $(x(t),y(t))$ be a solution of this fast-slow system with initial condition $(x,\widetilde{y}_\epsilon(x))$. If there are some $\epsilon$-independent constants $C,\nu_0,k\geq0$ with $\frac{N-1}{2}\geq k$ such that all solutions $\widetilde{\xi}(t)$ of the non-autonomous linear system
\begin{align}
\epsilon\,\dot{\widetilde{\xi}} &= D_yf_\epsilon(x,\widetilde{y}_\epsilon)[\widetilde{\xi}] - \epsilon\,D\widetilde{y}_\epsilon(x)[D_yg_\epsilon(x,\widetilde{y}_\epsilon)[\widetilde{\xi}]],\label{stability_eqn_thm}
\end{align}
satisfy 
\begin{align}
||\widetilde{\xi}(t)||\leq C\exp(\epsilon^k \nu_0 |t-t_0|)||\widetilde{\xi}(t_0)||,\label{stability_assumption}
\end{align}
then there are $\epsilon$-independent positive constants $E_0,\alpha_0$ such that 
\begin{align}
||y(T) - \widetilde{y}_\epsilon(x(T))||\leq \epsilon^{\frac{N+1}{2}}E_0
\end{align}
for all $T\in[0,\frac{\alpha_0}{\epsilon^k\nu_0}]$. That is, the solution $(x(t),y(t))$ will stay within $O(\epsilon^{\frac{N+1}{2}})$ of the slow manifold $\widetilde{y}_\epsilon$ on an $O(\epsilon^{-k})$ time interval for all sufficiently large $N$.

\end{theorem}

\begin{proof}
According to  Proposition \ref{nf_prop} we may apply the change of dependent variables $(x,y)\mapsto (x,\xi)$, where $\xi = y -\widetilde{y}_\epsilon(x)$ is the normal deviation from the slow manifold, in order to write the fast-slow system in the form \eqref{nf_thm_one}-\eqref{nf_thm_two}. In terms of these dependent variables the initial conditions are $(x(0),\xi(0)) = (x,0)$, and the solution of interest is $(x(t),\xi(t))$.

Let $U_{t,t_0}:Y\rightarrow Y$ be the propagator for the linear system \eqref{stability_eqn_thm} so that
\begin{align}
\epsilon\,\partial_tU_{t,t_0} &= D_yf_\epsilon(x,\widetilde{y}_\epsilon)[U_{t,t_0}]\nonumber\\
& - \epsilon\,D\widetilde{y}_\epsilon(x)[D_yg_\epsilon(x,\widetilde{y}_\epsilon)[U_{t,t_0}]],
\end{align}
and $U_{t_0,t_0}= \text{id}$.
Using $U_{t,t_0}$ we may introduce the variation-of-parameters ansatz $\xi(t) = U_{t,0}\zeta(t)$, where $\zeta(0) = 0$. Because $\dot{\xi} = (\partial_tU_{t,0})\zeta(t) + U_{t,0}\dot{\zeta}(t)$, $\zeta(t)$ must satisfy the nonlinear equation
\begin{align}
\epsilon\,\dot{\zeta} = U_{0,t}\mathcal{N}_\epsilon(x,U_{t,0}\zeta) + \epsilon^{N+1}\,U_{0,t}r_\epsilon(x),
\end{align}
which may also be written as the nonlinear integral equation 
\begin{align}
\epsilon\,\zeta(T) &= \int_0^TU_{0,t}\mathcal{N}_\epsilon(x,U_{t,0}\zeta)\,dt \nonumber\\
&+ \epsilon^{N+1}\int_0^TU_{0,t}r_\epsilon(x)\,dt.\label{integral_eqn_prf}
\end{align}

We may now use the integral equation \eqref{integral_eqn_prf} together with hypothesis \eqref{stability_assumption} to estimate the size of $||\zeta(t)||$ as follows. First we take the norm of each side of \eqref{integral_eqn_prf} and apply the triangle inequality to obtain
\begin{align}
\epsilon\,||\zeta(T)||&\leq \int_0^T ||U_{0,t}\mathcal{N}_\epsilon(x,U_{t,0}\zeta)||\,dt\nonumber\\
&+ \epsilon^{N+1}\int_0^T||U_{0,t}r_\epsilon(x)||\,dt.\label{intermediate_inequality}
\end{align}
Next we use hypothesis \eqref{stability_assumption} and the uniform bounds on $D^2_yf_\epsilon,D^2_yg_\epsilon,r_\epsilon$ to establish the estimates
\begin{align}
 ||U_{0,t}\mathcal{N}_\epsilon(x,U_{t,0}\zeta)||&\leq C \,\exp\left(\epsilon^k\,\nu_0t\right) || \mathcal{N}_\epsilon(x,U_{t,0}\zeta)||\nonumber\\
 &\leq C \,N_0 \,\exp\left(\epsilon^k\,\nu_0t\right)\,||U_{t,0}\zeta(t)||^2\nonumber\\
 &\leq C^3\,N_0\,\exp\left(3\,\epsilon^k\,\nu_0\,t\right)\,|| \zeta(t) ||^2
\end{align}
and 
\begin{align}
|| U_{0,t}r_\epsilon(x)||\leq C\,\exp\left(3\,\epsilon^k\,\nu_0\,t\right)\,R_0
\end{align}
for $t>0$. Here $R_0\geq0$ is a uniform bound on $||r_\epsilon(x)||$ and  $N_0\geq0$ may be expressed in terms of the uniform bounds $||D^2_yf_\epsilon(x.y)||\leq F_0$, $||D^2_yg_\epsilon(x,y)||\leq G_0$, and $||D\widetilde{y}_\epsilon(x)||\leq Y_0$ as
\begin{align}
N_0 = F_0 + \epsilon\,Y_0\,G_0.
\end{align}
Combining these estimates with the inequality \eqref{intermediate_inequality} then implies
\begin{align}
\epsilon\,||\zeta(T)||&\leq C^3\,N_0\,\int_0^T \exp(3\,\epsilon^k\,\nu_0\,t) ||\zeta(t)||^2\,dt\nonumber\\
& + \epsilon^{N+1}\,C\,R_0\int_0^T\exp(3\,\epsilon^k\,\nu_0\,t)\,dt.\label{basic_integral_inequality}
\end{align}

We now use the Bihari-LaSalle inequality to bound $||\zeta(T)||$ by the solution of the integral \emph{equality} corresponding to Eq.\,\eqref{basic_integral_inequality}. Namely, $||\zeta(T)|| \leq z(t)$ where $z(t)$ is the solution of the ordinary differential equation
\begin{align}
\epsilon \dot{z} = \exp(3\,\epsilon^k\,\nu_0\,t) (C^3 \,N_0 z^2 + \epsilon^{N+1}\,C\,R_0),
\end{align}
with initial condition $z(0)=0$. As is readily verified, an explicit expression for $z(T)$ is given by
\begin{align}
z(T) &= \frac{1}{C}\sqrt{\frac{R_0}{N_0}}\epsilon^{\frac{N+1}{2}}\nonumber\\
&\times\text{tan}\left(\frac{\epsilon^{\frac{N-1}{2}-k}C^2 \sqrt{R_0N_0}}{3\nu_0}[\exp(3\epsilon^k \nu_0 T)-1]\right).
\end{align}
Using $\text{tan}\,x\leq \frac{1}{\pi/2 - x}$ for $x\in[0,\pi/2)$, we therefore obtain the estimate
\begin{align}
||\zeta(T)||\leq \frac{\frac{1}{C}\sqrt{\frac{R_0}{N_0}}\epsilon^{\frac{N+1}{2}}}{\frac{\pi}{2} - \frac{\epsilon^{\frac{N-1}{2}-k}C^2\sqrt{R_0N_0}}{3\nu_0}[\exp(3\epsilon^k\nu_0T)-1]}.\label{zeta_ineq}
\end{align}
Using $||\xi(T)|| = ||U_{T,0}\zeta(T)||\leq C\exp(\epsilon^k\nu_0T)||\zeta(T)||\leq C\exp(3\epsilon^k\nu_0T)||\zeta(T)||$, the inequality \eqref{zeta_ineq} then implies 
\begin{align}
||\xi(T)||\leq \frac{\sqrt{\frac{R_0}{N_0}}\epsilon^{\frac{N+1}{2}}\exp(3\epsilon^k\nu_0T)}{\frac{\pi}{2} - \frac{\epsilon^{\frac{N-1}{2}-k}C^2\sqrt{R_0N_0}}{3\nu_0}[\exp(3\epsilon^k\nu_0T)-1]}.\label{xi_ineq}
\end{align}
Because the right-hand-side of Eq.\,\eqref{xi_ineq} is monotonically increasing as a function of $T$ for each $E_0> \frac{2}{\pi}\sqrt{\frac{R_0}{N_0}}$ the norm $||\xi(T)||$ must satisfy the inequality 
\begin{align}
||\xi(T)||\leq \epsilon^{\frac{N+1}{2}} E_0\label{desired_ineq}
\end{align}
for $0\leq T\leq T^*$, where $T^*$ is the solution of the equation
\begin{align}
E_0 = \frac{\sqrt{\frac{R_0}{N_0}}\exp(3\epsilon^k\nu_0T^*)}{\frac{\pi}{2} - \frac{\epsilon^{\frac{N-1}{2}-k}C^2\sqrt{R_0N_0}}{3\nu_0}[\exp(3\epsilon^k\nu_0T^*)-1]}.
\end{align}
Simple algebraic manipulations finally show that when $\frac{N-1}{2}\geq k$ we have $T^*\geq \alpha_0/(\nu_0\epsilon^k)$, where
\begin{align}
\alpha_0 = \frac{1}{3}\text{ln}\left(\frac{\frac{\pi}{2}E_0}{\sqrt{\frac{R_0}{N_0}}+\frac{C^2 E_0\sqrt{R_0N_0}}{3\nu_0}}\right),
\end{align} which completes the proof.

\end{proof}

\begin{remark}
Note that Theorem \ref{nse} does not require uniform bounds on the \emph{first derivative} $D_yf_\epsilon(x,y)$. Avoiding assumptions on the boundedness of $D_yf_\epsilon(x,y)$ is important in infinite dimensions because in many interesting problems $D_yf_\epsilon$ is an unbounded operator like the Laplacian. As discussed in Ref.\,\citep{Kristiansen_2016}, the boundedness of $D_yf_\epsilon(x,y)$ is also irrelevant to establishing uniform bounds on the remainder term $r(x)$ in the normal form. 

\end{remark}

\begin{example}
Consider the zero-drag Abraham-Lorentz dynamics introduced in Section \ref{zero_drag_sec}. We have already shown that this system admits a fast-slow split (Example \ref{FS_example_hard}) and computed the associated formal slow manifold $y_\epsilon^* = y_0^* +\epsilon\,y_1^* + \dots$ up to and including the $O(\epsilon^2)$ coefficient $y_2^*$. (Example \ref{zero_drag_sm_example}) We will now apply the ideas behind Theorem \ref{nse} in order to provide a rigorous estimate of this system's normal stability properties. We will revisit this question in Section \ref{normal_ellipticity_section}, where will use the theory of adiabatic invariants to refine our estimate.

Let $\widetilde{y}_\epsilon = y_0^* + \epsilon\,y_1^*$ be the first-order slow manifold given by naive power series truncation. Recall that $y_0^* = 0$ and $y_1^*$ was shown in Eq.\,\eqref{curv_drift_slave} to be given by the well-known curvature drift velocity. 
Given a solution $(x(t),y(t))$ of the zero-drag equations with initial condition $(x(0),y(0)) = (x_0,\widetilde{y}_\epsilon^*(x_0))$, we would like to estimate the timescale over which the normal deviation from the first-order slow manifold $\xi(t) = y(t) - \widetilde{y}_\epsilon(x(t)) = (\xi_1(t),\xi_2(2))$ remains small. (Note that initially the deviation is $0$.)

As a first step, we will establish a stability estimate for the linearized normal dynamics defined by Eq.\,\eqref{stability_eqn} with initial condition $\widetilde{\xi}(t_0) = \widetilde{\xi}_0$. Remember that this equation is non-autonomous due to the dependence on $x(t)$ of the linear operators on the right-hand-side. Setting $\widetilde{\xi} = (\widetilde{\xi}_1,\widetilde{\xi}_2)$, we may rewrite \eqref{stability_eqn} in the form
\begin{align}
\epsilon\,\frac{d}{dt}\begin{pmatrix}\widetilde{\xi}_1\\ \widetilde{\xi}_2\end{pmatrix}
=\,\zeta\,|\bm{B}(\bm{x}(t))|\begin{pmatrix}
0 & 1\\
-1 & 0
\end{pmatrix}\begin{pmatrix}
\widetilde{\xi}_1\\
\widetilde{\xi}_2
\end{pmatrix} +\epsilon\, \mathbf{L}_\epsilon(t)\begin{pmatrix}
\widetilde{\xi}_1\\
\widetilde{\xi}_2
\end{pmatrix},\label{example_ls}
\end{align} 
where the time-dependent $2\times2$ matrix $\mathbf{L}_\epsilon(t)$ is given by
\begin{align}
\mathbf{L}_\epsilon(t) = D_yf_1(x,\widetilde{y}) - D\widetilde{y}_\epsilon(x)\cdot D_y g_0(x,\widetilde{y}_\epsilon).
\end{align}
While we leave finding the explicit form of $\mathbf{L}_\epsilon(t)$ as an exercise for the reader, we do note that $\mathbf{L}_\epsilon(t)$ satisfies the following (time-independent) inequality for each two-component vector $\widetilde{\xi}$:
\begin{align}
|\mathbf{L}_\epsilon(t)\widetilde{\xi}|\leq \frac{\sqrt{2E}}{L_1}Q\left(\epsilon\,\frac{\sqrt{2E}}{L_1 B_\text{min}}\right)|\widetilde{\xi}|.\label{example_estimate_L}
\end{align}
Here $E$ is the energy of the solution $(x(t),y(t))$, $L_1^{-1}$ is an upper bound on the matrix norm $||\nabla\bm{b}||$, $B_{\text{min}}$ is the minimum value of the magnetic field strength, and $Q$ is a second-order polynomial with positive coefficients that depend at most on an upper bound for the gradient of the field line curvature  $L_2^{-2}\geq|\nabla(\bm{b}\cdot\nabla\bm{b})|$.
While the linear system \eqref{example_ls} cannot be solved in closed form, we may obtain a useful estimate of the type \eqref{stability_assumption} needed in Theorem \ref{nse} by introducing the variation of parameters ansatz $\widetilde{\xi}(t) = W_{t,t_0}\widetilde{\zeta}(t)$, where the matrix $W_{t,t_0}$ is given by
\begin{align}
W_{t,t_0} = &
\begin{pmatrix}
\cos\varphi(t) & -\sin\varphi(t)\\
\sin\varphi(t) & \cos\varphi(t)\\
\end{pmatrix}\\
\varphi(t) = & - \zeta \int_{t_0}^{t}|\bm{B}(\bm{x}(t))|\,dt.
\end{align}
The significance of the orthogonal matrix $W_{t,t_0}$ is that if it were the case that $\mathbf{L}_\epsilon(t)  =0$, then the new dependent variable $\widetilde{\zeta}(t)$ would be constant in time as a result of the simple identity 
\begin{align}
\partial_tW_{t,t_0} = \zeta |\bm{B}(\bm{x}(t))|\,\begin{pmatrix}
0 & 1\\
-1 & 0
\end{pmatrix}W_{t,t_0}.
\end{align}
Because $\mathbf{L}_\epsilon(t)$ is actually not zero, substituting the variation of parameters ansatz into the linearized normal stability equation \eqref{example_ls} leads to the nontrivial evolution law for $\widetilde{\zeta}$:
\begin{align}
\dot{\widetilde{\zeta}}(t) = W_{t,t_0}^T\mathbf{L}_\epsilon(t)W_{t,t_0}\widetilde{\zeta}(t),
\end{align}
with initial condition $\widetilde{\zeta}(t_0) =\widetilde{\xi}_0$. This evolution law may be equivalently written as the integral equation
\begin{align}
\widetilde{\zeta}(T) = \widetilde{\xi}_0 + \int_{t_0}^{T}W_{t,t_0}^T\mathbf{L}_\epsilon(t)W_{t,t_0}\widetilde{\zeta}(t)\,dt,
\end{align}
for $T\geq t_0$,which implies the integral inequality
\begin{align}
|\widetilde{\zeta}(T)|&\leq |\widetilde{\xi}_0| + \int_{t_0}^{T}|W_{t,t_0}^T\mathbf{L}_\epsilon(t)W_{t,t_0}\widetilde{\zeta}(t)|\,dt\nonumber\\
&\leq |\widetilde{\xi}_0| + \int_{t_0}^{T} |\mathbf{L}_\epsilon(t)W_{t,t_0}\widetilde{\zeta}(t)|\,dt\nonumber\\
&\leq |\widetilde{\xi}_0| + \int_{t_0}^{T} \frac{\sqrt{2E}}{L_1}Q\left(\epsilon\,\frac{\sqrt{2E}}{L_1 B_\text{min}}\right)|W_{t,t_0}\widetilde{\zeta}(t)|\,dt\nonumber\\
&\leq |\widetilde{\xi}_0| + \int_{t_0}^{T} \frac{\sqrt{2E}}{L_1}Q\left(\epsilon\,\frac{\sqrt{2E}}{L_1 B_\text{min}}\right)|\widetilde{\zeta}(t)|\,dt,
\end{align}
where we have applied the triangle inequality, the orthogonality of $W_{t,t_0}$, and the inequality \eqref{example_estimate_L}. We finally invoke Gronwall's inequality and the time reversal symmetry of our argument so far to conclude
\begin{align}
|\widetilde{\xi}(t)| &= |\widetilde{\zeta}(t)|\nonumber\\
&\leq \exp\left(\frac{\sqrt{2E}}{L_1}Q\left(\epsilon\frac{2E}{L_1 \,B_{\text{min}}}\right)|t-t_0|\right)|\widetilde{\xi}(0)|,\label{example_lnse}
\end{align}
which says that this slow manifold is linearly normally stable on a timescale comparable to the particle transit time $\tau_{\text{transit}}\sim L_1/\sqrt{2E}$. Note that this timescale is significantly longer (by a factor of $1/\epsilon$) than the cyclotron timescale.

Now we will use Theorem \ref{nse} to upgrade the linearized normal stability estimate \eqref{example_lnse} into a nonlinear normal stability estimate. In other words, instead of bounding the growth of a general solution $\widetilde{\xi}(t)$ of the linear system \eqref{stability_eqn}, we will now bound the deviation $\xi(t)$ of the specific solution $(x(t),y(t))$ of the nonlinear fast-slow system from the first-order slow manifold $\widetilde{y}_\epsilon$. All that we have to do is check that the hypotheses in the statement of Theorem \eqref{nse} are satisfied. Apparently the parameters $N,k,\nu_0$ are given by 
\begin{align}
N &= 1\,\text{ (order of the slow manifold)}\\
k & = 0\,\text{ (temporal exponent)}\\
\nu_0 & = \frac{\sqrt{2E}}{L_1} Q\left(\epsilon\frac{2E}{L_1 \,B_{\text{min}}}\right)\,\text{ (linear stability timescale)}.
\end{align}
Assuming for simplicity that all derivatives of the magnetic field are uniformly bounded, it is straightforward to verify that the second derivatives $D^2f_\epsilon$, $D^2g_\epsilon$ and the quantities $\widetilde{y}_\epsilon(x),D\widetilde{y}_{\epsilon}(x), r_\epsilon(x)$ are uniformly bounded as well. The Theorem therefore allows us to conclude that any trajectory of the zero-drag equations with $(x(0),y(0)) = (x_0,\widetilde{y}_\epsilon(x))$ will remain within $O(\epsilon)$ of the slow manifold for an $O(1)$ interval of time. A similar argument starting from an order $N>1$ slow manifold leads to the same conclusion but with an $O(\epsilon^{(N+1)/2})$ error bound. Therefore if we would like to generate solutions that stay on the formal slow manifold with $O(\epsilon^M)$ accuracy the Theorem implies that we should initialize those trajectories on a slow manifold of order $2M-1$. Note in particular that using this argument increasing the accuracy of the slow manifold does not generally lead to an improved normal stability timescale. 

\end{example}

\subsection{Normal hyperbolicity}
While the argument based on Gronwall's inequality used in Theorem \ref{nse} may be applied to the slow manifold of any fast-slow system, for certain special types of slow manifolds more specialized techniques are available to study normal stability. In particular there is a powerful theory due to Fenichel \citep{Fenichel_1979} available for analyzing the so-called normally-hyperbolic slow manifolds. Such slow manifolds arise frequently in fast-slow systems that exhibit dissipation, instability, or a combination thereof on the $O(\epsilon)$ timescale. Fenichel's theory is rigorous in general only in finite dimensions, but does suggest strategies for analysis in infinite dimensions as well.

We say that a finite-dimensional fast-slow system's slow manifold is \emph{normally hyperbolic} if for each $x\in X$ the eigenvalues of the derivative $D_yf_0(x,y_0^*(x))$ all have non-zero real part. This terminology is best understood after re-scaling time according to $t = \epsilon\,\tau$ so that the fast-slow system becomes $dy/d\tau = f_\epsilon(x,y)$, $dx/d\tau = \epsilon\,g_\epsilon(x,y)$. In the limit $\epsilon\rightarrow 0$ dynamics on this short timescale are described by the relatively simple system $dy/d\tau = f_0(x,y)$, $dx/d\tau = 0$, which shows that the slow variable $x$ is frozen and that the fast variable $y$ has an equilibrium point $y = y_0^*(x)$ for each $x$. Moreover the linearization of $y$-dynamics about these equilibrium points is described by the family of ordinary differential equations $d\,\delta y/d\tau = D_yf_0(x,y_0^*(x))[\delta y] $. Generally speaking an equilibrium point of a dynamical system whose associated linearized dynamics is characterized by eigenvalues with non-zero real parts is referred to as a hyperbolic fixed point. Apparently the limit of a normally-hyperbolic slow manifold comprises a family of hyperbolic fixed points parameterized by $x\in X$ for the fast-slow system's limiting short timescale dynamics. Therefore dynamics in the vicinity of a normally-hyperbolic slow manifold tend to be rapidly attractive in some directions and rapidly repulsive in others.

Perhaps the most remarkable feature of normally-hyperbolic slow manifolds is their resilience. As is true of all slow manifolds the $\epsilon\rightarrow 0$ limit of a normally-hyperbolic slow manifold is a true invariant set for the limiting short-time dynamics. But while in general the limiting slow manifold only persists as a formal slow manifold $y_\epsilon^*$ when $\epsilon$ moves slightly above zero a normally-hyperbolic slow manifold survives the transition to finite-$\epsilon$ (finite but sufficiently small) as a genuine invariant set. Moreover dynamics on this invariant set are slow and dynamics in the vicinity of the invariant set resemble that of the $\epsilon\rightarrow 0$ limit. One might think that this all means the formal power series $y_\epsilon^*$ converges to a function whose graph is equal to the normally-hyperbolic slow manifold, but this is not correct in general. While a true invariant manifold does exist, its dependence on $\epsilon$ is only $C^\infty$ if the right-hand-side of the fast-slow system is $C^\infty$ in $(x,y,\epsilon)$. (See \citep{Jones_1995} for a detailed discussion of the regularity properties for normally-hyperbolic slow manifolds.) The actual relationship between the invariant set and the formal slow manifold is the following: the Taylor coefficients in $\epsilon$ of the invariant manifold are given precisely by the coefficients in the formal power series $y_\epsilon^*$.

In the special case where the eigenvalues of $Df_0(x,y_0^*(x))$ all have strictly-negative real parts the normal stability argument given in Theorem \ref{nse} may be modified (essentially by allowing for a negative time constant $\nu_0$) to show that this special type of normally-hyperbolic slow manifold will often be normally stable over unbounded time intervals. This suggests the presence of a true invariant set that attracts nearby trajectories. However this type of reasoning does not reveal the essence of what makes normally-hyperbolic slow manifolds special. In fact the Gronwall-type argument is incapable of even suggesting the presence of a true invariant manifold when eigenvalues with positive real parts (i.e. normal instabilities) coexist with eigenvalues with negative real parts. While the details of the proof of persistence of normally-hyperbolic slow manifolds lie beyond the scope of this Review we refer interested readers to \citep{Fenichel_1979} and \citep{Jones_1995} for detailed expositions.

A striking implication of the persistence of normally-hyperbolic slow manifolds in general is the persistence of normally-hyperbolic slow manifolds with \emph{strictly-unstable} normal dynamics, i.e. linearized normal dynamics characterized by eigenvalues of $D_yf_0(x,y_0^*(x))$ with strictly positive real parts. Such slow manifolds are repulsive to nearby trajectories, but nevertheless there is in principle an exact way to choose initial conditions that do not excite the rapidly-growing normal modes. This has interesting implications for models of charged particles interacting with their own field.

In the weakly-relativistic regime seemingly careful derivations of the charged particle equations of motion lead to the Abraham-Lorentz model that we have referred to frequently in this Review. However this model has an apparent physical flaw related to the presence of so-called ``runaway solutions" characterized by a particle's acceleration escaping to infinity. Spohn in \citep{Spohn_2000} recognized that this runaway behavior may be identified with the rapid repulsion of a normally-hyperbolic slow manifold with purely-unstable normal dynamics. This instability appeared explicitly in Example \ref{weak_drag_SM_example}, where we saw that the dynamics normal to the Abraham-Lorentz slow manifold are characterized by the equation $\delta\dot{\bm{a}} = \frac{3}{2}\delta\bm{a}$, which has the single eigenvalue $\lambda = 3/2$ with algebraic multiplicity $3$. Spohn also recognized that Fenichel's theory on persistence of normally-hyperbolic slow manifolds, particularly of the unstable variety, implied the existence of a true invariant set close to the limiting slow manifold $\bm{a}_0^* = \zeta \bm{v}\times\bm{B}(\bm{x})$ that is necessarily free of the bothersome runaway solutions. The physical implication of Spohn's observations is that all of the physical content of the Abraham-Lorentz model is contained in the reduction of that model to its exact slow manifold. This gives a partial rigorous justification of the ``order reduction" technique proposed by Landau and Lifshitz \citep{Landau_fields_1975} for eliminating runaway solutions; to first order in the small parameter $\epsilon$ dynamics on the true Abraham-Lorentz slow manifold are the same as the Landau-Lifshitz equations. It also points to the limitations inherent to the order reduction method -- the exact expression for the equations of motion on the Abraham-Lorentz slow manifold do in fact differ slightly from the Landau-Lifshitz equation and this difference cannot be computed in explicit form in general. It is therefore an interesting challenge to consider practical methods for numerically simulating the ``full" Abraham-Lorentz slow manifold reduction and determining whether there are important effects contained in the full reduction that cannot be captured by low-order truncations like the Landau-Lifshitz equations.

\subsection{Normal ellipticity\label{normal_ellipticity_section}}
We say that a finite-dimensional fast-slow system's slow manifold is \emph{normally elliptic} if for each $x\in X$ the derivative $D_yf_0(x,y_0^*(x))$ has purely imaginary eigenvalues and is diagonalizable over the field of complex numbers. As was true of normal hyperbolicity this terminology is best understood after re-scaling time according to $t = \epsilon\,\tau$ so that the fast-slow system becomes $dy/d\tau = f_\epsilon(x,y)$, $dx/d\tau = \epsilon\,g_\epsilon(x,y)$. In the limit $\epsilon\rightarrow 0$ dynamics on this short timescale are described by the relatively simple system $dy/d\tau = f_0(x,y)$, $dx/d\tau = 0$, which shows that the slow variable $x$ is frozen and that the fast variable $y$ has an equilibrium point $y = y_0^*(x)$ for each $x$. Moreover the linearization of $y$-dynamics about these equilibrium points is described by the family of ordinary differential equations $d\,\delta y/d\tau = D_yf_0(x,y_0^*(x))[\delta y] $. Generally speaking an equilibrium point of a dynamical system whose associated linearized dynamics has purely imaginary eigenvalues is referred to as an elliptic fixed point. Apparently the limit of a normally elliptic slow manifold comprises a family of elliptic fixed points parameterized by $x\in X$ for the fast-slow system's limiting short timescale dynamics. Therefore dynamics in the vicinity of a normally-elliptic slow manifold tend to rapidly oscillate in the normal directions, at least on sufficiently short timescales.

Normally-elliptic slow manifolds are fragile, in stark contrast to the normally-hyperbolic variety. While the leading-order approximation of such a slow manifold $y_0^*$ is a true invariant manifold for the limiting short timescale dynamics, the formal power series $y_\epsilon^*$ usually is not the Taylor expansion of a true invariant manifold for finite-$\epsilon$ dynamics. This means that trajectories with initial conditions near a normally-elliptic slow manifold may eventually wander into regions of $(x,y)$-space where dynamics become fast. In this Section we will discuss one mechanism, resonance, that can produce such wandering and another mechanism, adiabatic invariance, that can suppress wandering over long intervals of time. Rather than develop these concepts in full generality we will illustrate them using a pair of examples.

\begin{example}\label{resonance_example}
Unlike normally-hyperbolic slow manifolds, normally-elliptic slow manifolds may not survive the transition from $\epsilon=0$ to finite but small $\epsilon$. In his Review MacKay \citep{MacKay_2004} gives a clear example of this phenomenon where resonances between fast normal oscillations and the slower oscillations tangent to the slow manifold cause the ``exact" slow manifold to blow up. (We refer the reader to Example 2 in the last reference.) On the other hand for practical purposes the question of existence of an exact slow manifold may be less relevant than the question of normal stability of approximate slow manifolds. If trajectories near an approximate slow manifold happen to stay near that approximate slow manifold over a long enough time interval then a reduced model based on slow manifold reduction may be applicable even when a nearby exact slow manifold does not exist. Conversely if those same trajectories diverge from the approximate slow manifold on a short enough timescale then the slow manifold reduction approach needs to be handled with considerable care.

It may be tempting to suppose that normally-elliptic slow manifolds automatically have exceptional normal stability properties because on the $O(\epsilon)$ timescale these objects are neutrally stable. And it is indeed usually the case that the ``sticking time" for a normally-elliptic slow manifold tends to be at least $O(1)$, meaning that trajectories near an approximate normally-elliptic slow manifold will experience many oscillation periods in the vicinity of the slow manifold. However resonance effects may still sometimes spoil normal stability over time intervals that are $O(\epsilon^{-k})$ for $k>0$. For example consider the following fast-slow coupled oscillator system,
\begin{align}
\dot{p}_1 &= -\frac{\Omega^2}{\epsilon}\left(1 + \frac{\epsilon \kappa^2}{2}q_2^2 - \frac{\epsilon\kappa^2}{2 \Omega^2} p_2^2\right)q_1 +\kappa^2 p_1p_2 q_2\label{neg_eng_one}\\
\dot{q}_1 & = \frac{1}{\epsilon}\left(1+\frac{\epsilon \kappa^2}{2\Omega^2} p_2^2 - \frac{\epsilon\kappa^2}{2}q_2^2\right) p_1 - \kappa^2 q_1p_2q_2\\
\dot{p}_2 & = \frac{\Omega^2}{\epsilon}\left(1 - \frac{\epsilon \kappa^2}{2}q_1^2 + \frac{\epsilon\kappa^2}{2\Omega^2}p_1^2\right)q_2 + \kappa^2 p_1q_1p_2\\
\dot{q}_2 & = -\frac{1}{\epsilon}\left(1 - \frac{\epsilon\kappa^2}{2\Omega^2}p_1^2 + \frac{\epsilon\kappa^2}{2}q_1^2\right)p_2 - \kappa^2 p_1 q_1 q_2\\
\dot{p}_3 &= - \omega^2 q_3\\
\dot{q}_3  & = p_3,\label{neg_eng_six}
\end{align}
where $\Omega,\omega,\kappa>0$ are real positive constants.
The fast variable $y = (q_1,p_1,q_2,p_2)$, the slow variable $x=(q_3,p_3)$, and it is easy to check that this system has an exact normally-elliptic slow manifold given by the slaving function $y^*(x) = (0,0,0,0)$. (Non-existence of a slow manifold is a non-issue here.) While trajectories of this system initially contained in the graph of $y^*$ will remain within that graph for all time, some trajectories that begin within $O(\epsilon)$ of that graph attain an $O(1)$ normal deviation after an $O(\epsilon^{-1})$ time interval. To see this change variables twice. First introduce the oscillator action-angle variables $p_i = -\Omega^{1/2}(2J_i)^{1/2}\sin\theta_i$, $q_i = \Omega^{-1/2}(2J_i)^{1/2}\cos\theta_i$ for $i = 1,2$ so that the system transforms into
\begin{align}
\dot{J}_1 &= 2\kappa^2J_1J_2\sin(2\theta_1+2\theta_2)\\
\dot{\theta}_1& = \frac{\Omega}{\epsilon} + \kappa^2 J_2 \cos(2\theta_1+2\theta_2)\\
\dot{J}_2 & = 2\kappa^2 J_1 J_2\sin(2\theta_1+2\theta_2)\\
\dot{\theta}_2 & = -\frac{\Omega}{\epsilon} + \kappa^2 J_1 \cos(2\theta_1+2\theta_2)\\
\dot{p}_3 &= - \omega^2 q_3\\
\dot{q}_3  & = p_3,
\end{align}
which clearly displays the leading-order resonance between the $(q_1,p_1)$ and $(q_2,p_2)$ oscillators.
Then introduce variables adapted to the resonance, $J  = J_1 - J_2$, $I = J_2 $, $\theta = \theta_1$, $\zeta = \theta_1+\theta_2$, to transform the system a second time into
\begin{align}
\dot{J} &= 0\label{first_res_vars_one}\\
\dot{\theta} & = \frac{\Omega}{\epsilon}+ \kappa^2 I \cos 2\zeta\\
\dot{I} & = 2\kappa^2 (J+I) I \sin 2\zeta\\
\dot{\zeta} & = \kappa^2 (J+2I) \cos 2\zeta\label{first_res_vars_four}\\
\dot{p}_3 &= - \omega_3^2 q_3\\
\dot{q}_3 &= p_3,\label{first_res_vars_six}
\end{align}
which highlights the fact that while $\theta_1$ and $\theta_2$ evolve rapidly individually the resonant combination $\zeta = \theta_1+\theta_2$ evolves slowly.
When written in the form \eqref{first_res_vars_one}-\eqref{first_res_vars_six} it is straightforward to verify that this fast-slow system has three independent first integrals, the action difference $J$, the $(q_3,p_3)$ oscillator energy 
\begin{align}
H_3 = \frac{1}{2}p_3^2 + \frac{1}{2}\omega^2q_3^2,
\end{align}
and the total $(q_1,p_1)$-$(q_2,p_2)$ oscillator energy
\begin{align}
H_{12} &= \Omega J + \epsilon \kappa^2 (J+I)I\cos2\zeta.
\end{align}
Note in particular that the $(q_2,p_2)$ oscillator is an example of a so-called ``negative energy mode," \citep{Morrison_neg_modes_1989} and that $h = (H_{12} - \Omega J )/\epsilon=  \kappa^2 (J+I)I\cos 2\zeta$ is also a constant of motion.  Trajectories contained in the graph of $y^*$ have $J = J_1 - J_2 = 0-0 = 0$ and $I = J_2 = 0$. Therefore a particular class of trajectories that begin within an $O(\epsilon)$ neighborhood of the slow manifold is characterized by initial conditions $J_1(0) = \epsilon\,\gamma$, $J_2(0) = \epsilon \gamma$, $\theta_1(0)= 0$, $\theta_2(0) = 0$, or $J = 0$, $I(0) = \epsilon \,\gamma$, $\theta(0) = 0$, $\zeta(0) = 0$, where $\gamma >0$ is any positive constant. For these trajectories the integral $h =  \epsilon^2 \kappa^2 \gamma^2>0 $, which implies that the values of $I$ and $\zeta$ along these trajectories must be related by
\begin{align}
I = \frac{\epsilon \gamma}{\sqrt{\cos2\zeta}}.\label{I_graph}
\end{align}
Apparently as $\zeta$ approaches $\pi/4$ the action $I = J_2$ approaches infinity, which already suggests a normal instability. In order to be more quantitative suppose that at time $T$ the value of $I$ has changed from $I(0) = \epsilon \gamma$ to $I(T) =\infty$, which corresponds to a trajectory initially within $O(\epsilon)$ of the slow manifold achieving an infinite normal deviation. According to the relationship \eqref{I_graph} the value of $\zeta$ when $t = T$ must be $\zeta(T) = \pi/4$. Then according to Eq.\,\eqref{first_res_vars_four} $T$ must be given by the formula
\begin{align}
T = \frac{1}{2 \kappa^2 \epsilon \gamma}\int_0^{\pi/4}\frac{d\zeta}{\sqrt{\cos 2\zeta}}= \frac{\sqrt{2}K(1/2)}{4\kappa^2 \epsilon \gamma}\approx \frac{.66}{\kappa^2 \epsilon\gamma},
\end{align}
where $K$ is the complete elliptic integral of the first kind. Thus each member of this family of trajectories wanders infinitely far from the normally-elliptic slow manifold after a finite time $T = O(1/\epsilon)$. Note that this explosive normal instability is an inherently nonlinear effect brought on by the resonant coupling; the linearized normal dynamics are perfectly stable. 

\end{example}\label{mu_example}

\begin{example}\label{adiabatic_invariance_example}
While resonances may destabilize (approximate) normally-elliptic slow manifolds over sufficiently long timescales, another mechanism sometimes help to mitigate this destabilization in Hamiltonian fast-slow systems --- adiabatic invariance. We will illustrate this phenomenon with a few toy examples before applying our observations to the normally-elliptic slow manifold contained in zero-drag Abraham-Lorentz dynamics.

Consider first the $3$ degree-of-freedom Hamiltonian oscillator system with symplectic form $\bm{\Omega}_\epsilon =  \epsilon \,dq_1 \wedge dp_2 + \epsilon \,dq_2 \wedge dp_2 + dq_3 \wedge dp_3$ and Hamiltonian 
\begin{align}
H &= \frac{1}{2} p_1^2 + \frac{1}{2}\Omega^2 q_1^2 + \frac{1}{2}p_2^2 + \frac{1}{2}\Omega^2 q_2^2 + \frac{1}{2}p_3^2 + \frac{1}{2}\omega^2 q_3^2\nonumber\\
& + \frac{\epsilon \kappa^2}{4}\left(\Omega^{-2}p_1^2 p_2^2 + \Omega^2 q_1^2 q_2^2 - p_1^2 q_2^2 - q_1^2 p_2^2 + 4 p_1 q_1 p_2 q_2 \right),\label{pos_en_ham}
\end{align}
where $\Omega,\omega,\kappa >0$ are positive real constants. Due to the placement of $\epsilon$ in the symplectic form Hamilton's equations for this system read $\dot{p}_i = -\epsilon^{-1}\,\partial_{q_i}H$, $\dot{q}_i = \epsilon^{-1}\,\partial_{p_i}H$ for $i=1,2$ and $\dot{p}_3 = -\partial_{q_3}H$, $\dot{q}_3 = \partial_{p_3}H$. More explicitly we have
\begin{align}
\dot{p}_1 &= -\frac{\Omega^2}{\epsilon}\left(1 + \frac{\epsilon \kappa^2}{2}q_2^2 - \frac{\epsilon\kappa^2}{2 \Omega^2} p_2^2\right)q_1 -\kappa^2 p_1p_2 q_2\\
\dot{q}_1 & = \frac{1}{\epsilon}\left(1+\frac{\epsilon \kappa^2}{2\Omega^2} p_2^2 - \frac{\epsilon\kappa^2}{2}q_2^2\right) p_1 + \kappa^2 q_1p_2q_2\\
\dot{p}_2 & = -\frac{\Omega^2}{\epsilon}\left(1 + \frac{\epsilon \kappa^2}{2}q_1^2 - \frac{\epsilon\kappa^2}{2\Omega^2}p_1^2\right)q_2 - \kappa^2 p_1q_1p_2\\
\dot{q}_2 & = \frac{1}{\epsilon}\left(1 + \frac{\epsilon\kappa^2}{2\Omega^2}p_1^2 - \frac{\epsilon\kappa^2}{2}q_1^2\right)p_2 + \kappa^2 p_1 q_1 q_2\\
\dot{p}_3 &= - \omega^2 q_3\\
\dot{q}_3  & = p_3.
\end{align}
This system has an exact normally-elliptic slow manifold $(q_1,p_1,q_2,p_2) = (0,0,0,0)$ and does suffer from a resonant perturbation because in the unperturbed action-angle variables $p_i = -\Omega^{1/2}(2J_i)^{1/2}\sin\theta_i$, $q_i = \Omega^{-1/2}(2J_i)^{1/2}\cos\theta_i$, $i = 1,2$, the Hamiltonian becomes
\begin{align}
H = \Omega (J_1 + J_2) + \frac{1}{2}p_3^2 + \frac{1}{2}\omega^2 q_3^2  + \epsilon\,\kappa^2 J_1 J_2 \cos(2\theta_1 - 2\theta_2),\label{pos_en_ham_aa}
\end{align}
which clearly displays the $(2,-2)$ resonance in the perturbation. However, in contrast to the situation from the previous example, the resonance in this system does not lead to a destabilization of the slow manifold. 

The normal stability of this system's slow manifold can be seen as follows. Let $H_{12} = H - \frac{1}{2}p_3^2 - \frac{1}{2}\omega^2 q_3^2$ be the $(q_1,p_1)$-$(q_2,p_2)$ oscillator energy. Note that $H_{12}$ does not depend on $(q_3,p_3)$. It is simple to verify that the Hessian of $H_{12}$ at $(q_1,p_1,q_2,p_2) = (0,0,0,0)$ is positive definite. It follows then from the Morse lemma (see for example  \citep{Abraham_2008} Section 3.2) that there is an open neighborhood $U$ of $(0,0,0,0)$ in $(q_1,p_1,q_2,p_2)$-space and coordinates $(x,y,z,w)$ on $U$ in which the function $H_{12}$ takes the form $H_{12} = \frac{1}{2}x^2 + \frac{1}{2}y^2 + \frac{1}{2}z^2 + \frac{1}{2} w^2$. Because $H_{12}$ is a first integral this implies that the norm defined by $||(x,y,z,w)||_H = \sqrt{H_{12}}$ is constant in time for any trajectory beginning within $U$. In particular any trajectory that begins within a radius-$\epsilon^N$ ball defined by the norm $||\cdot||_H$, where $N$ is any positive integer, will remain within that ball for all time. In other words the slow manifold in this system is nonlinearly normally stable.

The key difference between the stable slow manifold in this example and the unstable slow manifold in Example \ref{resonance_example} is that normal oscillations about the stable slow manifold did not contain negative energy modes while normal oscillations about the unstable slow manifold did. This observation has only limited utility however because our two examples have the unreasonably special property that the fast variable $y$ decouples completely from the slow variable $x$. How much of the dichotomy that we just exhibited is due to this degeneracy, and how much is not?

As a first step in answering this question let us modify the Hamiltonian \eqref{pos_en_ham} by making $\kappa^2$ a positive function of $q_3$ instead of a constant, e.g. $\kappa^2(q_3) = \kappa_0^2 (1+q_3^2/2)$. This modification does not change the form of the $(q_1,p_1,q_2,p_2)$ evolution equations, but it does substantially modify the form of the $(q_3,p_3)$ equations according to 
\begin{align}
\dot{p}_3 &= -\omega^2 q_3 - \frac{\epsilon}{4} \frac{d\kappa^2}{dq_3}\left(\Omega^{-2}p_1^2 p_2^2 + \Omega^2 q_1^2 q_2^2 - p_1^2 q_2^2 - q_1^2 p_2^2 + 4 p_1 q_1 p_2 q_2 \right)\\
\dot{q}_3 & = p_3.
\end{align}
Apparently this modification leads to non-trivial coupling between the fast and slow variables. In particular it is no longer the case that $H_{12}=H - \frac{1}{2}p_3^2 - \frac{1}{2}\omega^2 q_3^2$ is a constant of motion because energy can slosh between all three oscillators. This implies that the Lyapunov stability argument given earlier can not be used to deduce normal stability of the slow manifold. What can be done?

In spite of the nontrivial fast-slow coupling introduced by a $q_3$-dependent $\kappa^2$ this system still possesses a symmetry that was present in the constant-$\kappa$ case. This is most clear from the expression for the Hamiltonian in terms of unperturbed action-angle variables given in Eq.\,\eqref{pos_en_ham_aa}. (Note that this expression is still valid with a $q_3$-dependent $\kappa^2$.) If the angle variables $(\theta_1,\theta_2)$ are subject to the transformation $(\theta_1,\theta_2)\mapsto (\theta_1 + \Delta \theta, \theta_2 + \Delta \theta)$, where $\Delta\theta$ is any real constant, the Hamiltonian does not change. Noether's theorem for Hamiltonian systems therefore implies that the quantity $J = J_1 + J_2 = \Omega^{-1} (p_1^2 + \Omega^2 q_1^2)/2 + \Omega^{-1}(p_2^2 + \Omega^2 q_2^2)/2$ is a constant of motion. Fortuitously the action sum $J$ with $q_3$-dependent $\kappa$ has exactly the properties that $H_{12}$ had in the constant-$\kappa$ case that allowed the Lyapunov stability argument to function. We conclude that our $\kappa$-coupling does not destabilize the slow manifold. The explanation for stability has changed, however, from an argument based on negative energy modes to an argument based on symmetry.

Arguments based on symmetry, or more generally \emph{approximate symmetry}, are perhaps the most powerful tools for investigating stability properties of normally-elliptic slow manifolds in Hamiltonian fast-slow systems. As a demonstration consider modifying the Hamiltonian \eqref{pos_en_ham} again by adding a small symmetry-breaking perturbation. More specifically, use $\mathcal{H} = H + \epsilon\,\delta H$ as a Hamiltonian, where $H$ is as before (with $q_3$-dependent $\kappa$) and  $\delta H = p_1 \pi(q_3)$ where $\pi(q_3)$ is some smooth function of $q_3$. This modification is more significant than allowing for a non-constant $\kappa$ because the set $(q_1,p_1,q_2,p_2)=(0,0,0,0)$ is no longer a true invariant manifold; it is likely this this system \emph{does not} have an exact slow manifold. Moreover because in action-angle variables the perturbation $\delta H = -\pi(q_3) \,\Omega^{1/2}(2J_1)^{1/2}\sin\theta_1 $ the symmetry $(\theta_1,\theta_2)\mapsto (\theta_1 + \Delta \theta, \theta_2 + \Delta \theta)$ is broken. Therefore neither of our previous arguments for establishing normal stability apply to this system. We will nevertheless show that because the symmetry-breaking term is non-resonant the system still possesses an approximate symmetry. We will also show that this approximate symmetry is useful for analyzing stability properties of the slow manifold. 

We begin by transforming once again into the resonance-adapted coordinates $(J,\theta,I,\zeta)$ introduced in Example \ref{resonance_example}. In these coordinates the Hamiltonian $\mathcal{H}$ takes the form
\begin{align}
\mathcal{H} = \Omega J + \frac{1}{2}p_3^2 + \frac{1}{2}\omega^2 q_3^2 + \epsilon\, \kappa^2 (J + I) I \cos 2\zeta - \epsilon \pi (2\Omega)^{1/2}(J+I)^{1/2}\sin\theta.
\end{align}
Next we introduce a near-identity canonical change of coordinates that is equal to the $t = 1$ flow of a Hamiltonian system with Hamiltonian $K_2 = -\epsilon^2 \pi (2/\Omega)^{1/2} (J+I)^{1/2}\cos\theta$. In these coordinates $(\overline{J},\overline{\theta},\overline{I},\overline{\zeta},\overline{q}_3,\overline{p}_3)$ the Hamiltonian becomes
\begin{align}
\mathcal{H} &= \Omega \overline{J} + \frac{1}{2}\overline{p}_3^2 + \frac{1}{2}\omega^2 \overline{q}_3^2 + \epsilon\,\kappa^2(\overline{q}_3) (\overline{J}+\overline{I})\overline{I} \cos 2\overline{\zeta} - \epsilon^2\frac{1}{2}\pi^2(\overline{q}_3)\nonumber\\
& - \epsilon^2 p_3 \frac{d\pi}{d\overline{q}_3}\sqrt{\frac{2}{\Omega}}(\overline{J}+\overline{I})^{1/2}\cos\overline{\theta} + \epsilon^2 \pi(\overline{q}_3)\kappa^2(\overline{q}_3) \sqrt{\frac{2}{\Omega}}(\overline{J}+\overline{I})^{1/2}\overline{I}\sin(\overline{\theta} - 2\overline{\zeta})\nonumber\\
& +O(\epsilon^3).
\end{align}
Note that this transformation pushes the symmetry-breaking terms to one higher order in $\epsilon$ in the Hamiltonian. Finally we apply one more near-identity canonical transformation that averages over the $\theta$-dependence in the second-order Hamiltonian in order to push the symmetry-breaking terms to $O(\epsilon^3)$. In this last coordinate system $(\overline{\overline{J}},\overline{\overline{\theta}},\overline{\overline{I}},\overline{\overline{\zeta}},\overline{\overline{q}}_3,\overline{\overline{p}}_3)$ the Hamiltonian becomes
\begin{align}
\mathcal{H} = \Omega\overline{\overline{J}} + \frac{1}{2}\overline{\overline{p}}_3^2 + \frac{1}{2}\omega^2 \overline{\overline{q}}_3^2 + \epsilon\,\overline{\overline{\kappa}}^2(\overline{\overline{J}}+ \overline{\overline{I}})\overline{\overline{I}}\cos 2\overline{\overline{\zeta}} - \epsilon^2\frac{1}{2}\overline{\overline{\pi}}^2 + O(\epsilon^3).
\end{align}
The Hamiltonian is therefore invariant under rotations of $\overline{\overline{\theta}}$ up to third order in $\epsilon$. In fact through the application of more near-identity canonical transformations the symmetry-breaking terms can be pushed to as high an order in $\epsilon$ as desired. However even without the application of additional near-identity coordinate transformations we can already see that $\dot{\overline{\overline{J}}} = O(\epsilon^2)$. This implies that if $\overline{\overline{J}}=0$ when $t=0$ then it will remain $O(\epsilon)$ over an $O(1/\epsilon)$ time interval, and will require at least $O(1/\epsilon^2)$ seconds in order to attain a value that is $O(1)$. In other words $\overline{\overline{J}}$ is an adiabatic invariant. The significance of this observation is that the zero level set of $\overline{\overline{J}}$ may be shown to define a slow manifold, and the adiabatic invariance of $\overline{\overline{J}}$ implies directly that this slow manifold is normally-stable over $O(1/\epsilon)$ time intervals. (We invite readers to check this fact themselves; the demonstration uses the fact that $\overline{\overline{J}} = J + O(\epsilon)$ and that the normal Hessian of $J$ on the limiting slow manifold is positive definite. These observations allow one to apply a perturbative version of the Lyapunov stability argument.)

More generally the presence of adiabatic invariants provides normally-elliptic slow manifolds with better stability properties. An interesting physical implication of this fact relates to the zero-drag Abraham-Lorentz equations. It is an old and venerable result that the (drag-free) motion of a charged particle in a (smooth) strong magnetic field possesses an adiabatic invariant $\mu = \mu_0 + \epsilon\,\mu_1 + \epsilon^2\,\mu_2 + \dots$ to all-orders in $\epsilon$. (That all of the coefficients in the formal power series for $\mu$ exist follows from the fact that this system possesses only a single fast angle. More generally resonances may cause the series to break down at some order. See Ref.\,\citep{Burby_gc_2013} for explicit expressions for $\mu_k$ with $k\leq 2$.) Concretely this means that if $\mu^{(N)} = \sum_{k = 1}^N\mu_k$ is an order-$N$ approximation of $\mu$ then $d\mu^{(N)}/dt = O(\epsilon^{N+1})$. Therefore trajectories that begin on the zero level set of $\mu^{(N)}$ require at least $O(1/\epsilon^{N+1})$ seconds to attain $\mu^{(N)} = O(1)$. This establishes a strong normal stability result for the zero-drag slow manifold because we have already mentioned that this slow manifold is coincident with the zero level set of the adiabatic invariant. (Note that this implication again uses the fact that $\mu = \mu_0 + O(\epsilon)$ and that the normal Hessian of $\mu_0$ along the limiting slow manifold is positive definite. )
\end{example}

\section{Slow Manifolds and Implicit Numerical Integrators\label{SMI_sec}}
A peculiar feature of plasmas is the variety of spatio-temporal scales that they support. While this variety is an essential part of the rich and exploitable phenomenology of plasma physics, it also makes the simulation of systems such as magnetic fusion reactors, accretion discs, and inertially-confined plasmas a daunting task in general. As a result two general trends have emerged in the field of plasma simulations that aim to cope with the inherent stiffness of multi-scale plasma dynamics in different ways. 

First, and perhaps most traditionally, there is the approach based on analytical model reduction. This approach involves first developing theoretical insight into the stiffness inherent to the plasma process of interest, usually through a combination of experimental observations and simple scaling arguments. Using this insight asymptotic methods such as slow manifold reduction, or even ad-hoc tricks, may sometimes be applied to develop reduced plasma models that are free of stiffness, and yet capable of capturing the essence of the physics. Once a non-stiff reduced model has been identified efficient simulations may become much easier to identify. In the area of strongly-magnetized plasma physics perhaps the most striking example of the model-reduction approach to addressing stiffness is embodied by simulations of gyrokinetics \citep{Taylor_1967,Frieman_1982,Beer_1995,Sugama_2000,Brizard_POP_2000,Brizard_2007,Schekochihin_2009,Scott_2010,Abel_2013}, which is a reduced model of magnetized plasma microturbulence that eliminates the need to resolve cyclotron timescale dynamics.

The second general approach to simulating stiff plasma dynamics is based on implicit integration schemes. Here the basic idea is to exploit the fact that implicit integrators tend to have more robust numerical stability properties than their explicit counterparts. In the most favorable cases implicit schemes can successfully ``step over" the extremely short timescales responsible for a problem's stiffness, thereby achieving similar computational savings as in the model reduction approach.  A number of recent examples of how to successfully apply this approach to coping with stiffness have been produced by the Applied Math and Plasma Physics group at Los Alamos National Laboratory, especially by Chac\'on and collaborators. See for examples Refs. \citep{Taitano_2018,Chacon_HOLO_2017,Chacon_Stanier_2016,Chacon_Chen_2016}. 

There are shortcomings associated with each of these general approaches to simulating stiff dynamics. 
\begin{enumerate}
	\item Reduced models are typically formulated using scaling limits, and therefore are known only as asymptotic series. This means that physical effects captured by higher-order terms in the expansions are extremely difficult to account for. Moreover, models derived as scaling limits are known to include certain unphysical scaling artifacts (spurious instabilities, stiffness, etc.) that may easily corrupt a simulation. An example of such a scaling artifact is the high-$\bm{k}$ behavior of the whistler wave dispersion relation \citep{Farmer_2019} in Hall magnetohydrodynamics.  
	\item Current approaches to designing efficient implicit integration schemes, including their preconditioners, suffer from a lack of general guiding mathematical principles. For example a generally challenging aspect of the HOLO approach to implicit integration reviewed in Ref. \citep{Chacon_HOLO_2017} is identifying a useful high-order (HO) low-order (LO) split of the system variables. In the best case a HO-LO split of the problem will simultaneously ease the large memory requirements associated with the HO variables by way of nonlinear elimination and replace the stiff update for the whole system with a non-stiff update for just the LO variables. However state-of-the-art approaches to finding such useful HOLO splits are based on problem-dependent intuition.
	\item It is difficult to know \emph{a priori} that the short timescales being stepped-over by a large implicit step are sufficiently \emph{inactive} to justify stepping over them.
\end{enumerate}

The purpose of this Section is to describe how slow manifold theory naturally leads to a third approach to simulating temporally-stiff dynamics that retains most of the benefits and avoids most of the pitfalls of the previous two approaches. The rough idea is to use the analytical insight into a problem's stiffness afforded by a fast-slow split to systematically ease the technical difficulties associated with implicit integration. We refer to the implementation of this idea as \emph{slow manifold integration} (SMI).

The theoretical foundation for SMI rests on the notion of discrete-time fast-slow systems.
\begin{definition}[Discrete-time fast-slow system]
Let $X,Y,\Gamma$ be finite-dimensional vector spaces. A \emph{discrete-time fast-slow system with parameter} $\gamma\in \Gamma$ is a family of smooth mappings $\Phi_{\gamma}:X\times Y\rightarrow X\times Y$ that depends smoothly on $\gamma$ near $0\in \Gamma$ with the following properties.
\begin{itemize}
\item[(1)] There is a family of smooth mappings $\psi_x:Y\rightarrow Y$ smoothly parameterized by $x\in X$ such that $\Phi_0(x,y) = (x,\psi_x(y))$.
\item[(2)] For each $x\in X$ the mapping $\psi_x$ has a unique non-degenerate fixed point $y=y_0^*(x)$. In particular $D\psi_x(y_0^*(x)) - \text{id}_Y$ is an invertible linear map $Y\rightarrow Y$.
\end{itemize}
\end{definition}

\begin{example}\label{projection_DTFS}
Perhaps the simplest example of a discrete-time fast-slow system has the trivial parameter space $\Gamma = \{0\}$ and is defined by the map $\Phi(x,y) = (x,0)$. Apparently in this case we have $\psi_x(y) = 0$, which has the unique fixed point $y=y_0^*(x) = 0$. This fixed point is non-degenerate because $D\psi_x(0) = 0$, which implies $D\psi_x(y_0^*(x)) - \text{id}_Y = -\text{id}_Y$ is invertible.
\end{example}

The theory of discrete-time fast-slow systems mimics the theory of continuous-time fast-slow system in a number of respects. For example where invariant graphs for continuous-time fast-slow systems obey the invariance equation, invariant graphs for discrete-time fast-slow systems obey the following discrete-time analogue of the invariance equation.

\begin{proposition}\label{DT_inv_eq_prop}
If $\Phi_{\gamma} = (G_\gamma,F_\gamma):(x,y)\mapsto (G_\gamma(x,y),F_\gamma(x,y))$ is a discrete-time fast-slow system with parameter $\gamma$ and there is a smooth $\gamma$-dependent function $y_\gamma^*:X\rightarrow Y$ such that $S_\gamma = \{(x,y)\in X\times Y\mid y = y^*_\gamma(x)\}$ is an invariant manifold for $\Phi_\gamma$ then $y_\gamma^*$ is a solution of the \emph{discrete-time invariance equation}  
\begin{align}
F_\gamma(x,y_\gamma^*(x)) = y_\gamma^*(G_\gamma(x,y_\gamma^*(x))).\label{DT_inv_eq}
\end{align}
\end{proposition}

A more significant parallel between the the continuous- and discrete-time theories pertains to the existence of formal slow manifolds. In order to show this we first define a \emph{discrete-time formal slow manifold} to be a formal power series solution of the discrete-time invariance equation. This definition is slightly less trivial than the continuous-time case because we allow for vector valued parameter $\gamma$.
\begin{definition}[Discrete-time formal slow manifold]\label{DTFSM_def}
A \emph{discrete-time formal slow manifold} for a discrete-time fast-slow system with parameter $\gamma \in \Gamma$ is a formal power series $y_\gamma^*:X\rightarrow Y$ of the form 
\begin{align}
y_\gamma^*(x) = y_0^*(x) + y_1^*(x)[\gamma] + y_2^*(x)[\gamma,\gamma] + y_3^*(x)[\gamma,\gamma,\gamma] + \dots,
\end{align}
where the coefficient $y_k^*(x)$ is a symmetric $k$-linear function of $\gamma$ for each $k$, that solves the discrete-time invariance equation \eqref{DT_inv_eq} to all orders in $\gamma$.
\end{definition}

As the following Theorem shows, each discrete-time fast-slow system possesses a unique discrete-time formal slow manifold.

\begin{theorem}[existence and uniqueness of DT formal slow manifolds]\label{formal_DTSM_existence}
Associated with each discrete-time fast-slow system $\Phi_{\gamma} = (G_\gamma,F_\gamma)$ is precisely one discrete-time formal slow manifold. Moreover the lowest-order coefficients of the discrete-time formal slow manifold $y_\gamma^*$ are determined by the formulas
\begin{align}
\psi_x(y_0^*(x)) &= y_0^*(x)\\
y_1^*(x)[\gamma] & = \left[D\psi_x(y_0^*) - \mathrm{id}_Y\right]^{-1}\bigg[ Dy_0^*(x)[G_1(x,y_0^*)[\gamma]]  - F_1(x,y_0^*)[\gamma]\bigg],
\end{align}
where $F_1(x,y)[\gamma] = \partial_\gamma F_0(x,y)[\gamma]$ and similarly for $G_1$.
\end{theorem}

With the basic structural properties of discrete-time fast-slow systems in place we are now in a good position to define and discuss slow manifold integrators more precisely.
\begin{definition}[Slow manifold integrator]
A \emph{slow manifold integrator} for a fast slow system $\dot{x} = g_\epsilon(x,y)$, $\epsilon\,\dot{y} = f_\epsilon(x,y)$ is a two-parameter family of mappings $\Psi_{(h,\epsilon)}:X\times Y\rightarrow X\times Y$ with the following properties.
\begin{itemize}
\item[(1)] For fixed $\epsilon > 0$ and sufficiently small $h$ the mapping $\Psi_{(h,\epsilon)}$ approximates the time-$h$ flow of the fast-slow system. 
\item[(2)] The two-parameter family of mappings $\Phi_{(h,\delta)} = \Psi_{(h,h\delta)}$ is a discrete-time fast-slow system with parameter $\gamma = (h,\delta)\in\mathbb{R}^2$.
\end{itemize}
\end{definition}

\begin{remark}
The parameters $h$ and $\delta$ should be interpreted as the integrator's timestep and the ratio of $\epsilon$ to that timestep. When any integrator is used to step over the short $O(\epsilon)$ timescale the parameter $\delta $ is close to $0$. Property $(2)$ in the definition is therefore a constraint on the behavior of the integrator when the timestep is chosen to resolve the $O(1)$ timescale but not the $O(\epsilon)$ timescale, i.e. when $h$ is chosen to satisfy $\epsilon \ll h\ll 1$. In contrast property $(1)$ constrains the behavior of the integrator when the timestep is chosen to satisfy $h\ll \epsilon\ll 1$.
\end{remark}

Slow manifold integrators address the shortcomings mentioned earlier of both model reduction and generic implicit integration. This may be seen as follows.

Consider first the problems associated with model reduction. Because slow manifold integrators are by definition integrators for (unreduced) fast-slow systems, they automatically avoid any complications related to scaling artifacts that arise within reduced models. Moreover slow manifold integrators are all but forced to be implicit schemes, or at least implicit in the fast variable $y$. This follows from the fact that any integrator that is a slow manifold integrator must be well-defined when the timestep $h$ is much larger than $\epsilon$. Therefore slow manifold integrators will typically possess the numerical stability properties required to step over a fast-slow system's stiff timescale without the necessity of removing that stiffness analytically. We remark that this same benefit is shared with most other kinds of implicit integrators.

Next consider the preconditioning problem that plagues typical implicit integrators when taking large timesteps.
Because a slow manifold integrator in the large timestep regime $\epsilon\ll h\ll 1$ comprises a discrete-time fast slow system there is a powerful and general strategy for alleviating this problem. To see this suppose that a slow manifold integrator $(x_{n+1},y_{n+1}) = \Phi_{(h,\delta)}(x_n,y_n)$ is implicitly defined by the relations 
\begin{align}
A(x_{n+1},y_{n+1};x_n,y_n,h,\delta) &= 0\label{implicit_one}\\
B(x_{n+1},y_{n+1};x_n,y_n,h,\delta) &= 0,\label{implicit_two}
\end{align}
where $A$ takes values in $X$, $B$ takes values in $Y$, and both $A$ and $B$ are smooth in each of their arguments. When $(h,\delta) = (0,0)$ the solution of these implicit equations must be given by $y_{n+1} = \psi_{x_n}(y_{n})$, $x_{n+1} = x_n$. In other words the inverse of the mapping $C_{(x_n,y_n)}:(x_{n+1},y_{n+1})\mapsto (A(x_{n+1},y_{n+1};x_n,y_n,0,0),B(x_{n+1},y_{n+1};x_n,y_n,0,0))$ (assuming it exists in a neighborhood of the origin) applied to $(0,0)\in X\times Y$ must be given by
\begin{align}
C_{(x_n,y_n)}^{-1}(0,0) = (x_n,\psi_{x_n}(y_n)).
\end{align}
Therefore if we apply $C_{(x_n,y_n)}^{-1}$ to both sides of the implicit equations \eqref{implicit_one}-\eqref{implicit_two} and use the fundamental theorem of calculus we may identify the following alternative set of implicit equations defining the integrator:
\begin{align}
\begin{pmatrix}
x_{n+1}\\
y_{n+1}
\end{pmatrix} =
\begin{pmatrix}
x_{n}\\
\psi_{x_n}(y_n)
\end{pmatrix} & - h\int_0^1 DC_{(x_n,y_n)}^{-1}(A_\lambda,B_\lambda)\begin{pmatrix} (\partial_hA)_\lambda\\ (\partial_hB)_\lambda\end{pmatrix}\,d\lambda\nonumber\\
& - \delta \int_0^1 DC_{(x_n,y_n)}^{-1}(A_\lambda,B_\lambda)\begin{pmatrix} (\partial_\delta A)_\lambda\\ (\partial_\delta B)_\lambda\end{pmatrix}\,d\lambda,\label{better_implicit}
\end{align}
where $A_\lambda \equiv A(x_{n+1},y_{n+1};x_n,y_n,\lambda h,\lambda \delta)$ and $B_\lambda\equiv B(x_{n+1},y_{n+1};x_n,y_n,\lambda h,\lambda \delta)$. Remarkably this new implicit defining equation \eqref{better_implicit} is \emph{nearly explicit} because all of the implicit terms are multiplied by either $h$ or $\delta$, each of which are naturally small parameters in the large timestep regime $\epsilon\ll h\ll  1$. Therefore if the derivatives of the implicit terms are well behaved conventional iterative schemes for solving Eq.\,\eqref{better_implicit} such as Picard iteration will not require any further preconditioning in order to ensure rapid converge to the solution. Note furthermore that in order to obtain the nearly-explicit defining equation \eqref{better_implicit} the most complicated operation involved is computing the derivative matrix of $C^{-1}_{(x_n,y_n)}$ for fixed $(x_n,y_n)$. Because finding this derivative only involves analyzing the behavior of the original defining equations \eqref{implicit_one}-\eqref{implicit_two} with $(h,\delta)=(0,0)$, this step can often be performed by hand. 

Finally consider the issue of possibly-questionable accuracy of implicit integrators with large timesteps. Because a slow manifold integrator is required to be an accurate approximation of the fast-slow system's flow when all timescales are resolved by $h$ it is natural to expect that a \emph{single iteration} of a slow manifold integrator will be accurate when initial conditions are chosen close to the fast-slow system's formal slow manifold. Indeed, trajectories that begin near a slow manifold that is normally stable on an $O(1)$ timescale exhibit only small-amplitude fast dynamics over the course of a timestep $h$ that satisfies $\epsilon\ll h\ll 1$. Therefore a slow manifold integrator with such a timestep \emph{is} resolving all of the solution's timescales with a level of accuracy controlled by how close the initial condition is to the slow manifold. When considering the accuracy of a slow manifold integrator after \emph{many iterations} it is helpful to exploit the fact that by Theorem \ref{formal_DTSM_existence}  a slow manifold integrator always possesses its own discrete-time formal slow manifold. To the extent that this discrete-time formal slow manifold coincides with the continuous-time formal slow manifold similar arguments to those given in Section \ref{interpretation_SM} may be used to establish a discrete normal stability timescale. Over this stable timescale the slow manifold integrator will be guaranteed to keep a solution initially close to the slow manifold within a small neighborhood thereof. Because a single iteration of a slow manifold integrator should be accurate when applied to a phase point close to the slow manifold this in turn ensures that the local timestepping error will remain small over the discrete normal stability timescale.

\begin{example}\label{IM_example}
When implicit midpoint discretization is applied to a fast-slow system $\dot{x} = g_\epsilon(x,y)$, $\epsilon\,\dot{y} = f_\epsilon(x,y)$ the result is the following system of implicit evolution equations
\begin{align}
\frac{x_{n+1}-x_n}{h} =& g_\epsilon\left(\frac{x_n+x_{n+1}}{2},\frac{y_n + y_{n+1}}{2}\right)\label{IM_one}\\
\frac{y_{n+1}-y_n}{h} = &\frac{1}{\epsilon} f_\epsilon\left(\frac{x_n+x_{n+1}}{2},\frac{y_n + y_{n+1}}{2}\right).\label{IM_two}
\end{align}
The mapping $\Psi_{(h,\epsilon)}:(x_n,y_n)\mapsto (x_{n+1},y_{n+1})$ so defined comprises a fundamental example of a slow manifold integrator. To see that it is indeed a slow manifold integrator first rewrite the implicit defining equations \eqref{IM_one}-\eqref{IM_two} in terms of $\delta = \epsilon/h$ to obtain
\begin{align}
{x_{n+1}-x_n} =& h\,g_{h\delta}\left(\frac{x_n+x_{n+1}}{2},\frac{y_n + y_{n+1}}{2}\right)\label{IMdelta_one}\\
\delta(y_{n+1}-y_n) = & f_{h\delta}\left(\frac{x_n+x_{n+1}}{2},\frac{y_n + y_{n+1}}{2}\right).\label{IMdelta_two}
\end{align}
It follows that when $(h,\delta) = (0,0)$ the mapping $\Phi_{(h,\delta)} \equiv \Psi_{(h,h\delta)}$ is defined by the equations 
\begin{align}
x_{n+1} - x_n &= 0\\
0& = f_0\left(\frac{x_n+x_{n+1}}{2},\frac{y_n + y_{n+1}}{2}\right).
\end{align}
These equations may be solved explicitly in terms of the continuous-time limiting slow manifold $y_0^*$ as $x_{n+1} = x_n$ and $y_{n+1} = -y_n + 2 y_0^*(x_n)$. This shows that this integrator is a slow manifold integrator with $\psi_x(y) = -y + y_0^*(x)$. Note that because the fixed point of this $\psi_x$ is given by $y = y_0^*(x)$ an immediate corollary of this observation is that the discrete-time formal slow manifold agrees with the continuous-time formal slow manifold to leading order. We leave it as an exercise for the reader to check how closely the implicit-midpoint formal slow manifold approximates the continuous-time formal slow manifold.

Let us now analyze the single-timestep accuracy of this slow manifold integrator when $(h,\delta)$ is near $(0,0)$, i.e. in the timestep regime where $\epsilon\ll h\ll 1$. Note that while it is well known that the implicit midpoint method has formal second-order accuracy when $h\ll \epsilon\ll 1$, it is not immediately clear how the local approximation error scales in the regime $\epsilon\ll h \ll1$. In order to make progress we introduce three technical assumptions that ensure the derivatives of $f_\epsilon(x,y)$ and $g_\epsilon(x,y)$ with respect to $(x,y)$ do not vary too wildly with $(x,y,\epsilon)$. The most straightforward of these is the following.
\begin{itemize}
\item[(A1)] Each of the partial derivatives $D_xg_\epsilon(x,y),D_yg_\epsilon(x,y),D_xf_\epsilon(x,y),D_yf_\epsilon(x,y)$ is bounded in the appropriate induced norms uniformly in $(x,y,\epsilon)$ by positive constants $G_X^0$, $G_Y^0$, $F_X^0$, $F_Y^0$, respectively. The second derivatives are also uniformly bounded.
\end{itemize}
The second ensures that the separation between the $O(1)$ and $O(\epsilon)$ timescales is large uniformly in $(x,y)$,
\begin{itemize}
\item[(A3)] The derivative $  D_yf_\epsilon(x,y)$ is invertible for each $(x,y,\epsilon)$. Moreover there is an $(x,y,\epsilon)$-independent positive constant $\omega$ such that $|| [D_yf_\epsilon(x,y)]^{-1}||\leq 1/\omega$ for all $(x,y,\epsilon)$.
\end{itemize}
In this assumption $||\cdot||$ denotes the induced operator norm and the constant $\omega$ may be interpreted as $\epsilon$ time the slowest frequency characterizing dynamics normal to the slow manifold.
The third and most technical assumption places an even stronger constraint on the variability of $D_yf_\epsilon(x,y)$, namely
\begin{itemize}
\item[(A3)] There is a positive constant $\Gamma$ independent of $(x,y,\epsilon)$ such that
\begin{align}
\bigg|\bigg| \left(\int_0^1 [D_yf_\epsilon(x,y)]^{-1}D_yf_\epsilon(x+\lambda u,y+\lambda v)\,d\lambda\right)^{-1} \bigg|\bigg|\leq 1+\Gamma,
\end{align}
for all $(u,v)\in X\times Y$.

\end{itemize}
Note in particular that the assumption (A3) is satisfied with $\Gamma = 0$ if $D_yf_\epsilon(x,y)$ is independent of $(x,y)$. 

With assumptions (A1-3)  in mind, suppose that $(x_n,y_n)\in X\times Y$ is a given point in the fast-slow system's phase space. We would like to estimate the difference between the single-timestep evolution of $(x_n,y_n)$ predicted by implicit midpoint, $(x_{n+1},y_{n+1})$, and the time-$h$ evolution of $(x_n,y_n)$ provided by the true system dynamics, $(\hat{x}_{n+1},\hat{y}_{n+1})$. As we have already mentioned, we are interested in the regime $\epsilon\ll h\ll 1$, or $(h,\delta)\approx (0,0)$.  Let $(e_X,e_Y) = (x_{n+1} - \hat{x}_{n+1} ,y_{n+1} - \hat{y}_{n+1}) $ be the deviation whose size we would like to estimate. According to implicit midpoint $(x_{n+1},y_{n+1})$ is determined implicitly by
\begin{align}
\delta(y_{n+1} - y_n) &= f_{h\delta}\left(\frac{x_{n+1}+x_n}{2},\frac{y_{n+1}+y_n}{2}\right)\label{IM_a}\\
x_{n+1} - x_n & = h\,g_{h\delta}\left(\frac{x_{n+1}+x_n}{2},\frac{y_{n+1}+y_n}{2}\right).\label{IM_b}
\end{align}
By integrating the exact fast-slow system in time we see that $(\hat{x}_{n+1},\hat{y}_{n+1})$ satisfies the similar set of equations
\begin{align}
\delta(\hat{y}_{n+1} - y_n) &= \overline{f}\label{ex_a}\\
\hat{x}_{n+1} - x_n& = h\,\overline{g}\label{ex_b},
\end{align}
where $\overline{f} = h^{-1}\int_{0}^h f_{h\delta}(\hat{x}(t_0),\hat{y}(t_0))\,dt_0$ and $\overline{g} = h^{-1}\int_0^h g_{h\delta}(\hat{x}(t_0),\hat{y}(t_0))\,dt_0$ are the time-$h$ orbit averages of $f_\epsilon$ and $g_\epsilon$ along the unique solution of the fast-slow system $(\hat{x}(t),\hat{y}(t))$ with $(\hat{x}(0),\hat{y}(0)) = (x_n,y_n)$. Subtracting \eqref{ex_a} from \eqref{IM_a}, subtracting \eqref{ex_b} from \eqref{IM_b}, and applying the fundamental theorem of calculus therefore leads to the following system of equations satisfied by the deviation $(e_X,e_Y)$,
\begin{align}
\delta\, e_Y &= \underline{f} - \overline{f} + \frac{1}{2} F_X[e_X] + \frac{1}{2} F_Y[e_Y]\label{linear_error_one}\\
e_X& = h\,\underline{g} - h\,\overline{g} + \frac{h}{2} G_X[e_X] + \frac{h}{2} G_Y[e_Y],\label{linear_error_two}
\end{align}
where we have introduced the convenient shorthand
\begin{align}
Q_Z &= \int_0^1 D_zq_{h\delta}\left(\frac{\hat{x}_{n+1} + x_n}{2} + \lambda\frac{e_X}{2},\frac{\hat{y}_{n+1} + y_n}{2} + \lambda\frac{e_Y}{2}\right)\,d\lambda\\
\underline{q} & = q_{h\delta}\left(\frac{\hat{x}_{n+1} + x_n}{2} ,\frac{\hat{y}_{n+1} + y_n}{2}\right)
\end{align}
with $Q \in\{F,G\}$, $Z \in\{X,Y\}$. We would like to solve Eqs.\,\eqref{linear_error_one}-\eqref{linear_error_two} for $(e_X,e_Y)$ as functions of the residuals $R_f = \underline{f} - \overline{f}$ and $R_g = \underline{g} - \overline{g}$ so that the norm of $(e_X,e_Y)$ may be estimated in terms of the norm of $(R_g,R_f)$. Before this can be done however we must establish invertibility of the linear map $L = \delta\,\text{id}_Y - \frac{1}{2} F_Y$, which we turn to next.

To see that $L$ is invertible first observe that $\underline{\Omega}^{-1}F_X$, with $\underline{\Omega} = D_yf_\epsilon(\hat{x}_{n+1}/2 +x_n/2,\hat{y}_{n+1}/2 +y_n/2 )$, must be invertible by assumption (A3) with $(u,v) = (e_X,e_Y)$ and $(x,y) = (\hat{x}_{n+1}+x_n,\hat{y}_{n+1}+y_n)/2$. It follows immediately that $F_X$ is invertible. Next consider the parameter-dependent linear map $L_\lambda = \lambda\,\delta\,\text{id}_Y - \frac{1}{2} F_Y $. We have just established that $L_0$ is invertible. It follows then from the openness of the set of invertible matrices that $L_\lambda$ must also be invertible for $\lambda$ sufficiently close to zero. In fact we may use assumption (A2) to bound the norm of $L_\lambda^{-1}$ as follows. Write $L_\lambda = -\frac{1}{2}F_X (1 - 2\,\lambda \,\delta\,F_X^{-1})$ and define $U_\lambda = (1 - 2\,\lambda \,\delta\,F_X^{-1})^{-1}$. Because $L_\lambda^{-1} = -2\,U_\lambda F_X^{-1}$ the induced norm of $L_\lambda^{-1}$ is bounded by the norm of $U_\lambda$ according to 
\begin{align}
||L_\lambda^{-1}||\leq 2||F_X^{-1}||\,||U_\lambda||\leq 2 ||(\underline{\Omega}^{-1}F_X)^{-1}\underline{\Omega}^{-1}||\,|| U_\lambda||\leq \frac{2 (1+\Gamma)}{\omega}||U_\lambda||,
\end{align}
we we note in passing that we have applied assumptions (A2) and (A3). The norm of $U_\lambda$ may be bound in turn by first recognizing that $U_\lambda$ solves the ordinary differential equation $\partial_\lambda U_\lambda = - U_\lambda A U_\lambda$ with $A = - 2\delta \,F_X^{-1}$, and then applying a variant of the Gronwall inequality. This leads to the estimate 
\begin{align}
||U_\lambda ||\leq\frac{1}{1-\lambda ||A||} \leq \frac{1}{1 - \frac{2\lambda\,\delta\,(1+\Gamma)}{\omega}},
\end{align}
which shows that $L_\lambda $ is invertible when $\lambda < \omega/(2\delta(1+\Gamma))$. In particular $L = L_1$ is invertible and $L^{-1}$ is bounded according to 
\begin{align}
|| L^{-1}|| \leq \frac{\frac{2 (1+\Gamma)}{\omega}}{1 - \delta \frac{2 (1+\Gamma)}{\omega}}\label{L_mat_bound}
\end{align}
as long as $\delta < \omega/(2(1+\Gamma))$.

We may now return to the task of solving Eqs.\,\eqref{linear_error_one}-\eqref{linear_error_two} for $(e_X,e_Y)$ as functions of the residuals $(R_g,R_f)$. Because $L = \delta\,\text{id}_Y - \frac{1}{2} F_Y $ is invertible Eq.\,\eqref{linear_error_one} implies that $e_Y$ may be expressed in terms of $e_Y$ and $R_f$ as
\begin{align}
e_Y = L^{-1}R_f + \frac{1}{2} L^{-1}F_X[e_X].\label{proto_ey}
\end{align}
Substituting this result into Eq.\,\eqref{linear_error_two} then shows that $e_X$ must satisfy
\begin{align}
e_X = h R_g + \frac{h}{2}G_X[e_X] + \frac{h}{2}G_YL^{-1}R_f +  \frac{h}{4}G_YL^{-1}F_X[e_X].
\end{align}
This equation may be used to solve for $e_X$ provided the linear map $\mathcal{I} = 1 - \frac{h}{2}G_X - \frac{h}{4} G_Y L^{-1}F_X$ is invertible. Resorting to assumption (A1), and repeating the argument given in the previous paragraph, it is straightforward to show that the invertibility of $\mathcal{I}$ is guaranteed and the norm of $\mathcal{I}^{-1}$ is bounded by
\begin{align}
||\mathcal{I}^{-1}||\leq \left(1 - \frac{h}{2}G_X^0 - \frac{h}{4}G_Y^0F_X^0\frac{\frac{2 (1+\Gamma)}{\omega}}{1 - \delta \frac{2 (1+\Gamma)}{\omega}} \right)^{-1}\label{I_mat_bound}
\end{align}
provided the timestep is chosen sufficiently small to ensure
\begin{align}
h < \left(\frac{1}{2}G_X^0 +\frac{1}{4}G_Y^0F_X^0\frac{\frac{2 (1+\Gamma)}{\omega}}{1 - \delta \frac{2 (1+\Gamma)}{\omega}}\right)^{-1}.
\end{align}
Therefore $e_X$ may be written in terms of $(R_g,R_f)$ as
\begin{align}
e_X = h \mathcal{I}^{-1}R_g + \frac{h}{2}\mathcal{I}^{-1}G_YL^{-1}R_f,\label{ex_formula}
\end{align}
and by Eq.\,\eqref{proto_ey} the fast-variable error $e_Y$ may be written in similar terms as
\begin{align}
e_Y =   \frac{h}{2}L^{-1}F_X\mathcal{I}^{-1}R_g +L^{-1}R_f+ \frac{h}{4}L^{-1}F_X\mathcal{I}^{-1}G_YL^{-1}R_f.\label{ey_formula}
\end{align}

In order to complete the derivation of our single-step error bound in the regime $\epsilon\ll h\ll 1$ we must now bound the residuals $R_f$ and $R_g$. We only consider $R_f$; bounding $R_g$ will follow exactly the same argument. By definition $R_f = \underline{f} - \overline{f}$, where $\underline{f}$ is $f_{h\delta}$ evaluated at the mean of the orbit's endpoints and $\overline{f}$ is the orbit average of $f_{h\delta}$. If $f_M$ denotes the midpoint quadrature approximation of the orbit average of $f_{h\delta}$ we may write $R_f = \underline{f} - f_M + f_M - \overline{f}$. Therefore the norm of the residual bounded by $||R_f|| \leq || \underline{f}-f_M || + || f_M-\overline{f} ||$. The second term is bounded by the usual midpoint quadrature error estimate, $ || f_M-\overline{f} ||\leq C_0\,\delta\,h^3 || \hat{y}^{\prime\prime} ||_{[0,h]}$, where $C_0$ is an $O(1)$ positive constant and $|| \hat{y}^{\prime\prime} ||_{[0,h]}\equiv \text{sup}_{t\in[0,h]}|| \hat{y}^{\prime\prime}(t)||$. The first term may be bounded by applying the fundamental theorem of calculus twice to estimate the difference between the mean of the orbit endpoints and the orbit at the midpoint, resulting in $|| \underline{f}-f_M || \leq \frac{h^2}{8} F_X^0 || \hat{x}^{\prime\prime}||_{[0,h]} + \frac{h^2}{8} F_Y^0 || \hat{y}^{\prime\prime} ||_{[0,h]}$. The fast residual is therefore bounded according to 
\begin{align}
||R_f|| \leq C_0 \delta h^3  || \hat{y}^{\prime\prime\prime} ||_{[0,h]}  + \frac{h^2}{8} F_X^0 || \hat{x}^{\prime\prime}||_{[0,h]} + \frac{h^2}{8} F_Y^0 || \hat{y}^{\prime\prime} ||_{[0,h]}.\label{f_residual}
\end{align}
Applying the same reasoning to $R_g$ leads to the bound
\begin{align}
||R_g|| \leq C_0  h^2  || \hat{x}^{\prime\prime\prime} ||_{[0,h]}  + \frac{h^2}{8} G_X^0 || \hat{x}^{\prime\prime}||_{[0,h]} + \frac{h^2}{8} G_Y^0 || \hat{y}^{\prime\prime} ||_{[0,h]}.\label{g_residual}
\end{align}
Finally we gather the bounds \eqref{f_residual}-\eqref{g_residual}, \eqref{I_mat_bound}, and \eqref{L_mat_bound}, and apply the triangle inequality to the formulas \eqref{ex_formula}-\eqref{ey_formula} to obtain the following explicit \emph{a priori} error bound for a single implicit midpoint step with $\epsilon\ll h\ll 1$:
\begin{align}
||e_Y||\leq \alpha_1\, \delta h^3 || \hat{y}^{\prime\prime\prime} ||_{[0,h]} + \alpha_2\, h^3 || \hat{x}^{\prime\prime\prime} ||_{[0,h]} + \alpha_3\, h^2\,|| \hat{x}^{\prime\prime}||_{[0,h]}+\alpha_4\,h^2\, || \hat{y}^{\prime\prime} ||_{[0,h]}\label{fast_var_estimate_IM}\\
||e_X||\leq \beta_1\,\delta\,h^4  || \hat{y}^{\prime\prime\prime} ||_{[0,h]} + \beta_2\,h^3 || \hat{x}^{\prime\prime\prime} ||_{[0,h]}  + \beta_3\,h^3 || \hat{x}^{\prime\prime}||_{[0,h]} + \beta_4\,h^3|| \hat{y}^{\prime\prime} ||_{[0,h]},\label{slow_var_estimate_IM}
\end{align}
where the $\alpha_k,\beta_k$ are each $O(1)$ positive constants that may be extracted from the above formulas in straightforward fashion.

The error estimates \eqref{fast_var_estimate_IM}-\eqref{slow_var_estimate_IM} demonstrate several notable features of the implicit midpoint method in the large timestep regime. First they suggest how crude an approximation the method is for solutions away from the slow manifold. Such solutions have large time derivatives, $\hat{y}^{(n)}=O(1/(h\delta)^n)$, which implies that the right-hand-sides of \eqref{fast_var_estimate_IM}-\eqref{slow_var_estimate_IM} are not necessarily small for $(h,\delta)$ near zero. On the other hand if $(x_n,y_n)$ is sufficiently close to the fast-slow system's formal slow manifold and the slow manifold is normally stable over the time interval $[0,h]$ the bounds are much more useful. By the zero derivative principle explained in Section \ref{zdp_sec} solutions initialized sufficiently close to the slow manifold will have $O(1)$ time derivatives. In particular if $(x_n,y_n)$ is sufficiently close to the slow manifold to ensure that the second and third time derivatives are $O(1)$ then Eqs.\,\eqref{fast_var_estimate_IM}-\eqref{slow_var_estimate_IM} show that the slow variable evolution has the same formal accuracy as the small-timestep implicit midpoint method. Interestingly howerver the formal accuracy in the fast variable for the same type of solution is one order in $h$ worse than small-timestep implicit midpoint. In practice this dichotomy should not prove problematic because the slow variable $x$ is often associated with ``macroscopic" quantities that are more easily measured than the ``microscopic" quantities encoded in the fast variable $y$.

As a final remark we note that choosing initial conditions so that the norms of the time derivatives in Eqs.\,\eqref{fast_var_estimate_IM}-\eqref{slow_var_estimate_IM} are $O(1)$ may be facilitated through a careful application of the zero-derivative principle. Instead of solving for the needed coefficients in $y_\epsilon^*$ by hand, one may apply a root finding algorithm to the functions $\mathfrak{D}_\epsilon^N$ introduced in Section \ref{zdp_sec} whose zero levels were shown to define approximate slow manifolds. During the root-finding procedure the $\mathfrak{D}_\epsilon^N(x,y)$ may be calculated by hand, which is certainly easier that solving for the $y_k^*$, or as in Ref.\,\citep{Gear_2006} a separate small-timestep integrator may be used to estimate $d^Ny/d\tau^N$ numerically. 

Within the geophysical fluid dynamics community the importance of choosing initial conditions for simulations that are close to the slow manifold has been recognized for some time. See for example Ref.\,\cite{Leith_1980} for an application of the nonlinear normal mode initialization introduced in\,\cite{Machenhauer_1977}, and Ref.\,\cite{Vautard_1986} for analytic expansion initialization. Failure to select initial conditions on the slow manifold leads to oscillations in numerical solutions of various forms of the primitive equations that are inconsistent with atmostpheric data. This phenomenon is by no means unique to the modeling of geophysical flows. Figure \ref{IC_fig} shows an example of what happens when the implicit midpoint method is used to integrate the zero-drag Abraham-Lorentz equations using a large numerical timestep and initial conditions that are not close to the system's slow manifold. From the perspective of the general theory of slow manifolds, the origin of the unphysical oscillations generated by the ``off-manifold" initial condition in Figure \ref{IC_fig} is exactly the same as the origin of unphysical gravity-inertia waves in simulations of a quasi-geostrophic atmosphere --- in each case a normally-elliptic slow manifold is to blame. More generally, any time that a reduced plasma model may be identified with a normally-elliptic slow manifold large-timestep implicit simulations of the corresponding parent model must invariably cope with unphysical grid-scale oscillations. Because such slow manifolds occur in many plasma models that exhibit both timescale separation and weak dissipation, slow manifold initialization is an important consideration in computational plasma physics even outside of the scope of slow manifold integrators.

\begin{figure}
\includegraphics[scale = .8]{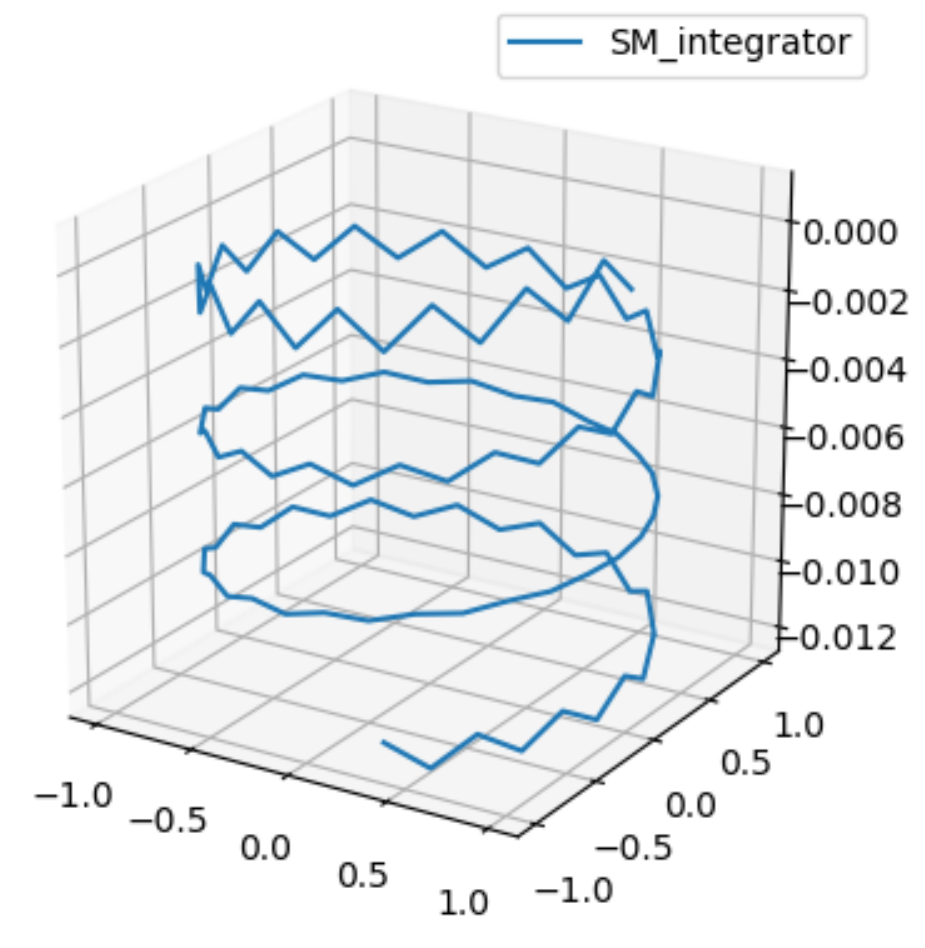}
\includegraphics[scale = .8]{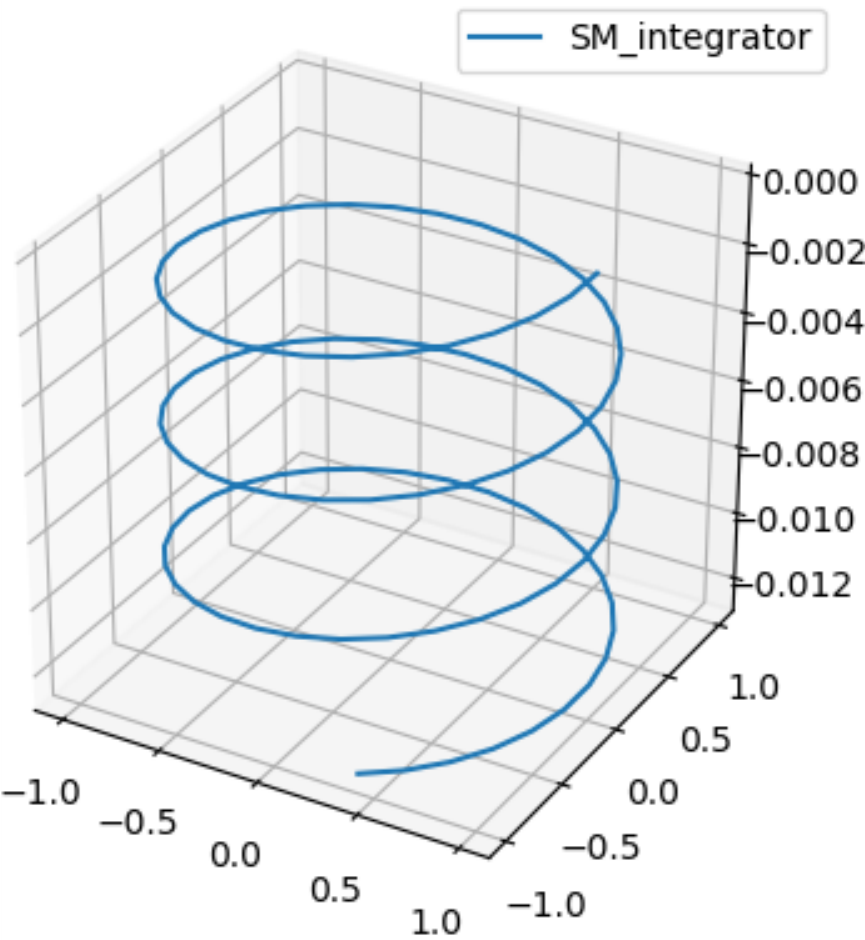}
\caption{\label{IC_fig} Results of integrating the zero-drag Abraham-Lorentz system with implicit midpoint using initial conditions off (left) and on (right) the slow manifold. In Cartesian coordinates, the magnetic field is $\bm{B} = \nabla\psi\times \bm{e}_z$ with $\psi = 4(x^2 + y^2)^2$. The parameter $\epsilon = 10^{-1}$ and the numerical timestep $\Delta t = .2\gg \epsilon$ does not resolve $\epsilon$. The off-manifold initial condition is $(x,y,z,v_\parallel,v_1,v_2) = (1,0,0,1,1,1)$, while the on-manifold initial condition is $(1,0,0,1,v^*_1,v^*_2)$, where $v_1^*\bm{e}_1 + v_1^*\bm{e}_2 = \epsilon \,\bm{v}_{\perp1}^* + \epsilon^2\,\bm{v}_{\perp 2}^*$ and $\bm{v}_{\perp 1}^*,\bm{v}_{\perp 2}^*$ are given in Eqs.\,\eqref{zero_drag_y1}-\eqref{zero_drag_y2}. The integration on the right agrees well with the true solution, while the integration on the left exhibits unphysical grid-scale oscillations.}
\end{figure} 

\end{example}

The notion of slow manifold integration should be compared with the related notion of asymptotic-preserving (AP) scheme. For the theoretical foundations of AP schemes we refer readers to the  Review article \citep{Jin_2012}. For several notable applications of the AP idea to problems in plasma physics we recommend \citep{Filbet_2010,Degond_2016,Degond_2017,Degond_2010}. An AP scheme for a fast-slow system may be defined as an integrator $\Phi_{(h,\epsilon)}:(x_n,y_n)\mapsto (x_{n+1},y_{n+1})$ for the fast-slow system such that (a) for fixed $\epsilon$ and sufficiently-small $h$ the map $\Phi_{(h,\epsilon)}$ approximates the fast-slow system's time-$h$ flow, and (b) as $\epsilon\rightarrow 0$ the map $\Phi_{(h,\epsilon)}$ approximates the time-$h$ flow of the system's limiting slow manifold reduction $\dot{x} = g_0^*(x)$. Apparently property (a) for AP schemes coincides with property (1) in the definition of slow manifold integrators, but property (b) for AP schemes is different from a slow manifold integrator's property (2). This difference is not superficial; AP schemes are not necessarily slow manifold integrators, and slow manifold integrators are not necessarily AP schemes.

For a simple example of an AP scheme that is not a slow manifold integrator consider the fast-slow system $\epsilon\,\dot{y} = A y$, $\dot{x} = g(x)$, where $A$ is an invertible matrix and $g$ is an arbitrary smooth function of $x$. The slow manifold for this system is  $y_\epsilon^*(x) = 0$ and the limiting slow manifold reduction is simply $\dot{x} = g(x)$. Suppose we discretize time using the scheme
\begin{align}
y_{n+1} &= y_n + \frac{h}{\epsilon} A \bigg([1 - w(h/\epsilon)]y_n + w(h/\epsilon)\left[y_{n+1} - y_n - \gamma_0\right]\bigg)\\
x_{n+1} & = x_n + h\,g(x_n)
\end{align}
where $w(h/\epsilon) = \exp(-\epsilon^2/h^2)$ and $\gamma_0\in Y$ is a non-zero constant. For fixed $\epsilon$ this scheme converges exponentially quickly to the forward Euler method as $h$ tends to zero, and therefore satisfies property (a). On the other hand fixing $h$ and sending $\epsilon$ to zero leads to the limit scheme
\begin{align}
y_{n+1} &= y_n + \gamma_0\label{ex_fwd_euler_limit_y}\\
x_{n+1}& = x_n + h\,g(x_n).\label{ex_fwd_euler_limit}
\end{align}
Because Eq.\,\eqref{ex_fwd_euler_limit} comprises the forward Euler discretization for $\dot{x} = g(x)$ this limit scheme approximates the time-$h$ flow of the limiting slow manifold reduction, which implies that property (b) is satisfied. We conclude that the scheme is AP. However it cannot be a slow manifold integrator because the limiting $y$ evolution \eqref{ex_fwd_euler_limit} has no fixed points.

For an example of a slow manifold integrator that is not AP consider Lagrangian system defined by the $\epsilon$-dependent Lagrangian
\begin{align}
\ell_\epsilon(\bm{q},\bm{p},\dot{\bm{q}},\dot{\bm{p}}) = \frac{1}{2}\epsilon |\dot{\bm{p}}|^2 + \frac{1}{2}\epsilon |\dot{\bm{q}}|^2  + \bm{p}\cdot\dot{\bm{q}} - H(\bm{q},\bm{p}),
\end{align}
where $\bm{q},\bm{p}\in\mathbb{R}^d$ and $H$ is any smooth function on $\mathbb{R}^{2d}$. In first-order form, the Euler-Lagrange equations associated with this Lagrangian are given by $\epsilon\,\dot{\bm{f}} = \bm{v} - \partial_{\bm{p}}H$, $\epsilon\,\dot{\bm{v}} = -\bm{f} - \partial_{\bm{q}}H$, $\dot{\bm{p}} = \bm{f}$, $\dot{\bm{q}} = \bm{v}$. This is a fast-slow system with slow variable $x = (\bm{q},\bm{p})$ and fast variable $y = (\bm{v},\bm{f})$. The limiting slow manifold is $y_0^*(x) = (\partial_{\bm{p}}H, -\partial_{\bm{q}}H)$, whence the limiting slow manifold reduction is $\dot{\bm{q}} = \partial_{\bm{p}}H$, $\dot{\bm{p}} = -\partial_{\bm{q}}H$. Notice that the limiting slow manifold reduction coincides with the well-known canonical Hamilton's equations. Suppose we discretize time using the scheme 
\begin{align}
\epsilon \frac{\bm{v}_{n+1/2} - \bm{v}_{n-1/2}}{h} &= -\frac{\bm{f}_{n+1/2} + \bm{f}_{n-1/2}}{2} - \partial_{\bm{q}}H\left(\bm{q}_n, \bm{p}_n \right)\\
\epsilon\frac{\bm{f}_{n+1/2} - \bm{f}_{n-1/2}}{h} & = \frac{\bm{v}_{n+1/2} + \bm{v}_{n-1/2}}{2}  - \partial_{\bm{p}}H\left(\bm{q}_n ,\bm{p}_n \right)\\
  \frac{\bm{q}_{n+1}-\bm{q}_n}{h}  & = \bm{v}_{n+1/2}\\
 \frac{\bm{p}_{n+1} - \bm{p}_n}{h} & = \bm{f}_{n+1/2}.
\end{align} 
%\begin{align}
%\epsilon \frac{\bm{v}_{n+1/2} - \bm{v}_{n-1/2}}{h} &= -\frac{\bm{f}_{n+1/2} + \bm{f}_{n-1/2}}{2} - \frac{1}{2}\partial_{\bm{q}}H\left(\bm{q}_n-\frac{h}{2}\bm{v}_{n-1/2}, \bm{p}_n - \frac{h}{2}\bm{f}_{n-1/2}\right)\nonumber\\& -\frac{1}{2}\partial_{\bm{q}}H\left(\frac{\bm{q}_{n+1} + \bm{q}_n}{2},\frac{\bm{p}_{n+1}+\bm{p}_n}{2}\right)\\
%\epsilon\frac{\bm{f}_{n+1/2} - \bm{f}_{n-1/2}}{h} & = \frac{\bm{v}_{n+1/2} + \bm{v}_{n-1/2}}{2}  - \frac{1}{2}\partial_{\bm{p}}H\left(\bm{q}_n - \frac{h}{2}\bm{v}_{n-1/2},\bm{p}_n - \frac{h}{2}\bm{f}_{n-1/2}\right)\nonumber\\
% &- \frac{1}{2}\partial_{\bm{p}}H\left(\frac{\bm{q}_{n+1} + \bm{q}_n}{2},\frac{\bm{p}_{n+1}+\bm{p}_n}{2}\right)\\
% \frac{\bm{q}_{n+1}-\bm{q}_n}{h}  & = \bm{v}_{n+1/2}\\
% \frac{\bm{p}_{n+1} - \bm{p}_n}{h} & = \bm{f}_{n+1/2}.
%\end{align} 
For fixed $\epsilon$ this scheme is formally second-order accurate in $h$, implying that property (1) is satisfied. As $(h,\delta) = (h,\epsilon/h)$ tends to zero the scheme limits to
\begin{align}
\bm{f}_{n+1/2} = & -\bm{f}_{n-1/2} - 2\partial_{\bm{q}}H(\bm{q}_n,\bm{p}_n)\\
\bm{v}_{n+1/2} = & -\bm{v}_{n-1/2} + 2\partial_{\bm{p}}H(\bm{q}_n,\bm{p}_n)\\
\bm{p}_{n+1} = & \bm{p}_n\\
\bm{q}_{n+1} = & \bm{q}_n,
\end{align}
which shows that property (2) is satisfied. (C.f. Example \ref{IM_example}.) Therefore this scheme defines a slow manifold integrator. However in Ref.\,\citep{Ellison_2018} Ellison \emph{et. al.} have shown that even for simple choices of $H$ the scheme obtained by sending $\epsilon$ to zero while holding $h$ fixed,
\begin{align}
\frac{\bm{f}_{n+1/2} + \bm{f}_{n-1/2}}{2}&=  - \partial_{\bm{q}}H\left(\bm{q}_n, \bm{p}_n \right)\\
\frac{\bm{v}_{n+1/2} + \bm{v}_{n-1/2}}{2}& =    \partial_{\bm{p}}H\left(\bm{q}_n ,\bm{p}_n \right)\\
  \frac{\bm{q}_{n+1}-\bm{q}_n}{h}  & = \bm{v}_{n+1/2}\\
 \frac{\bm{p}_{n+1} - \bm{p}_n}{h} & = \bm{f}_{n+1/2},
\end{align}
exhibits so-called \emph{parasitic modes}. (See Section II.A in the latter reference.) These parasitic modes prevent this limit scheme from being a stable integrator for the limit system, even when initial conditions are chosen on the slow manifold. Using terminology from the AP literature, this limit scheme is not even weakly-AP because well-prepared initial data will eventually evolve away from the well-prepared constraint set defined by the slow manifold.

While the class of AP schemes differs from the class of slow manifold integrators a number of schemes do lie at the intersection of the two. For example the implicit midpoint method applied within a fast-slow split is both a slow manifold integrator and an AP scheme. The same is true of the backward Euler method. In cases where both concepts apply it may be beneficial to recognize that a scheme is a slow manifold integrator, even after the AP property has been verified. For example doing so allows one to immediately invoke Theorem \ref{formal_DTSM_existence} to conclude that the scheme possesses a discrete-time formal slow manifold and to estimate how closely the discrete-time formal slow manifold approximates the continuous-time formal slow manifold. If the discrete-time formal slow manifold disagrees with the continuous-time formal slow manifold at, say, third order, then there is no hope for the scheme to capture physical effects, e.g. drifts, associated with third-order corrections to the slow manifold. While we will not discuss this further in this Review, methods for analyzing the normal stability of continuous-time formal slow manifolds may also be adapted to the analyses of the stability of discrete-time formal slow manifolds; in order for a slow manifold integrator to be weakly-AP its formal slow manifold must be nonlinearly normally stable.

There are many interesting open questions pertaining to slow manifold integrators. May composition methods be used to construct higher-order slow manifold integrators out of lower-order examples like implicit midpoint or backward Euler? Can the order of the discrete-time slow manifold be increased without increasing the formal order of the scheme? In the context of non-dissipative fast-slow systems can normally-elliptic slow manifold integrators be found that exhibit discrete-time slow manifolds with good normal stability properties over $O(1/\epsilon^k)$ time intervals? Is there a discrete-time analogue of the zero-derivative principle that would allow the numerical generation of a points on a discrete-time slow manifold? Are there benefits to developing an integrator for a reduced model by first developing a slow manifold integrator for its parent model and then applying discrete-time slow manifold reduction? These questions as well as warrant further investigation.

\section{Presymplectic Slow Manifold Reduction\label{hamiltonian_SM}}
In this section we will study slow manifolds and slow manifold reduction in the context of fast-slow systems that happen to possess Hamiltonian structure. This category of fast-slow systems deserves special attention because many models of plasma behavior with weak dissipation may be approximated as Hamiltonian systems over  sufficiently short timescales. (See for instance Refs.\,\cite{Morrison_1980} and \cite{Marsden_1982}, where the Hamiltonian structure underlying the collisionless Vlasov-Maxwell system was uncovered for the first time, and Ref.\,\cite{Morrison_MHD_1980} for the discovery the MHD's Poisson bracket.) We will show that a fast-slow system's Hamiltonian structure may be exploited systematically to identify slow manifold reductions that comprise Hamiltonian systems in their own right. In particular we will describe a method of constructing slow manifold reductions for Hamiltonian fast-slow systems that leads to \emph{inheritance} of the parent model's Hamiltonian structure by the reduced model. In connection with this inherited Hamiltonian structure we will also discuss inheritance of continuous symmetries, which by way of Noether's theorem provides a useful method for studying conservation laws in Hamiltonian slow manifold reductions.

Recently a number of fundamental non-dissipative reduced plasma models have been shown to arise as slow manifold reductions, and had their Hamiltonian structures explained by the ideas of inheritance that will be described in this Section. In particular the Hamiltonian structures underlying (extended) MHD \cite{Burby_two_fluid_2017}, kinetic MHD \cite{Burby_Sengupta_2017_pop}, guiding center theory \cite{Burby_loops_2019}, and nonlinear wave-mean-flow interaction \cite{Burby_Ruiz_2019} have been explained in this manner. In his review \cite{MacKay_2004} MacKay also discusses inheritance of Hamiltonian structure by slow manifolds, and treats several examples outside of plasma physics. Our discussion will cover some of the same ground as MacKay as well as additional topics.

\subsection{Inheritance of Hamiltonian structure}

The foundational set of ideas that underlies inheritance of Hamiltonian structure by slow manifolds is the following. Let $(Z,\omega)$ be a symplectic manifold, a smooth manifold $Z$ equipped with a closed non-degenerate two-form $\omega$. Let $U$ be a Hamiltonian vector field on $M$ with Hamiltonian $H$. Then $U$, $H$, and $\omega$ are related by \emph{Hamilton's equations} \cite{Abraham_2008}, which read
\begin{align}
\iota_U\omega = \mathbf{d}H,\label{basic_symplectic_hamilton_equation}
\end{align}
where $i_U$ is the interior product with respect to $U$ and $\mathbf{d}$ is the exterior derivative on $Z$. We say that $U$ has a symplectic Hamiltonian structure. If $S\subset Z$ is an invariant manifold for $U$ then the dynamics of $U$ restricted to $S$ possess a \emph{pre}symplectic Hamiltonian structure, meaning that on the manifold $S$ there is a closed two-form $\omega_S$ and a smooth function $H_S$ such that the vector field $U_S$ induced by $U$ on $S$ satisfies
\begin{align}
\iota_{U_S}\omega_S = \mathbf{d}H_S.\label{primordial_reduced_hamilton}
\end{align}
In this previous formula $\mathbf{d}$ is the exterior derivative on $S$. Note that because $\omega_S$ may be degenerate Eq.\,\eqref{primordial_reduced_hamilton} does not by itself uniquely determine $U_S$, in contrast to Eq.\,\eqref{basic_symplectic_hamilton_equation}. Instead the component of $U_S$ that lies in the kernel of $\omega_S$ must be determined by requiring $U_S$ to be the restriction of $U$ to the invariant manifold $S$.  If $I_S:S\rightarrow M $ is the inclusion map for $S$ the two-form $\omega_S$ and the smooth function $H_S$ are given by
\begin{align}
\omega_S &= I_S^*\omega\\
H_S & = I_S^*H,
\end{align}
where $I_S^*$ is the pullback operator associated with the inclusion. In other words the reduced dynamics on the invariant manifold $S$ naturally \emph{inherit} a Hamiltonian structure from the unreduced dynamics on $Z$.

% example: energy surface in zero-drag AL dynamics

The formula Eq.\,\eqref{primordial_reduced_hamilton} is simple to establish. By Hamilton's equations for $U$ and the commutativity of the exterior derivative with pullback we have the identity
\begin{align}
\mathbf{d}H_S = I_S^*\mathbf{d}H = I_S^*(\iota_U\omega).\label{proto_reduced_ham}
\end{align}
By the definition of pullback if $\dot{s}\in T_sS$ is a vector tangent to $S$ at $s\in S$ we have
\begin{align}
I_S^*(\iota_U\omega)(s)[\dot{s}] &= \omega(I_S(s))[U(I_S(s)), T_sI_S[\dot{s}]]\nonumber\\
& = \omega(I_S(s))[T_sI_S[U_S(s)],T_sI_S[\dot{s}]]\nonumber\\
& = (I_S^*\omega)(s)[U_S(s),\dot{s}]\nonumber\\
& =( \iota_{U_S}\omega_S)(s)[\dot{s}],\label{inherit_exact_eqn1}
\end{align}
where we have also used the fact that $U_S$ and $U$ are $I_S$-related, i.e. $T_sI_S[U_S(s)] = U(s)$. Because Eq.\,\eqref{inherit_exact_eqn1} holds for all $s\in S$ we conclude that $I_S^*(\iota_U\omega) = \iota_{U_S}\omega_S$. In light of Eq.\,\eqref{proto_reduced_ham} this implies Eq.\,\eqref{primordial_reduced_hamilton}.

Suppose now that $U$ is the infinitesimal generator of a Hamiltonian system on $Z$ that admits a fast-slow split. Recall (c.f. Definition \ref{fs_split_def}) that this means there are coordinates $(x,y)$ on $Z$ in which the dynamical system $\dot{z} = U(z)$ is equivalent to a fast-slow system in sense of Definition \ref{FS_def}. If $\bm{\psi}_\epsilon:z\mapsto (x,y)$ denotes the coordinate transformation and $W_\epsilon(x,y) = (g_\epsilon(x,y),f_\epsilon(x,y)/\epsilon)$ denotes the vector field $U$ expressed in the $(x,y)$-coordinates then $W_\epsilon$ satisfies its own set of Hamilton's equations because 
\begin{align}
\iota_U\omega &= \mathbf{d}H\nonumber\\
\Rightarrow\bm{\psi}_{\epsilon*}(\iota_U\omega) &= \bm{\psi}_{\epsilon*}(\mathbf{d}H)\nonumber\\
\Rightarrow  \iota_{W_\epsilon}\Omega_\epsilon &= \mathbf{d}\mathcal{H}_\epsilon,
\end{align}
where $\Omega_\epsilon = \bm{\psi}_{\epsilon *}\omega$ and $\mathcal{H}_\epsilon = \bm{\psi}_{\epsilon*}H$ denote the symplectic form and the Hamiltonian function expressed in the $(x,y)$-coordinates. We may therefore synthesize the theory of slow manifolds in fast-slow systems with the theory of invariant manifolds in Hamiltonian systems to study the possibility that slow manifolds in Hamiltonian fast-slow systems inherit Hamiltonian structure. If slow manifolds were truly invariant manifolds this would be described completely by the argument in the previous paragraph. However, because slow manifolds are only invariant objects to all-orders in perturbation theory the inheritance phenomenon must be reexamined.

First we define the class of Hamiltonian fast-slow systems that we will study as follows.
\begin{definition}[Hamiltonian fast-slow system]
Let $W_\epsilon(x,y) = (g_\epsilon(x,y),f_\epsilon(x,y)/\epsilon)$ be the infinitesimal generator of a fast-slow system on $X\times Y$. This fast-slow system is \emph{Hamiltonian} if there is an $\epsilon$-dependent function $H_\epsilon$ and an $\epsilon$-dependent closed $2$-form $\Omega_\epsilon$ that depend smoothly on $\epsilon$ and such that
\begin{align}
\iota_{W_\epsilon}\Omega_\epsilon = \mathbf{d}H_\epsilon.
\end{align}
\end{definition}
\noindent Note that we work in the class of \emph{presymplectic} manifolds, which is strictly larger than the class of symplectic manifolds, and which is distinct from the class of Poisson manifolds. We refer the reader to MacKay's review in Ref.\,\citep{MacKay_2004} for a discussion of slow manifolds for Hamiltonian systems on symplectic and Poisson manifolds. A practical reason for considering the presymplectic case is the flexibility it provides when working with systems with gauge degrees of freedom. For example, in \cite{Burby_loops_2019} a reparameterization symmetry of phase space loops is exploited when formulating guiding center dynamics as slow manifold dynamics in loop space. This reparameterization symmetry naturally leads to a presymplectic, rather than symplectic structure on loop space. While the gauge freedom associated with loop reparameterization may be eliminated in order to obtain a symplectic formulation of loop dynamics, calculations are simpler when working directly with the presymplectic formulation. Another illustrative example of this point from plasma physics may be found in \cite{Burby_two_fluid_2017}, where the Hamiltonian structure of magnetohydrodynamics is explained using slow manifold reduction. There a useful technical step in the calculations of the Hamiltonian structure was eliminating the electron number density using the Gauss constraint. Because the Gauss constraint may be understood as a momentum map, this elimination of the electron density may be understood as restriction to a level set of that momentum map. As is true for symplectic Hamiltonian systems in general, restriction of the system's symplectic form to the momentum map's level set necessarily produces a presymplectic, rather than symplectic form. (The degenerate directions may be related to the symmetry underlying the momentum map.) More generally, presymplectic manifolds are the natural setting for Lagrangian dynamical systems with degenerate Lagrangians. An interesting future research direction would be formulating a theory of slow manifold reduction for Hamiltonian dynamical systems on Dirac manifolds \cite{Courant_1990}, which generalize Hamiltonian dynamical systems on both presymplectic and Poisson manifolds.

The basic result governing the inheritance of Hamiltonian structure by slow manifolds is the following.
\begin{theorem}[formal inheritance of Hamiltonian structure]\label{formal_hamiltonian_inheritance}
If $\epsilon\,\dot{y} = f_\epsilon(x,y)$, $\dot{x} =  g_\epsilon(x,y)$ is a Hamiltonian fast-slow system then there are formal power series 
\begin{align}
H_{\epsilon}^* &= H_{0}^* + \epsilon \,H_{1}^* + \epsilon^2 H_{2}^* + \dots\\
\Omega_{\epsilon}^* &= \Omega_{0}^* + \epsilon\,\Omega_{1}^* + \epsilon^2\,\Omega_{2}^* + \dots,
\end{align}
where each $H_{k}^*$ is a smooth function on $X$, each $\Omega_{k}^*$ is a closed $2$-form on $X$, and we have the equality of formal power series
\begin{align}
\iota_{g_\epsilon^*}\Omega_{\epsilon}^* = \mathbf{d}H_{\epsilon}^*.\label{formal_hamilton_equation}
\end{align}
(Recall from Def.\,\ref{formal_SMR_definition} that $g_\epsilon^*$ is the formal power series generator of slow dynamics on $X$.) Moreover we have the following explicit formulas for $\Omega_{\epsilon}^*$ and $H_{\epsilon}^*$:
\begin{align}
H_{\epsilon}^*(x) &= H_{\epsilon}(x,y_\epsilon^*(x))\label{thm_slow_ham}\\
\Omega_{\epsilon}^*(x)[V_x,W_x] &= \Omega_\epsilon(x,y_\epsilon^*(x))\left[(V_x,Dy_\epsilon^*(x)[V_x]),(W_x,Dy_\epsilon^*(x)[W_x])\right].\label{thm_slow_symp}
\end{align}
\end{theorem}

\begin{proof}
By Hamilton's equations, for each $(\delta x,\delta y)\in T_{(x,y)}X\times Y$
\[
\Omega_\epsilon(x,y)[(g_\epsilon(x,y),f_\epsilon(x,y)/\epsilon),(\delta x,\delta y)] = \mathbf{d}H_\epsilon(x,y)[(\delta x,\delta y)].
\]
In particular if we set $y = y_\epsilon^*(x)$ and $\delta y = Dy_\epsilon^*(x)[\delta x]$, where $y_\epsilon^*$ is the fast-slow system's formal slow manifold, then we obtain the following equality of formal power series
\begin{align}
\Omega_\epsilon(x,y_\epsilon^*)[(g_\epsilon(x,y_\epsilon^*),f_\epsilon(x,y_\epsilon^*)/\epsilon),(\delta x , Dy_\epsilon^*(x)[\delta x])] = \mathbf{d}H_\epsilon(x,y_\epsilon^*)[(\delta x , Dy_\epsilon^*(x)[\delta x])].
\end{align}
Because the invariance equation defining the formal slow manifold says $f_\epsilon(x,y_\epsilon^*(x)) = \epsilon\,Dy_\epsilon^*(x)[g_\epsilon(x,y_\epsilon^*)]$ this implies
\begin{align}
\Omega_\epsilon(x,y_\epsilon^*)[(g_\epsilon(x,y_\epsilon^*),Dy_\epsilon^*(x)[g_\epsilon(x,y_\epsilon^*)]),(\delta x , Dy_\epsilon^*(x)[\delta x])] = \mathbf{d}H_\epsilon(x,y_\epsilon^*)[(\delta x , Dy_\epsilon^*(x)[\delta x])],
\end{align}
or in terms of $\Omega_\epsilon^*$ and $g_\epsilon^*$,
\begin{align}
\Omega_\epsilon^*(x)[g_\epsilon^*(x),\delta x] = \mathbf{d}H_\epsilon(x,y_\epsilon^*)[(\delta x,Dy_\epsilon^*(x)[\delta x])].
\end{align}
The formula \eqref{formal_hamilton_equation} now follows from the chain rule for formal power series applied to $\mathbf{d}H_\epsilon^*$:
\begin{align}
\mathbf{d}H_\epsilon^*(x)[\delta x] = \mathbf{d}H_\epsilon(x,y_\epsilon^*(x))[(\delta x,Dy_\epsilon^*(x)[\delta x])].
\end{align}

To finish the proof we must now demonstrate that the coefficients in the formal power series $\Omega_\epsilon^*$ are all closed. To that end let $\widetilde{y}_\epsilon = y_0^* + \epsilon\,y_1^* + \dots+ \epsilon^N y_N^*$ be the $N$'th-order slow manifold defined by naive power series truncation. By the fundamental theorem of calculus there is an $O(1)$ formal power series $\alpha_\epsilon$ whose coefficients are $2$-forms on $X$ such that 
\begin{align}
\Omega_\epsilon^* = \widetilde{I}_\epsilon^*\Omega_\epsilon + \epsilon^{N+1} \alpha_\epsilon,
\end{align}
where $\widetilde{I}_\epsilon(x) = (x,\widetilde{y}_\epsilon(x))$ is the slow manifold's inclusion map. By the commutativity of pullback with exterior differentiation we therefore have
\begin{align}
\mathbf{d}\Omega_\epsilon^* = \epsilon^{N+1}\mathbf{d}\alpha_\epsilon,
\end{align}
which says that all coefficients in the formal power series $\mathbf{d}\Omega_\epsilon^*$ of order less than or equal to $N$ must vanish. Because $N$ is arbitrary this implies $\mathbf{d}\Omega_\epsilon^* = 0$ as a formal power series.

\end{proof}

\begin{example}\label{zero_drag_formal_structure}
We will apply Theorem \ref{formal_hamiltonian_inheritance} to identify the formal Hamiltonian structure underlying Abraham-Lorentz slow dynamics in the zero-drag regime. In particular we will find the first few terms in the expansions of $H_{\epsilon}^*$ and $\Omega_{\epsilon}^*$ and relate them to Littlejohn's pioneering results \cite{Littlejohn_1981,Littlejohn_1982,Littlejohn_1983,Littlejohn_1984} on the Hamiltonian structure underlying guiding center theory. (Recall from Example \ref{second_order_smr_zero_drag}  that the slow manifold in the zero-drag equations corresponds to guiding center dynamics on the zero level set of the magnetic moment adiabatic invariant. This is why there \emph{should} be a relationship between the guiding center Hamiltonian structure and the Hamiltonian structure of the zero-drag slow dynamics.)

First we must demonstrate that the zero-drag equations may be written as a fast-slow system that satisfies the hypotheses of Theorem \ref{formal_hamiltonian_inheritance}. As mentioned in Section \ref{zero_drag_sec}, the zero-drag equations arise from a variational principle.   One implication of this variational principle is that the zero-drag equations possess a symplectic Hamiltonian structure with Hamiltonian 
\begin{align}
H_\epsilon(\bm{x},\bm{v})  = \epsilon\,\frac{1}{2}|\bm{v}|^2, \label{zero_drag_ham_example}
\end{align}
and symplectic form
\begin{align}
\omega_\epsilon(\bm{x},\bm{v})[(\delta\bm{x}_1,\delta\bm{v}_1),(\delta\bm{x}_2,\delta\bm{v}_2)] = -\zeta\,\bm{B}(\bm{x})\cdot\delta\bm{x}_1\times\delta\bm{x}_2 + \epsilon\,(\delta\bm{x}_1\cdot \delta\bm{v}_2 - \delta\bm{x}_2\cdot\delta\bm{v}_1).\label{zero_drag_form_example}
\end{align}
In fact it is straightforward to verify directly that the infinitesimal generator of zero-drag dynamics $U_\epsilon(\bm{x},\bm{v}) = (\bm{v},\epsilon^{-1}\,\zeta\,\bm{v}\times\bm{B}(\bm{x}))$ satisfies $\iota_{U_\epsilon}\omega_\epsilon = \mathbf{d}H_\epsilon$. It follows that the zero-drag equations written in terms of their fast-slow split $x= (\bm{x},v_\parallel)$, $y = (v_1,v_2)$ (c.f. Example \ref{FS_example_hard}) also possess a Hamiltonian structure that may be obtained from Eqs.\,\eqref{zero_drag_ham_example} -\eqref{zero_drag_form_example} by merely changing coordinates. In particular if $W_\epsilon = (f_\epsilon/\epsilon,g_\epsilon)$ with $f_\epsilon$ and $g_\epsilon$ given in Eqs.\,\eqref{zero_drag_new_Feps_one}-\eqref{zero_drag_new_Geps_two} then $\iota_{W_\epsilon}\Omega_\epsilon = \mathbf{d}H_\epsilon$ with
\begin{align}
H_\epsilon(x,y) = \epsilon\,\frac{1}{2}v_\parallel^2 + \epsilon\,\frac{1}{2}v_1^2 + \epsilon\,\frac{1}{2}v_2^2,
\end{align}
and
\begin{align}
&\Omega_\epsilon(x,y)[(\delta x_1,\delta y_1),(\delta x_2,\delta y_2)] =\nonumber\\
& - \left(\zeta\,\bm{B}(\bm{x}) + \epsilon\,v_\parallel\,\nabla\times\bm{b}+\epsilon\,\nabla\times\bm{v}_\perp\right)\cdot \delta\bm{x}_1\times\delta\bm{x}_2\nonumber\\
& + \epsilon\,(\delta\bm{x}_1\cdot\delta\bm{v}_2 - \delta\bm{x}_2\cdot\delta\bm{v}_1).
\end{align} 
Here we have introduced the useful shorthand notation $\delta\bm{v} = \delta v_\parallel\,\bm{b} + \delta v_1\,\bm{e}_1 + \delta v_2\,\bm{e}_2$, and we are reusing the shorthand $\bm{v}_\perp = v_1\,\bm{e}_1+v_2\,\bm{e}_2$ introduced in Eq.\,\eqref{perp_shorthand}. We are therefore justified in applying Theorem \ref{formal_hamiltonian_inheritance} directly.

According to the formula \eqref{thm_slow_ham} the formal power series $H_{\epsilon}^*$ may be written 
\begin{align}
H_\epsilon^*(x) = \epsilon\,\frac{1}{2}v_\parallel^2 + \epsilon\,\frac{1}{2}|\bm{v}_{\perp\epsilon}^*|^2,
\end{align}
where $\bm{v}_{\perp\epsilon}^*$ (c.f. Eq.\,\eqref{perp_shorthand}) encodes the slaving function $y_\epsilon^*= ((v_{1})_\epsilon^*,(v_{2})_\epsilon^*)$ for the perpendicular velocity. We have effectively already examined the first few terms of the formal power series $H_\epsilon$ in Example \ref{energy_breakdown_example}, where we considered $E_\epsilon^* = H_\epsilon^*/\epsilon$. We may therefore immediately write 
\begin{align}
H_0^*(x) &= 0\label{h0_zero_drag}\\
H_1^*(x) & = \frac{1}{2}v_\parallel^2\label{h1_zero_drag}\\
H_2^*(x) & = 0\\
H_3^*(x) & = \frac{1}{2}|\bm{v}_c|^2,\label{h3_zero_drag}
\end{align}
where the curvature drift velocity $\bm{v}_c$ was given earlier in Eq.\,\eqref{zero_drag_y1}. 

According to the formula \eqref{thm_slow_symp} the formal power series $\Omega_\epsilon^*$ is specified by 
\begin{align}
&\Omega_\epsilon^*(x)[\delta x_1,\delta x_2] =\nonumber\\
& - \left(\zeta\,\bm{B}(\bm{x}) +\epsilon\,v_\parallel\,\nabla\times\bm{b} +\epsilon\,\nabla\times\bm{v}_{\perp\epsilon}^*\right)\cdot \delta\bm{x}_1\times\delta\bm{x}_2\nonumber\\
&+\epsilon\,(\delta\bm{x}_1\,\delta v_{\parallel 2} - \delta\bm{x}_2\, \delta v_{\parallel 1})\cdot(\bm{b}(\bm{x}) + \partial_{v_\parallel}\bm{v}_{\perp\epsilon}^*).
\end{align}
The first several coefficients of the formal power series expansion for $\Omega_\epsilon^*$ are therefore given by
\begin{align}
\Omega_0^*(x)[\delta x_1,\delta x_2] =& - \zeta\,\bm{B}(\bm{x})\cdot\delta\bm{x}_1\times\delta\bm{x}_2\label{omega0_zero_drag}\\
\Omega_1^*(x)[\delta x_1,\delta x_2] =& -v_\parallel\,\nabla\times\bm{b}(\bm{x})\cdot\delta\bm{x}_1\times\delta\bm{x}_2 + (\delta\bm{x}_1\,\delta v_{\parallel 2} - \delta\bm{x}_2\,\delta v_{\parallel 1})\cdot\bm{b}(\bm{x})\label{omega1_zero_drag}\\
\Omega_2^*(x)[\delta x_1,\delta x_2] =&-( \nabla\times\bm{v}_c)\cdot\delta\bm{x}_1\times\delta\bm{x}_2\nonumber\\
& + (\delta\bm{x}_1\,\delta v_{\parallel 2} - \delta\bm{x}_2\, \delta v_{\parallel 1})\cdot \partial_{v_\parallel}\bm{v}_c\\
\Omega_3^*(x)[\delta x_1,\delta x_2] = &-( \nabla\times\bm{v}_{\perp 2}^*)\cdot\delta\bm{x}_1\times\delta\bm{x}_2\nonumber\\
&+ (\delta\bm{x}_1\,\delta v_{\parallel 2} - \delta\bm{x}_2\, \delta v_{\parallel 1})\cdot \partial_{v_\parallel}\bm{v}_{\perp 2}^*,
\end{align}
where the second-order slaving function $\bm{v}_{\perp 2}^*$ was given in Eq.\,\eqref{zero_drag_y2}. The skeptical reader is encouraged to verify that each of the $2$-forms $\Omega_0^*,\Omega_1^*,\Omega_2^*,\Omega_3^*$ is indeed closed, as guaranteed by Theorem \ref{formal_hamiltonian_inheritance}.

In order to scrutinize these formal power series in light of Littlejohn's work it is most convenient to refer to Ref.\,\citep{Littlejohn_1984}, where Littlejohn gives an expression for the so-called Poincar\'e-Cartan form for guiding center dynamics in Eq.\,(12) and and expression for the guiding center Hamiltonian in Eq.\,(13). Consider first the Hamiltonian. The leading-order contribution to Littlejohn's guiding center Hamiltonian is $\frac{1}{2} v_\parallel^2 + \mu |\bm{B}|(\bm{x})$, where $\mu$ is the magnetic moment adiabatic invariant. Apparently this expression agrees with our Eq.\,\eqref{h1_zero_drag} precisely when $\mu=0$. This is consistent with out earlier assertion that the zero-drag slow manifold corresponds to guiding center motion on the $\mu = 0$ level set. Now consider the symplectic form. Littlejohn's symplectic form may be recovered from his Poincar\'e-Cartan form $\overline{\theta}$ in Eq.\,(12) by first writing $\overline{\theta} = \vartheta - H\,dt$ where $\vartheta$ has no $dt$-component and then defining $\Omega_{\text{gc}} = -\mathbf{d}\vartheta$. We then restrict to one of the level sets of the magnetic moment $\mu$, which enables a direct comparison with our $\Omega_\epsilon^*$. This immediately leads to the coincidences $\Omega_{\text{gc}0} =\Omega_0^* $ and $\Omega_{\text{gc}1} = \Omega_1^*$, again indicating consistency with our earlier claim about the $\mu=0$ level set. (Note however that these last coincidences do not require $\mu=0$.) However for $\Omega_{\text{gc}2}$ we find
\begin{align}
\Omega_{\text{gc}2}(x)[\delta x_1,\delta x_2] = -\mu \nabla\times\bm{R}\cdot \delta\bm{x}_1\times\delta\bm{x}_2,
\end{align}
which does not agree with our $\Omega_2^*$, even when $\mu = 0$. Such disagreement is actually to be expected in a comparison as naive as the one we have just performed. The correct way to compare with Littlejohn's results would account for the fact that in deriving his expression for the guiding center Poincar\'e-Cartan form Littlejohn applied a near-identity coordinate transformation $\mathsf{T}_{\text{gc}}$ to $(\bm{x},\bm{v})$-space. In other words Littlejohn's $(\bm{x},v_\parallel)$ are not the same as the ``Cartesian" $(\bm{x},v_\parallel)$ used by us. We will not perform such an analysis here. Instead we refer the interested reader to Ref.\,\citep{Parra_Calvo_2014} for a representative example of such a comparison, and to Theorem 2 in Ref.\,\citep{Burby_loops_2019} where a strictly more general equivalence problem is solved.
\end{example}

The content of Theorem \ref{formal_hamiltonian_inheritance} may be summarized by the statement that the formal slow manifold reduction $g_\epsilon^*$ of any fast-slow system with Hamiltonian structure possesses its own Hamiltonian structure to all orders in perturbation theory. In particular the Theorem \emph{does not say} that every slow manifold reduction of order $N$ possesses a Hamiltonian structure. A generic slow manifold reduction can ``break" the presymplectic structure for the same reasons that such a slow manifold reduction can break conservation laws. (c.f. Example \ref{energy_breakdown_example}.) We are therefore lead to ask the question: is there a way to construct slow manifold reductions of order $N$ that \emph{do} possess a Hamiltonian structure?

We define slow manifold reductions that possess Hamiltonian structure as follows.

\begin{definition}[Hamiltonian slow manifold reduction]
Let $\widetilde{g}_\epsilon$ be a slow manifold reduction of order $N$ for a fast-slow system $\epsilon\,\dot{y} = f_\epsilon(x,y)$, $\dot{x} =  g_\epsilon(x,y)$. This slow manifold reduction is \emph{Hamiltonian} if there is an $\epsilon$-dependent function $\widetilde{H}_\epsilon$ and an $\epsilon$-dependent closed $2$-form $\widetilde{\Omega}_\epsilon$ that each depend smoothly on $\epsilon$ and such that
\begin{align}
\iota_{\widetilde{g}_\epsilon}\widetilde{\Omega}_\epsilon = \mathbf{d}\widetilde{H}_\epsilon.
\end{align}
The function $\widetilde{H}_\epsilon$ is the \emph{approximate reduced Hamiltonian}. The $2$-form $\widetilde{\Omega}_\epsilon$ is the \emph{approximate reduced presymplectic form}.
\end{definition}

\noindent Our question is therefore: given a Hamiltonian fast-slow system, how may we construct Hamiltonian slow manifold reductions of order $N$?

When $N=0$ there is essentially no work to do. This is summarized in the following proposition.

\begin{proposition}[limiting reductions are Hamiltonian]\label{limiting_hamiltonian_struc}
If $\epsilon\,\dot{y} = f_\epsilon(x,y)$, $\dot{x} = g_\epsilon(x,y)$ is a Hamiltonian fast-slow system then the limiting slow manifold reduction $\dot{x} = g_0^*(x,y_0^*(x))$ is also Hamiltonian. 
\end{proposition}
\begin{proof}
This is a simple corollary of Theorem \ref{formal_hamiltonian_inheritance}. Indeed, the $O(1)$-component of the formal power series equality \eqref{formal_hamilton_equation} is
\begin{align}
\iota_{g_0^*}\Omega_{0}^* = \mathbf{d}H_0^*.
\end{align}
Therefore $g_\epsilon^*$ is a Hamiltonian slow manifold reduction of order $0$.
\end{proof}

\begin{example}\label{limit_ham_example}
An immediately consequence of Proposition \ref{limiting_hamiltonian_struc} is that the limiting slow manifold reduction for the zero-drag Abraham-Lorentz equations is Hamiltonian. Recall that this slow manifold reduction is given explicitly by
\begin{align}
\dot{x}_0^* &= v_\parallel\,\bm{b}(\bm{x})\\
\dot{v}_{\parallel 0}^* & = 0.
\end{align}
According to the results obtained in Example \ref{zero_drag_formal_structure}, the limiting Hamiltonian $H_0^*$ and presymplectic form $\Omega_0^*$ for this system are given by Eqs.\,\eqref{h0_zero_drag} and \eqref{omega0_zero_drag}, respectively. Because $H_0^* = 0$ the presymplectic Hamilton's equations for this limiting reduction are
\begin{align}
0=\iota_{g_0^*}\Omega_0^* = -\zeta \bm{B}\times\dot{\bm{x}}_0^*\cdot d\bm{x},
\end{align}
which is easy to verify without recourse to Proposition \ref{limiting_hamiltonian_struc}. An equivalent way of expressing the fact that the limiting slow manifold reduction for the zero-drag equations is Hamiltonian is the following. Let $D_0\subset Q\times \mathbb{R}$ be an embedded closed $2$-dimensional disc in $(\bm{x},v_\parallel)$-space, and let $D_t$ be the evolution of $D_0$ given by flowing $D_0$ along $g_0^*$ for $t$ seconds. Then the magnetic flux threading the $\bm{x}$-space projection of $D_0$ is equal to the magnetic flux threading the $\bm{x}$-space projection of $D_t$. This statement is reminiscent of the so-called frozen-in law from ideal MHD, although our statement applies to $2$-dimensional discs in the $4$-dimensional $(\bm{x},v_\parallel)$-space, whereas the usual frozen-in law pertains to $2$-dimensional discs in the $3$-dimensional $\bm{x}$-space.

\end{example}

When $N>0$ the question is more challenging. In his review \citep{MacKay_2004} MacKay gives a technique for constructing Hamiltonian slow manifold reductions in the special case where the pullback of the fast-slow system's presymplectic form $\Omega_\epsilon$ to some slow manifold of order $N$ is actually symplectic, i.e. when the pulled-back $2$-form is non-degenerate. However MacKay does not attempt to deal with more general cases where the pulled-back $2$-form is either degenerate or nearly so. The nearly-degenerate case is relevant to, for example, guiding center theory \citep{Burby_loops_2019}, the zero-drag Abraham-Lorentz equations, and kinetic MHD \citep{Burby_Sengupta_2017_pop}. The degenerate case is relevant to systems with gauge symmetry such as the ideal two-fluid system \citep{Burby_two_fluid_2017}. The following Theorem gives a method of constructing Hamiltonian slow manifold reductions of any order that explicitly deals with degeneracy. 

\begin{theorem}[constructing Hamiltonian slow manifold reductions]\label{ham_construction}
Let $\epsilon\,\dot{y} = f_\epsilon(x,y)$, $\dot{x} = g_\epsilon(x,y)$ be a Hamiltonian fast-slow system with presymplectic form $\Omega_\epsilon$ and Hamiltonian $H_\epsilon$. Let $\widetilde{\Omega}_\epsilon$ and $\widetilde{H}_\epsilon$ be order-$M$ approximations of the reduced presymplectic form and Hamiltonian on $X$ provided by Theorem \ref{formal_hamiltonian_inheritance}, i.e. $\widetilde{\Omega}_\epsilon - \Omega_\epsilon^* = O(\epsilon^{M+1})$ and $\widetilde{H}_\epsilon - H_\epsilon^* = O(\epsilon^{M+1})$. Assume $\widetilde{\Omega}_\epsilon$ and $\widetilde{H}_\epsilon$ satisfy the following properties.
\begin{itemize}
\item[(a)] For $x\in X$ and $\epsilon >0$ there is a smooth splitting of each tangent space $T_xX= K_x\oplus C_x$ as the sum of the kernel of the approximate presymplectic form $K_x = \{W_x\in T_xX\mid \iota_{W_x}\widetilde{\Omega}_\epsilon(x) = 0\}$ and a complementary subspace $C_x$.
\item[(b)] $\widetilde{\Omega}_\epsilon$ restricted to $C_x$ is invertible onto its image with $O(\epsilon^{-d})$ inverse for some non-negative integer $d$.
\item[(c)] $\mathbf{d}\widetilde{H}_\epsilon(x)$ is contained in the image of $\widetilde{\Omega}_\epsilon(x)$.
\end{itemize}
If for each $x\in X$ we define the tangent vector $\widetilde{g}_\epsilon^\perp(x)\in T_xX$ as the unique solution of $\iota_{\widetilde{g}_\epsilon^\perp(x)}\widetilde{\Omega}_\epsilon(x) = \mathbf{d}\widetilde{H}_\epsilon(x)$ contained in $C_x$, and $\widetilde{g}_\epsilon^\prime$ is \emph{any} slow manifold reduction (not necessarily Hamiltonian) of order $ M-d$, then 
\begin{align}
\widetilde{g}_\epsilon(x) = \widetilde{g}_\epsilon^\perp(x) + \pi_K(\widetilde{g}_\epsilon^\prime(x)),
\end{align}
where $\pi_K$ is the projection onto $K_x$ relative to $C_x$, is a Hamiltonian slow manifold reduction of order $M-d$.

\end{theorem}

\begin{remark}
Note that property (a) allows for the kernel of $\widetilde{\Omega}_\epsilon$ to change discontinuously at $\epsilon = 0$. This singular behavior is not uncommon in practice, and will be demonstrated in the context of the zero-drag Abraham-Lorentz equations. Also note that in finite dimensions: property (a) is implied by the weaker condition that $K_x$ depend smoothly on $x$ and $\epsilon >0$; and property (b) is automatically satisfied. Finally observe that $d>0$ may only occur when the kernel of $\widetilde{\Omega}_\epsilon$ changes discontinuously at $\epsilon = 0$.
\end{remark} 

\begin{remark}
On a general presymplectic manifold $(M,\omega)$ the equation $\iota_X\omega = \mathbf{d}H$ may not be solvable for the vector field $X$ because $\mathbf{d}H$ may not be in the image of $\omega$ over all of $M$. In these degenerate cases the Gotay-Nester constraint algorithm \citep{Gotay_1978,Gotay_1979} may be used to find a proper (possibly empty) submanifold of $M$ on which $X$ can be defined. In this Section we sidestep many of these delicate issues by defining Hamiltonian fast-slow systems to obey presymplectic Hamilton's equations that \emph{are} solvable. As a consequence the formal power series presymplectic Hamilton's equation \eqref{formal_hamilton_equation} must be solvable to all orders in perturbation theory. Nevertheless truncating the series for $\Omega_\epsilon^*$ and $H_\epsilon^*$ haphazardly may easily spoil the solvability property. Property (c) requires that this does not happen. We will not discuss in this Review truncation strategies that ensure (c) is satisfied. We remark however that the analysis of the Hamiltonian structure underlying extended MHD in Ref.\,\citep{Burby_two_fluid_2017} had to deal with this issue in the special case where the kernel of $\Omega_\epsilon^*$ was independent of $\epsilon >0$; in this case each coefficient of $\mathbf{d}H_\epsilon^*$ is necessarily in the image of $\Omega_\epsilon^*$, which simplifies matters substantially.
\end{remark}

\begin{remark}
If the kernel of the approximate presymplectic form is empty the above method for constructing Hamiltonian slow manifold reductions is a less-refined version of MacKay's \citep{MacKay_2004} Hamiltonian slow manifold reduction strategy. MacKay's method also includes a special technique for constructing the approximate reduced symplectic form and Hamiltonian that ensure the resulting slow manifold reduction contains all of the original system's equilibria that are sufficiently close to the slow manifold.
\end{remark}

\begin{proof}
Because $\widetilde{\Omega}_\epsilon$ and $\widetilde{H}_\epsilon$ are order-$M$ approximations of $\Omega_\epsilon^*$ and $H_\epsilon^*$ there must be $O(1)$ formal power series $\alpha_\epsilon$ and $\phi_\epsilon$ such that
\begin{align}
\widetilde{\Omega}_\epsilon - \Omega_\epsilon^* &= \epsilon^{M+1}\,\alpha_\epsilon\\
\widetilde{H}_\epsilon - H_\epsilon^* & = \epsilon^{M+1}\,\phi_\epsilon.
\end{align}
Therefore the vector field $\widetilde{g}_\epsilon$ satisfies 
\begin{align}
\iota_{\widetilde{g}_\epsilon}\widetilde{\Omega}_\epsilon &= \mathbf{d}\widetilde{H}_\epsilon\nonumber\\
&= \mathbf{d}H_\epsilon^* + \epsilon^{M+1}\mathbf{d}\phi_\epsilon\nonumber\\
&= \iota_{g_\epsilon^*}\Omega_\epsilon^* + \epsilon^{M+1}\,\mathbf{d}\phi_\epsilon\nonumber\\
& = \iota_{g_\epsilon^*}\widetilde{\Omega}_\epsilon + \epsilon^{M+1}\,\left(\mathbf{d}\phi_\epsilon - \iota_{g_\epsilon^*}\alpha_\epsilon\right),
\end{align}
or 
\begin{align}
\iota_{\widetilde{g}_\epsilon - g_\epsilon^*}\widetilde{\Omega}_\epsilon  = \epsilon^{M+1}\,\left(\mathbf{d}\phi_\epsilon - \iota_{g_\epsilon^*}\alpha_\epsilon\right).
\end{align}
If $\epsilon^{-d}\widetilde{P}_\epsilon(x) $ denotes the inverse of $\widetilde{\Omega}_\epsilon(x)$ restricted to $C_x$, where $\widetilde{P}_\epsilon = O(1)$ as $\epsilon\rightarrow 0$, we therefore conclude that
\begin{align}
\widetilde{g}_\epsilon - g_\epsilon^* &= \pi_K(\widetilde{g}_\epsilon - g_\epsilon^*) + \epsilon^{M-d+1}\,\widetilde{P}_\epsilon\cdot\left(\mathbf{d}\phi_\epsilon - \iota_{g_\epsilon^*}\alpha_\epsilon\right)\nonumber\\
& =  \pi_K(\widetilde{g}_\epsilon^\prime - g_\epsilon^*) + \epsilon^{M-d+1}\,\widetilde{P}_\epsilon\cdot\left(\mathbf{d}\phi_\epsilon - \iota_{g_\epsilon^*}\alpha_\epsilon\right).
\end{align}
This implies the desired result because $\widetilde{g}_\epsilon^\prime - g_\epsilon^* = O(\epsilon^{M-d +1})$ by hypothesis.
\end{proof}

\begin{example}
We may illustrate the application of Theorem \ref{ham_construction}, as well as some of the subtleties associated with the Theorem's technical hypotheses, in the context of the zero-drag Abraham-Lorentz equations.

Consider first the limiting Hamiltonian slow manifold reduction presented in Example \ref{limit_ham_example}. Here the approximate reduced presymplectic form and Hamiltonian are each of order $M=0$. Moreover the approximate reduced presymplectic form $\widetilde{\Omega}_\epsilon = \Omega_0^*$ has a $2$-dimensional $\epsilon$-independent kernel because
\begin{align}
\Omega_0^*(x)[\delta x_1,\delta x_2] = -\zeta \bm{B}(\bm{x})\cdot\delta\bm{x}_1\times\delta\bm{x}_2,
\end{align} 
which vanishes for all $\delta x_2$ when $\delta x_1 = (\delta\bm{x}_1,\delta v_{\parallel 1})$  satisfies $\delta\bm{x}_1\times\bm{B}(\bm{x}) = 0$, i.e. for all $\delta x_1$ of the form $\delta x_1 = (\delta s_1\bm{b}(\bm{x}),\delta v_{\parallel 1})$ with $(\delta s_1,\delta v_{\parallel 1})\in\mathbb{R}^2$. It follows that for each $x\in X$ we have the direct sum decomposition $T_xX = K_x\oplus C_x$ with
\begin{align}
K_x &= \{(\delta\bm{x},\delta v_\parallel)\mid \delta\bm{x}\times\bm{b}(\bm{x}) = 0\}\\
C_x &= \{(\delta \bm{x},\delta v_\parallel)\mid \delta v_\parallel = 0\text{ and }\bm{b}(\bm{x})\cdot \delta\bm{x}=0\},
\end{align}
where $K_x$ is the kernel of the $2$-form and $C_x$ is a $2$-dimensional complementary subspace. The linear mapping defined by $2$-form restricted to $C_x$, namely
\begin{align}
(\delta\bm{x}_\perp,0)\in C_x\mapsto -\zeta \bm{B}(\bm{x})\times\delta\bm{x}_\perp\cdot d\bm{x},
\end{align}
has an image comprising all $1$-forms of the form $\bm{w}_\perp\cdot d\bm{x}$ with $\bm{w}_\perp\cdot\bm{b}(\bm{x}) =0$, and an $O(1)$ inverse given by
\begin{align}
\bm{w}_\perp\cdot d\bm{x}\mapsto \frac{\zeta}{|\bm{B}|^2}\bm{B}(\bm{x})\times\bm{w}_\perp.
\end{align}
Moreover because $\widetilde{H}_\epsilon = H_0^* = 0$ the differential $\mathbf{d}\widetilde{H}_\epsilon$ is trivially in the kernel of the $2$-form. These observations confirm that $\widetilde{\Omega}_\epsilon$ and $\widetilde{H}_\epsilon$ satisfy the hypotheses (a-c) in Theorem \ref{ham_construction} with $d=0$. The Theorem therefore asserts that the sum of $\widetilde{g}^\perp_\epsilon$ and $\pi_K(\widetilde{g}_\epsilon^\prime)$ ought to comprise an order $M-d = 0$ Hamiltonian slow manifold reduction of the zero-drag equations. This is indeed the case for the following reason. Because $\iota_{\widetilde{g}_\epsilon^\perp}\widetilde{\Omega}_\epsilon = \iota_{\widetilde{g}_\epsilon^\perp}\Omega_0^* = \mathbf{d}\widetilde{H}_\epsilon =  0$ it must be true that $\widetilde{g}_\epsilon^\perp(x)\in K_x$, which by definition of $\widetilde{g}_\epsilon^\perp$ can only be satisfied if $\widetilde{g}_\epsilon^\perp=0$. Therefore if we use the limiting slow manifold reduction to define $\widetilde{g}_\epsilon^\prime(x) = (v_\parallel\,\bm{b}(\bm{x}),0)$ in Theorem \ref{ham_construction} an explicit expression for $\widetilde{g}_\epsilon = (\dot{\widetilde{\bm{x}}},\dot{\widetilde{v}}_\parallel) = \pi_K(\widetilde{g}_\epsilon^\prime)$ is given by $\dot{\widetilde{\bm{x}}} = v_\parallel\bm{b}(\bm{x})$, $\dot{\widetilde{v}}_\parallel = 0$, which agrees with the limiting slow manifold reduction precisely.

Next consider the problem of constructing a second-order slow manifold reduction of the zero-drag equations. We previously constructed a possibly-non-Hamiltonian second-order slow manifold reduction for this system in Eq.\,\eqref{2nd_order_smr_1}-\eqref{2nd_order_smr_2}. In order to develop a Hamiltonian reduction of the same order it is first necessary to identify an approximate reduced presymplectic form $\widetilde{\Omega}_\epsilon$ and reduced Hamiltonian $\widetilde{H}_\epsilon$ that are sufficiently close to $\Omega_\epsilon^*$ and $H_\epsilon^*$ to ensure that Theorem \ref{ham_construction} provides a second-order reduction. Apparently this requires choosing the order of accuracy $M$ for $\widetilde{\Omega}_\epsilon$ and $\widetilde{H}_\epsilon$ to satisfy $M\geq 2+ d$, where $d$ is the scaling exponent characterizing the inverse of the $2$-form $\widetilde{\Omega}_\epsilon$. This observation would provide a straightforward means to selecting $M$ were it not for one complicating factor: the exponent $d$ depends on $M$. Before choose $M$ we must therefore understand how $d$ varies with $M$.

When $M=0$ we just saw that $d=0$. When $M=1$ an approximate reduced presymplectic form is given by $\widetilde{\Omega}_\epsilon = \Omega_0^* + \epsilon\,\Omega_1^*$ with $\Omega_0^*$ and $\Omega_1^*$ given in Eqs.\,\eqref{omega0_zero_drag}-\eqref{omega1_zero_drag}, or
\begin{align}
\widetilde{\Omega}_\epsilon(x)[\delta x_1,\delta x_2]=& -(\zeta\,\bm{B}+ \epsilon\,v_\parallel\,\nabla\times\bm{b})\cdot\delta\bm{x}_1\times\delta\bm{x}_2\nonumber\\
& +\epsilon\, (\delta\bm{x}_1\delta v_{\parallel 2} - \delta\bm{x}_2\,\delta v_{\parallel 1})\cdot\bm{b}.
\end{align}
The equation $\iota_{\delta x}\widetilde{\Omega}_\epsilon = \bm{\omega}\cdot d\bm{x} + \alpha\,dv_\parallel$ with $\delta x = (\delta\bm{x},\delta v_\parallel)$ is therefore equivalent to the linear system
\begin{align}
-(\zeta\,\bm{B}+ \epsilon\,v_\parallel\,\nabla\times\bm{b})\times \delta\bm{x} - \epsilon\,\delta v_\parallel\,\bm{b} = \bm{\omega}&\\
\epsilon\,\delta\bm{x}\cdot \bm{b} = \alpha&,
\end{align}
which can be solved for $\delta x$ as
\begin{align}
\delta\bm{x} = &\frac{\alpha}{\epsilon}\frac{\bm{B} + \epsilon\,\zeta\,v_\parallel\,\nabla\times\bm{b}}{|\bm{B}| + \epsilon\,\zeta\,v_\parallel\bm{b}\cdot\nabla\times\bm{b}}+\zeta\,\frac{\bm{b}\times\bm{\omega}}{|\bm{B}| + \epsilon\,\zeta\,v_\parallel\bm{b}\cdot\nabla\times\bm{b}}\\
\delta v_\parallel = &-\frac{1}{\epsilon} \frac{(\bm{B} + \epsilon\,\zeta\,v_\parallel\,\nabla\times\bm{b})\cdot\bm{\omega}}{|\bm{B}| + \epsilon\,\zeta\,v_\parallel\bm{b}\cdot\nabla\times\bm{b}}.
\end{align}
This shows that $\widetilde{\Omega}_\epsilon$ is invertible everywhere except where $D = |\bm{B}| + \epsilon\,\zeta\,v_\parallel\,\bm{b}\cdot\nabla\times\bm{b} = 0$ with an $O(\epsilon^{-1})$ inverse. We therefore conclude that away from the singular set $D=0$ (which we will excise from $(\bm{x},v_\parallel)$-space without further discussion) the exponent $d=1$. Moreover if $\widetilde{\Omega}_\epsilon^\prime$ is any other closed $2$-form such that $\widetilde{\Omega}_\epsilon - \widetilde{\Omega}_\epsilon^\prime = \epsilon^2 \beta_\epsilon$, with $\beta_\epsilon = O(1)$, then 
\begin{align}
(\widetilde{\Omega}_\epsilon^\prime)^{-1} &= (\widetilde{\Omega}_\epsilon - \epsilon^2\beta_\epsilon)^{-1}\nonumber\\
& = (\widetilde{\Omega}_\epsilon)^{-1}(1 - \epsilon^2 \widetilde{\Omega}_\epsilon^{-1}\beta_\epsilon)^{-1}\nonumber\\
& = (\widetilde{\Omega}_\epsilon)^{-1}(1 + O(\epsilon))\nonumber\\
& = O(\epsilon^{-1}),
\end{align}
which shows that $d=1$ for all $M\geq 1$. We say that the exponent $d$ \emph{stablizes} at $M=1$. 

Now that we have determined $d = 0$ when $M=0$ and $d = 1$ when $M\geq 1$, we see that the inequality $M\geq 2 + d$ may be replaced with $M\geq 3$. In other words we need to choose $\widetilde{\Omega}_\epsilon$ and $\widetilde{H}_\epsilon$ with third-order accuracy in order to obtain a Hamiltonian slow manifold reduction with second-order accuracy. To that end we set $\widetilde{\Omega}_\epsilon = \Omega_0^* + \epsilon\,\Omega_1^* + \epsilon^2\,\Omega_2^* + \epsilon^3\,\Omega_3^*$ and $\widetilde{H}_\epsilon = H_0^* + \epsilon\,H_1^* + \epsilon^2\,H_2^* + \epsilon^3\,H_3$ and consider the vector field $\widetilde{g}_\epsilon$ defined by Hamilton's equations $\iota_{\widetilde{g}_\epsilon}\widetilde{\Omega}_\epsilon = \mathbf{d}\widetilde{H}_\epsilon$. Setting $\widetilde{g}_\epsilon = (\dot{\widetilde{\bm{x}}}_\epsilon,\dot{\widetilde{v}}_{\parallel\epsilon})$, we find Hamilton's equations are equivalent to 
\begin{align}
-\zeta\bm{B}^*_\epsilon\times \dot{\widetilde{\bm{x}}}_\epsilon - \epsilon\,\dot{\widetilde{v}}_{\parallel\epsilon}\bm{b}^*_\epsilon = \epsilon^3\,\nabla\frac{1}{2}|\bm{v}_c|^2&\label{zero_drag_trunc_ham1}\\
\epsilon\,\bm{b}^*_\epsilon\cdot\dot{\widetilde{\bm{x}}}_\epsilon = \epsilon\, v_\parallel + \epsilon^3 \bm{v}_c\cdot\partial_{v_\parallel}\bm{v}_c,\label{zero_drag_trunc_ham2}
\end{align}
where we have defined
\begin{align}
\bm{B}^*_\epsilon =&\bm{B} + \epsilon\,\zeta\,v_\parallel \nabla\times\bm{b} \nonumber\\
& + \epsilon^2\,\zeta\,\nabla\times[\bm{v}_c+\epsilon\bm{v}_{\perp 2}^*]\\
\bm{b}^*_\epsilon =& \bm{b} + \epsilon\,\partial_{v_\parallel}\bm{v}_c + \epsilon^2\,\partial_{v_\parallel}\bm{v}_{\perp2}^*.
\end{align}
By applying a cross product with $\bm{b}_\epsilon^*$ and a dot product with $\bm{B}_\epsilon^*$ to Eq.\,\eqref{zero_drag_trunc_ham1} these equations may be solved explicitly for $(\dot{\widetilde{\bm{x}}}_\epsilon,\dot{\widetilde{v}}_{\parallel\epsilon})$, leading to
\begin{align}
\dot{\widetilde{\bm{x}}}_\epsilon &= \frac{\bm{B}_\epsilon^*}{\bm{B}_\epsilon^*\cdot\bm{b}_\epsilon^*}\left(v_\parallel + \epsilon^2\frac{1}{2} \partial_{v_\parallel}|\bm{v}_c|^2\right) + \epsilon^3\,\zeta\frac{\bm{b}_\epsilon^*\times\nabla\frac{1}{2}|\bm{v}_c|^2}{\bm{B}_\epsilon^*\cdot\bm{b}_\epsilon^*}\\
\dot{\widetilde{v}}_{\parallel\epsilon} & = - \epsilon^2 \frac{\bm{B}_\epsilon^*\cdot\nabla\frac{1}{2}|\bm{v}_c|^2}{\bm{B}_\epsilon^*\cdot\bm{b}_\epsilon^*}.
\end{align}
We leave it as an exercise for the reader to verify that the difference between these expressions and our earlier non-Hamiltonian second-order slow manifold reduction \eqref{2nd_order_smr_1}-\eqref{2nd_order_smr_2} is $O(\epsilon^3)$. We also challenge the reader to demonstrate that were we to replace $\widetilde{\Omega}_\epsilon $ with $\Omega_0^* +\epsilon\,\Omega_1^* + \epsilon^2\,\Omega_2^*$ the resulting Hamiltonian slow manifold reduction (with the same approximate reduced Hamiltonian) would \emph{still} be second-order. This result relies on special features of the zero-drag system that are not shared by generic Hamiltonian fast-slow systems, and is therefore undetectable by Theorem \ref{ham_construction}.

\end{example}

% discuss symplectic rectification here

\subsection{Inherited Symmetries}
Now that we have established a method for constructing Hamiltonian slow manifold reductions of any order we would like to explain some of the practical benefits of doing so. Perhaps the deepest benefit is the fact that Hamiltonian slow manifold reductions automatically possess an absolute integral invariant. If $D_0\subset X$ is any closed $2$-dimensional surface embedded in $x$-space and $D_t$ is the $t$-second evolution of that surface defined by flowing along the Hamiltonian slow manifold reduction $\widetilde{g}_\epsilon$ then 
\begin{align}
\int_{D_0}\widetilde{\Omega}_\epsilon = \int_{D_t}\widetilde{\Omega}_\epsilon.\label{gen_frozen_in}
\end{align}
We highlighted a particular instance of this ``generalized frozen-in law" in Example \ref{limit_ham_example}. However many of the consequences of Eq.\,\eqref{gen_frozen_in} are understood today only in very abstract terms. See for instance Ref.\,\citep{Gosson_2009} for a discussion of so-called symplectic capacities, which may be viewed as profound, although abstract generalizations of Liouville's theorem that follow from Eq.\,\eqref{gen_frozen_in}. One consequence of Eq.\,\eqref{gen_frozen_in} that \emph{is} accessible concerns Noether's theorem on the relationship between symmetries and conservation laws. As we will show, because Hamiltonian slow manifold reductions obey Eq.\,\eqref{gen_frozen_in} they are automatically systems for which Noether's theorem applies. Thus one practical benefit of using Hamiltonian slow manifold reductions over non-Hamiltonian slow manifold reductions (when Hamiltonian structure is available) is that exact conservation laws can be built into the reduction in a systematic manner, thereby helping to circumvent the difficulty highlighted in Example \ref{energy_breakdown_example}.

Noether's theorem guarantees that (almost) every symmetry of a Hamiltonian system implies the existence of a conservation law. Therefore within the realm of Hamiltonian systems the task of ensuring that a given system obeys a conservation law may often be recast as the task of ensuring that the system possesses a certain symmetry. In particular the task of constructing slow manifold reductions that inherit conservation laws from their Hamiltonian parent model may often be recast as the problem of finding Hamiltonian slow manifold reductions that inherit the parent model's symmetries. While not obviously the case, this reformulation is actually useful because it leads to a generally applicable and more-or-less complete resolution of the first-integral inheritance problem (c.f. Example \ref{energy_breakdown_example}) for Hamiltonian fast-slow systems. We are not aware of any similarly powerful approach to resolving this issue that \emph{does not} make use of Noether's theorem. For instance while one way to ensure that a slow manifold reduction inherits conservation laws from its parent model is to first express the parent model as a skew-gradient system (see Ref.\,\citep{McLachlan_1999} for a discussion of how general systems with first integrals can be written as skew-gradient system), the skew-symmetric tensor that must be identified for this purpose often has bothersome singularities that render the technique unwieldy; this drawback is absent in the approach based on Noether's theorem that we are about to describe. (We remark however that the skew-gradient technique applies to \emph{any} system with first integrals, not just the Hamiltonian variety.)

The rest of this Section will be devoted to describing in detail an approach to constructing Hamiltonian slow manifold reductions that inherit the symmetries, and therefore the Noether conservation laws from the underlying fast-slow system. First we will describe the interplay between the symmetry group for a (possibly non-Hamiltonian) fast-slow system and that system's formal slow manifold. Then we will exploit this interplay to prove that a certain class of Hamiltonian slow manifold reductions automatically inherit symmetries from their parent model.

We will limit our discussion to fast-slow systems that admit $\epsilon$-independent symmetries. We will refer to such systems as fast-slow systems with symmetry.

\begin{definition}[fast-slow systems with symmetry]
A \emph{fast-slow system with symmetry} is a fast slow system $\dot{x} = g_\epsilon(x,y)$, $\epsilon\,\dot{y} = f_\epsilon(x,y)$ whose flow commutes with an $\epsilon$-independent action of some connected Lie group $G$. In other words if $\mathcal{F}_{t}^\epsilon:X\times Y\rightarrow X\times Y$ is the time-$t$ flow of the fast-slow system and $\Phi_g:X\times Y\rightarrow X\times Y$ is a left $\epsilon$-independent $G$-action on $X\times Y$ then $\mathcal{F}_t^\epsilon\circ \Phi_g = \Phi_g\circ \mathcal{F}_t^\epsilon$ for each $g\in G$ and $\epsilon> 0$.
\end{definition}

The first remarkable feature of fast-slow systems with symmetry is that their limiting slow manifolds are automatically $G$-invariant, as show in the following proposition.

\begin{proposition}[limiting slow manifolds inherit symmetry]
The limiting slow manifold $S_0= \{(x,y)\in X\times Y\mid y = y_0^*(x)\}$ associated with a fast-slow system with symmetry is $G$-invariant, i.e. $\Phi_g(S_0) = S_0$ for each $g\in G$.
\end{proposition}

\begin{proof}
Let $W_\epsilon = (g_\epsilon,f_\epsilon/\epsilon)$ be the fast-slow system's infinitesimal generator. Because the fast-slow system's flow commutes with the $G$-action it must be true that $\Phi_g^*W_\epsilon = W_\epsilon$ for each $\epsilon > 0$. Therefore $V_\epsilon = \epsilon W_\epsilon = (\epsilon g_\epsilon,f_\epsilon)$ also satisfies $\Phi_g^*V_\epsilon = V_\epsilon$ for each $\epsilon>0$. Since $V_\epsilon$ is a smooth function of $\epsilon$ in a neighborhood of $0$ this means that $\Phi_g^*V_0 = V_0$, where $V_0 = (0,f_0)$, and that the $G$-action commutes with the flow of $V_0$. An immediate consequence is that $V_0$'s fixed point set $C_0 = \{(x,y)\in X\times Y\mid V_0(x,y) = 0\}$ is $G$-invariant. But $(x,y)\in C_0$ if and only if $f_0(x,y) = 0$, i.e. $y = y_0^*(x)$. Therefore $C_0 = S_0$ is $G$-invariant. 
\end{proof}

\begin{remark}
While we always assume that $f_0(x,y) =0$ has the unique solution $y = y_0^*(x)$ this proposition is also valid in more general fast-slow systems provided $S_0$ is redefined as the zero level set of $f_0$. 
\end{remark}

\begin{corollary}\label{limit_induced_cor}
Given any fast-slow system with symmetry there is a left $\epsilon$-independent $G$-action $\varphi^*_{0,g}:X\rightarrow X$on $X$ defined by the formula $\Phi_g(x,y_0^*(x)) =(\varphi^*_{0,g}(x),y_0^*(\varphi^*_{0,g}(x))) $.
\end{corollary}

More generally the formal slow manifold associated with a fast-slow system with symmetry is also $G$-invariant. The following proof of this result makes use of two basic features of fast-slow systems with symmetry: (1) $G$ naturally acts on the space of invariant manifolds, (2) formal power series solutions of the invariance equation are unique.

\begin{theorem}[formal slow manifolds inherit symmetry]
Let $y_\epsilon^*$ be the formal slow manifold associated with a fast-slow system with symmetry, and $\Phi_g(x,y) = (\varphi_g(x,y),\psi_g(x,y))$ the associated $G$-action. The formal mapping $x\mapsto \varphi_{\epsilon,g}^*(x) = \varphi_g(x,y_\epsilon^*(x))$ satisfies the identities $\varphi_{\epsilon,e}^*(x) = x$, $\varphi_{\epsilon,g_1g_2}^*(x) = \varphi_{\epsilon,g_1}^*(\varphi_{\epsilon,g_2}^*(x))$ for all $g_1,g_2\in G$ and $e\in G$ the identity element. In other words $\varphi_{\epsilon,g}^*$ defines a formal $G$-action on $X$.
In addition we have the equality of formal power series 
\begin{align}
y_{\epsilon}^*(x) = \psi_g([\varphi_{\epsilon,g}^*]^{-1}(x),y_\epsilon^*([\varphi_{\epsilon,g}^*]^{-1}(x)))\label{y_g},
\end{align} 
for each $g\in G$.

\end{theorem}
\begin{remark}
The idea behind this Theorem may be understood as follows. Suppose that $h:X\rightarrow Y$ is some smooth function such that the graph $H = \{(x,y)\in X\times Y\mid y = h(x)\}$ is $G$-invariant. Then $\Phi_g:X\times Y\rightarrow X\times Y$ induces a $G$-action on $X$, $\varphi_g^*:X\rightarrow X$, given by $\varphi_g^*(x) = \varphi_g(x,h(x))$.
Moreover If $(x,h(x))\in H$ is a point in $H$ then $\Phi_g(x,h(x)) = (\varphi_g(x,h(x)),\psi_g(x,h(x)))= (\overline{x},\overline{y})$ must be another point in $H$. This can only be the case if $\overline{y} = h(\overline{x})$, or $\psi_g(x,h(x)) = h(\varphi_g(x,h(x)))$. This last relationship may also be written $h(\overline{x}) = \psi_g([\varphi_g^*]^{-1}(\overline{x}),h([\varphi_g^*]^{-1}(\overline{x})))$ because $\varphi_g^*$ must be invertible for each $g\in G$. These observations show that the Theorem saying roughly that the ``graph" of $y_\epsilon^*$ is $G$-invariant. The precise statement of the Theorem cannot be this simple, however, because the ``graph" of $y_\epsilon^*$ does not really exist.

\end{remark}

\begin{proof}
The proof is based on the following heuristic argument that can't be taken literally because the ``graph" of $y_\epsilon^*$ is technically ill-defined. Let $\Gamma$ be an invariant manifold. Because the fast-slow system's flow and the $G$-action commute the image of $\Gamma$ under $\Phi_g$ must be another invariant manifold for each $g\in G$. Suppose now that $\Gamma$ happens to be given as the graph of some function $h_\epsilon(x)\in Y$ that is smooth in $(x,\epsilon)$ and that $\Phi_g(\Gamma)$ is also a graph for each $g\in G$. Let $h_{\epsilon,g}(x)$ be the graphing function; it is necessarily smooth in $(x,\epsilon)$. Because $\Phi_g(\Gamma)$ is an invariant manifold the function $h_{\epsilon,g}(x)$ must satisfy the invariance equation for each $g\in G$. In particular the formal power series expansion of $h_{\epsilon,g}(x)$ must be a formal solution of the invariance equation, i.e. $h_{\epsilon,g}$ expanded in its power series must be a formal slow manifold. But we know that the formal slow manifold is unique, which implies $y_\epsilon^* = h_{\epsilon,g}$ as formal power series for each $g\in G$. The Theorem may therefore be proved by setting $h_\epsilon = y_\epsilon^*$ in some formal sense.

To make the preceding argument precise we begin by observing that the pushforward of $W_\epsilon = (g_\epsilon,f_\epsilon/\epsilon)$ along $\Phi_g$ must be equal to $W_\epsilon$ for any fast-slow system with symmetry. In symbols we have 
\begin{align}
D\Phi_g(x,y) [W_\epsilon(x,y)] = W_\epsilon(\Phi_g(x,y)),
\end{align}
for each $(x,y)\in X\times Y$, $g\in G$, and $\epsilon\in \mathbb{R}$. Upon making the substitution $y = y_\epsilon^*(x)$ this establishes the following pair formal power series identities,
\begin{align}
D\varphi_g(x,y_\epsilon^*)[(g_\epsilon^*(x),Dy_\epsilon^*(x)[g_\epsilon^*(x)])] &= g_\epsilon(\Phi_g(x,y_\epsilon^*))\label{inv_id_one}\\
D\psi_g(x,y_\epsilon^*)[(g_\epsilon^*(x),Dy_\epsilon^*(x)[g_\epsilon^*(x)])] & = \frac{1}{\epsilon}f_\epsilon(\Phi_g(x,y_\epsilon^*)).\label{inv_id_two}
\end{align}
Next we observe that the formal invertibility of $x\mapsto\varphi_{\epsilon,g}^*(x)$ follows from the fact that when $\epsilon = 0$ the map is $x\mapsto \varphi_{0,g}^*(x)$, which is a $G$-action by Corollary \ref{limit_induced_cor}.
Having established the invertibility of  $x\mapsto\varphi_{\epsilon,g}^*(x)$ we will now use the identities \eqref{inv_id_one}-\eqref{inv_id_two} to verify directly the the formal power series $y_{\epsilon,g}^*(x) = \psi_g([\varphi_{\epsilon,g}^*]^{-1}(x),y_\epsilon^*([\varphi_{\epsilon,g}^*]^{-1}(x)))\label{y_g}$ satisfies the invariance equation for each $g\in G$. By the definition of $y_{\epsilon,g}^*$,
\begin{align}
y_{\epsilon,g}^*(\overline{x}) = \psi_g(x,y_\epsilon^*(x)),\label{four_ten}
\end{align} 
where $\overline{x} = \varphi_{\epsilon,g}^*(x)$. Implicitly differentiating this equation in $x$ along the direction $g_\epsilon^*(x)$ then gives
\begin{align}
Dy_{\epsilon,g}^*(\overline{x})[D\varphi_{\epsilon,g}^*(x)[g_\epsilon^*(x)]] = D\psi_g(x,y_\epsilon^*)[(g_\epsilon^*(x),Dy_\epsilon^*(x)[g_\epsilon^*(x)])].\label{four_eleven}
\end{align}
By the chain rule and identity \eqref{inv_id_one} $D\varphi_{\epsilon,g}^*(x)[g_\epsilon^*(x)] = D\varphi_g(x,y_\epsilon^*)[(g_\epsilon^*(x),Dy_\epsilon^*(x)[g_\epsilon^*(x)])] = g_\epsilon(\Phi_g(x,y_\epsilon^*))$. Identity \eqref{inv_id_two} in conjunction with Eq.\,\eqref{four_eleven} therefore imply
\begin{align}
Dy_{\epsilon,g}^*(\overline{x})[g_\epsilon(\Phi_g(x,y_\epsilon^*))] = \frac{1}{\epsilon }f_\epsilon (\Phi_g(x,y_\epsilon^*)).\label{four_twelve}
\end{align}
But $\Phi_g(x,y_\epsilon^*) = (\varphi_{\epsilon,g}^*(x),\psi_g(x,y_\epsilon^*)) = (\varphi_{\epsilon,g}^*(x),y_{\epsilon,g}^*(\overline{x})) $ by Eq.\,\eqref{four_ten}, whence Eq.\,\eqref{four_twelve} becomes
\begin{align}
\epsilon\,Dy_{\epsilon,g}^*(\overline{x})[g_\epsilon(\overline{x},y_{\epsilon,g}^*(\overline{x}))] = f_\epsilon(\overline{x},y_{\epsilon,g}^*(\overline{x})),
\end{align}
which is the invariance equation, as desired.

Because the formal power series $y_{\epsilon,g}^*$ satisfies the invariance equation for each $g\in G$ and slow manifolds are unique it follows that $y_{\epsilon,g}^* = y_\epsilon^*$ independent of $g$. This establishes the $G$-invariance of the formal slow manifold. To see now that $\varphi_{\epsilon,g}^*(x)$ defines a formal $G$-action on $X$ we note first that because $\Phi_g$ is a $G$-action we have $\Phi_{g_1g_2}(x,y) = \Phi_{g_1}(\Phi_{g_2}(x,y))$ for all $(x,y)\in X\times Y$ and $g_1,g_2\in G$. Substituting $y = y_\epsilon^*(x)$ into this identity therefore implies the formal power series identity $\Phi_{g_1g_2}(x,y_\epsilon^*) = \Phi_{g_1}(\Phi_{g_2}(x,y_\epsilon^*))$. But because $\Phi_g(x,y_\epsilon^*) = (\varphi_{\epsilon,g}^*(x),\psi_g(x,y_\epsilon^*)) = (\varphi_{\epsilon,g}^*(x),y_{\epsilon,g}^*(\varphi_{\epsilon,g}^*(x))) $, as we have already mentioned, we therefore have 
\begin{align}
\Phi_{g_1g_2}(x,y_\epsilon^*) &= (\varphi_{\epsilon,g_1g_2}^*(x),y_{\epsilon,g}^*(\varphi_{\epsilon,g_1g_2}^*(x))) \nonumber\\
& = (\varphi_{\epsilon,g_1g_2}^*(x),y_{\epsilon}^*(\varphi_{\epsilon,g_1g_2}^*(x))),
\end{align}
and 
\begin{align}
 \Phi_{g_1}(\Phi_{g_2}(x,y_\epsilon^*)) &= \Phi_{g_1} (\varphi_{\epsilon,g_2}^*(x),y_{\epsilon,g}^*(\varphi_{\epsilon,g_2}^*(x))) \nonumber\\
 & =  \Phi_{g_1} (\varphi_{\epsilon,g_2}^*(x),y_{\epsilon}^*(\varphi_{\epsilon,g_2}^*(x))) \nonumber\\
 & = (\varphi_{\epsilon,g_1}^*(\varphi_{\epsilon,g_2}^*(x)),y_\epsilon^*(\varphi_{\epsilon,g_1}^*(\varphi_{\epsilon,g_2}^*(x)))),
\end{align}
which gives the desired result.

\end{proof}

% establish formal power series SM Hamiltons equations, inheritance of symmetries, symplectic rectification.

The theorem shows that there must be a formal invariant manifold in the quotient space $X\times Y/ G$. Therefore a good way to construct $G$-invariant slow manifolds of any desired order is to (1) construct a truncation of the corresponding object in the quotient space, and then (2) define the slow manifold in the unreduced space as the preimage of truncated invariant manifold in the quotient along the quotient projection map. Hamiltonian slow manifold reductions constructed by pulling back the symplectic form and Hamiltonian to the resulting $G$-invariant slow manifold will then automatically inherit all Noether invariants because these pulled-back objects will be invariant with respect to the $G$-action restricted to the slow manifold.

%Terse description then given example of formalism applied to zero-drag regime. Be sure to highlight: baking conservation properties in slow manifold reductions of any order, Noether's theorem and slow manifolds, ...

\section{Application: Quasineutral Kinetic Plasmas\label{QN_application}}
The Vlasov-Maxwell (VM) system of equations from the kinetic theory of plasmas is an example of a degenerate fast-slow dynamical system. It is infinite dimensional but nevertheless the formal techniques presented above can be applied to obtain slow-manifold approximations in the quasi-neutral limit. (See Ref.\,\citep{Tronci_2015} for a variational formulation of the collisionless limiting neutral model.) The VM equations are given by
\begin{align}
\partial_tf_\sigma&=\sum_{\bar{\sigma}}C_\sigma\left(f_\sigma,f_{\bar{\sigma}}\right)-\nabla_x\cdot(\bm{v}f_\sigma)\\
\,&-\partial_{\bm{v}}\cdot\left(\frac{e_\sigma}{m_\sigma}\left(\bm{E}+\bm{v}\times \bm{B}\right)f_\sigma\right)\nonumber\\
\,&\epsilon_0\partial_t\bE+\sum_{\sigma}e_\sigma\int \bv f_\sigma\,d^3\bv=\mu^{-1}_0\nabla\times \bB\\
\,&\partial_t \bB=-\nabla\times \bE\\
\,&\epsilon_0\nabla\cdot \bE=\sum_{\sigma}e_\sigma\int f_\sigma\,d^3\bv
\end{align}
The first of these equations is known as the Vlasov equation. Here $f_\sigma(x,\bv)$ is the distribution of particle species and $C(f_\sigma,f_{\bar{\sigma}})$ is a collision term between particles. A derivation of the VM system is given in chapter 22 of \citep{Goldston_book_1995}. Our goal in this section will be to find a first-order asymptotic expansion for the electric field $\bE$ as a slow manifold. We make the following substitution that $f_\sigma(x,\bv)=n_\sigma(x)\rho_\sigma(x,\bw)$ where
\begin{align}
\bw=\bv-\bu_\sigma(x)\\
\int \rho_\sigma\,d^3\bw=1\\
\int \bw\rho_\sigma\,d^3\bw=0\\
n_\sigma=\int f_\sigma\,d^3\bv
\end{align}
One can quickly deduce that
\begin{align}
n_\sigma(x)\bu_\sigma(x)=\int \bv f_\sigma\,d^3\bv\label{num_Mvel}
\end{align}
The $\bu_\sigma(x)$ is a measure of mean velocity for particle species $\sigma$. The $n_\sigma(x)$ is a measure of particle number of species $\sigma$ and $\rho_\sigma(x,\bw)$ is a corresponding density function over velocity for each value of position $x$ in space. We also note here that the collision terms satisfy some conservation laws given by
\[
\int \sum_{\bar{\sigma}}C_\sigma(f_\sigma,f_{\bar{\sigma}})\,d^3\bv=0
\]
\begin{center}
Conservation of particle number
\end{center}
\[
\sum_{\sigma}\int m_\sigma \bm{v}\sum_{\bar{\sigma}}C_\sigma(f_\sigma,f_{\bar{\sigma}})\,d^3\bv=0
\]
\begin{center}
Conservation of momentum
\end{center}
\[
\sum_{\sigma}\int \frac{1}{2} m_\sigma |\bm{v}|^2\sum_{\bar{\sigma}}C_\sigma(f_\sigma,f_{\bar{\sigma}})\,d^3\bv=0
\]
\begin{center}
Conservation of energy
\end{center}
Now that we've replaced the distribution function $f_\sigma$ with $\rho_\sigma, \bu_\sigma$ and $n_\sigma$ we'll need to derive their evolution equations to replace Vlasov equation in the VM system. The simplest of these evolutions will be for the particle number $n_\sigma$, and indeed it goes as:
\begin{align}
\partial_t n_\sigma&=\int\partial_t f_\sigma\,d^3\bv\nonumber\\
\,&=\int \sum_{\bar{\sigma}}C_\sigma(f_\sigma,f_{\bar{\sigma}})\,d^3\bv-\nabla_x\cdot\int \bv f_\sigma\,d^3\bv\nonumber\\
\,&\phantom{=}-\int \frac{e_\sigma}{m_\sigma}(\bE+\bv\times \bB)\cdot\partial_{\bv} f_\sigma\,d^3\bv\nonumber\\
\,&=-\nabla_x\cdot\int \bv f_\sigma\,d^3\bv\nonumber\\
\,&+\int f_\sigma\partial_{\bv} \cdot\left(\frac{e_\sigma}{m_\sigma}(\bE+\bv\times\bB)\right)\,d^3\bv\nonumber\\
\partial_t n_\sigma&=-\nabla_x\cdot(n_\sigma\bu_\sigma)
\end{align}
We note here that integration by parts will be employed frequently and that when applied to the Lorentz force term, the $\bv$ divergence will return zero thanks to the cross-product and the fact that $\bE$ and $\bB$ have no $\bv$ dependence. Next, starting from the time derivative of equation (\ref{num_Mvel}) we can arrive at 
\begin{align}
\partial_t(\bu_\sigma)&=\frac{1}{n_\sigma}\int \bw\sum_{\bar{\sigma}}C_\sigma(n_\sigma\rho_\sigma,n_{\bar{\sigma}}\rho_{\bar{\sigma}})d^3\bw\\
\,&\phantom{=}-\bu_\sigma\cdot\nabla\bu_\sigma+\frac{e_\sigma}{m_\sigma}(\bE+\bu_\sigma\times \bB)\nonumber\\
\,&\phantom{=}-\frac{1}{n_\sigma}\nabla_x\cdot\bbP_\sigma\nonumber 
\end{align}
Here $\bbP_\sigma=\int n_\sigma\rho_\sigma\bw\bw\,d^3\bw$ is a pressure term. Finally we need the evolution of the density $\rho_\sigma$. This results in 
\begin{align}
\partial_t\rho_\sigma&=\frac{1}{n_{\sigma}}\left(\int \bw\sum_{\bar{\sigma}}C_{\sigma\bar{\sigma}}d^3\bw\right)\cdot\partial_{\bw}\rho_{\sigma}\\
\,&+\frac{1}{n_\sigma}\sum_{\bar{\sigma}}C_{\sigma\bar{\sigma}}+\frac{1}{n_\sigma}\nabla\cdot(n_\sigma\rho_\sigma\bu_\sigma)+\frac{1}{n_\sigma}\nabla\cdot(\bw n_\sigma\rho_\sigma)\nonumber\\
\,&-\partial_{\bw}\cdot(\bw\cdot\nabla\bu_\sigma n_\sigma\rho_\sigma)-\frac{1}{n_\sigma}(\nabla_x\cdot\bbP_\sigma)\cdot\partial_{\bw}\rho_\sigma\nonumber\\
\,&-\frac{e_\sigma}{m_\sigma}(\bw\times \bB)\cdot\partial_{\bw}\rho_\sigma\nonumber
\end{align}
We now have all the needed evolution equations. We'll now choose to reorganize VM and the above evolution equations into a nondimensional form that will be tied to our choice of small parameter. 
\begin{align*}
\begin{split}
t&=T\bar{t}\\
\bE(x)&=E_0\bar{\bE}\left(\frac{x}{L}\right)\\
\bB(x)&=B_0\bar{\bB}\left(\frac{x}{L}\right)
\end{split}
\quad
\begin{split}
n_\sigma(x)&=n_{0}\bar{n}_\sigma\left(\frac{x}{L}\right)\\
\rho_\sigma(x,\bw)&=\frac{1}{V^3_{th\sigma}}\bar{\rho}_\sigma\left(\frac{x}{L},\frac{\bw}{V_{th\sigma}}\right)\\
\bu_\sigma(x)&=u_{\sigma 0}\bar{\bu}_\sigma\left(\frac{x}{L}\right)
\end{split}
\end{align*}
Here $L$ and $T$ are the characteristic length and time-scales respectively and $V_{th\sigma}$ is the thermal velocity of species $\sigma$. We'll now scale these quantities further with respect to a small quantity $\epsilon$ that is proportional to the permittivity of free space $\epsilon_0$. 
\begin{align}
\ds L&=Tc\sqrt{\epsilon} &\phantom{BIG} \ds B_0&=\frac{m_i}{q_i}\frac{1}{T}\\
\ds u_0&=c\sqrt{\epsilon}  & \phantom{BIG} \ds E_0&=\frac{m_i}{q_i}\frac{1}{T}c\sqrt{\epsilon}\\
\ds V_{th}&=c\sqrt{\epsilon}  & \phantom{BIG} V_{th}&=V_{thi}=V_{the}\sqrt{m_r}\\
\ds n_0&=\frac{\epsilon_0m_i}{q_i^2}\frac{1}{\epsilon T^2} & \phantom{BIG} q_e&=-q_i
\end{align}
Note that the speeds in this scaling are all non-relativistic. If collisions are neglected this scaling recovers the neutral model in \cite{Tronci_2015} in a non-dimensional form so that the zeroth order approximation to the slow manifold corresponds to this neutral limit. The first order term in the approximation will recover what one could truly call a quasi-neutral approximation. In this scaling the Maxwell equations in VM become 
\begin{align}
\epsilon\nabla\cdot\bar{\bE}&=\bar{n}_i-\bar{n}_e\\
\epsilon\partial_t\bar{\bE}&=\nabla\times\bar{\bB}+\bar{n}_e\bar{\bu}_e-\bar{n}_i\bar{\bu}_i\label{Ampere}\\
\partial_t\bar{\bB}&=-\nabla\times\bar{\bE}\label{Faraday}
\end{align}
The evolution equations resulting from the Vlasov equation become
\begin{align}
\partial_t\bar{n}_\sigma&=-\nabla\cdot(\bar{n}_\sigma\bar{\bu}_\sigma)\\
\partial_t\bar{\bu}_e&=\frac{T\nu}{\sqrt{m_r}}\bar{\bF}_e-\frac{1}{\bar{n}_e}\nabla\cdot\bar{\bbP}_e-\bar{\bu}_e\cdot\nabla\bar{\bu}_e\\
\,&-\frac{1}{m_r}\bar{\bE}-\frac{1}{m_r}\bar{\bu}_e\times \bar{\bB}\nonumber\\
\partial_t\bar{\bu}_i&=T\nu\bar{\bF}_i-\frac{1}{\bar{n}_i}\nabla\cdot\bar{\bbP}_i-\bar{\bu}_i\cdot\nabla\bar{\bu}_i\\
\,&+\bar{\bE}+\bar{\bu}_i\times \bar{\bB}\nonumber\\
\partial_t\bar{\rho}_e&=T\nu\sum_{\bar{\sigma}}\bar{C}_{e\bar{\sigma}}+T\nu\bar{\bF}_i\cdot\partial_{\bw}\bar{\rho}_{e}\\
\,&-\frac{1}{m_r}(\bw\times \bar{\bB})\cdot\partial_{\bw}\bar{\rho}_e-\frac{1}{\sqrt{m_r}\bar{n}_e}\nabla\cdot\bar{\bbP}_e\cdot\partial_{\bw}\bar{\rho}_e\nonumber\\
\,&+\nabla\cdot\left(\bar{n}_e\bar{\rho}_e\bar{\bu}_e+\frac{ \bar{n}_e\bar{\rho}_e}{\sqrt{m_r}}\bw\right)-\partial_{\bw}\cdot(\bw\cdot\nabla\bar{n}_e\bar{\rho}_e\bar{\bu}_e)\nonumber\\
\partial_t\bar{\rho}_i&=T\nu\sum_{\bar{\sigma}}\bar{C}_{i\bar{\sigma}}+T\nu\bar{\bF}_e\cdot\partial_{\bw}\bar{\rho}_{i}\label{rhoievol}\\
\,&-(\bw\times \bar{\bB})\cdot\partial_{\bw}\bar{\rho}_i-\frac{1}{\bar{n}_i}\nabla\cdot\bar{\bbP}_i\cdot\partial_{\bw}\bar{\rho}_i\nonumber\\
\,&+\nabla\cdot\left(\bar{n}_i\bar{\rho}_i\bar{\bu}_i+ \bar{n}_i\bar{\rho}_i\bw\right)-\partial_{\bw}\cdot(\bw\cdot\nabla\bar{n}_i\bar{\rho}_i\bar{\bu}_i)\nonumber
\end{align}
Where $\bar{\bF}_\sigma=\frac{1}{\bar{n}_\sigma}\int\bw\sum_{\bar{\sigma}}\bar{C}_{\sigma,\bar{\sigma}}\,d^3\bw$ and $\nu$ is the collision frequency. In the following we'll drop the overbars notation in the equations.  

At this point we are ready to begin finding the slow manifold approximation $\bE^*(\bB,n_i,\bu_i,\rho_i,\rho_e)$ from the above infinite-dimensional dynamical system. The scaled Amp\'ere's law is where we will need to start. If we allow $\epsilon\to0$ we can easily see that there is no natural way to solve for $\bE^*$ in terms of the other dynamical variables. Hence this system is a degenerate fast-slow system. To prove this we'll put the VM system into a form amenable to definition 4. If $\bz=(\bE,\bB,n_i,\bu_i,\bu_e,\rho_i,\rho_e)$ where $n_e=n_i-\epsilon\nabla\cdot\bE$ then
\[
\epsilon\dot{\bz}=\bU_\epsilon(\bz)
\]
where $\bU_\epsilon(\bz)$ is given by the right hand side of \eqref{Ampere} and the right hand sides of \eqref{Faraday}-\eqref{rhoievol} multiplied by $\epsilon$. Hence $\bU_0(\bz)$ as a map between Banach spaces is defined as
\[
(\bE,\bB,n_i,\bu_i,\bu_e,\rho_i,\rho_e)\mapsto(\nabla\times\bB+n_i\bu_e-n_i\bu_i,\bm{0})
\]
To satisfy the definition of fast-slow dynamical systems we need to show that the Fr\'echet derivative $D\bU_0(\bz)$ has nontrivial overlap between its image and kernel when evaluated at a $\bz\neq\bm{0}$ such that $\bU_0(\bz)=\bm{0}$. Indeed, such a $\bz$ is given by 
\[
\bz=\left(\bE,\bB,n_i,\bu_i,\bu_i-\frac{1}{n_i}\nabla\times\bB,\rho_i,\rho_e\right)
\]
The Fr\'echet derivative of $\bU_0(\bz)$ at the above $\bz$ value acting on $\delta\bz$ is
\[
D\bU_0(\bz)[\delta\bz]=\left(\left(1-\frac{1}{n_i}\right)\nabla\times\delta\bB+\frac{\delta n_i}{n_i^2}\nabla\times\bB,\bm{0}\right)
\]
One can now easily see that if $\delta\bz=(\delta\bE,\bm{0})$ for any value of $\delta\bE$ then $D\bU_0(\bz)[\delta\bz]=\bm{0}$ and hence such a  $\delta\bz\in\text{Ker }D\bU_0(\bz)$. Interestingly, by the above formula for the Fr\'echet derivative, any $\delta\bz$ acted on by $D\bU_0(\bz)$ is sent to an element of the form $(\delta\bE,\bm{0})$ and hence $\text{Im } D\bU_0(\bz)\subset\text{Ker }D\bU_0(\bz)$. Therefore the definition of fast-slow system is satisfied.   
In particular, it is of differential index two. Hence upon taking a second time derivative of Amp\'ere's law we can recover a way to solve for the zeroth order term in the slow manifold approximation. This second time derivative is
\begin{align}
\epsilon\partial_{tt}\bE&=\frac{n_e}{\sqrt{m_r}}T\nu\bF_e-\nabla\cdot\bbP_e-\bu_e\nabla\cdot(n_e\bu_e)\label{Att}\\
\,&-\frac{n_e}{m_r}\bE-\frac{n_e}{m_r}\bu_e\times\bB\nonumber\\
\,&-n_iT\nu\bF_i+\nabla\cdot\bbP_i+\bu_i\nabla\cdot(n_i\bu_i)\nonumber\\
\,&-n_i\bE-n_i\bu_i\times\bB\nonumber\\
\,&-\nabla\times(\nabla\times\bE)\nonumber
\end{align}
However, since we are working in the infinite dimensional setting we have that our dynamical quantities are functionals. Hence we need to understand the time derivative via chain rule with Fr\'echet derivatives. Now we can replace and expand $\bE=\bE^*=\bE^*_0+\epsilon\bE^*_1+\dots$ and take $\epsilon\to0$ to obtain 
\begin{align}
n_e&=n_i
\end{align}
\begin{align}
\bu_e&=\bu_i-\frac{1}{n_i}\nabla\times\bB
\end{align}
One then uses these in the second derivative of Amp\'ere's law to solve for the 0th order term in $\bE^*$. Doing so gives
\begin{align}
\bE^*_0&=\calD^{-1}[\frac{n_i}{\sqrt{m_r}}T\nu\bF_e-\nabla\cdot\bbP_e+2(\bu_i\cdot\nabla)(\nabla\times\bB)\\
\,&-n_iT\nu\bF_i+\nabla\cdot\bbP_i-n_iM_r\bu_i\times\bB\nonumber\\
\,&-((\nabla\times\bB)\cdot\nabla)\left(\frac{1}{n_i}\nabla\times\bB\right)+\left(\frac{1}{m_r}\nabla\times\bB\right)\times\bB]\nonumber
\end{align}
Where $\calD:=\nabla\times(\nabla\times)+n_iM_r$ and $M_r=1+\frac{1}{m_r}$. Furthermore, let $\bS_0^*$ be such that $\bE_0^*=\calD^{-1}\bS_0^*$. In order to compute $\bE^*_1$ we will need to return to (\ref{Att}) and replace $n_e$ in terms of $n_i$ from Gauss' law and $\bu_e$ in terms of $\bu_i$ from Amp\'ere's law. In solving for $\bu_e$ we arrive at the equation
\begin{align}
\bu_e&=\frac{1}{n_i-\epsilon\nabla\cdot\bE^*}\left(\epsilon\partial_t\bE^*-\nabla\times\bB+n_i\bu_i\right)   
\end{align}
To handle the division by $\epsilon$ we will use a geometric series approximation so that
\begin{align}
\bu_e&=\frac{1}{n_i}\left(1+\epsilon\frac{\nabla\cdot\bE^*}{n_i}+O(\epsilon^2)\right)\left(\epsilon\partial_t\bE^*-\nabla\times\bB+n_i\bu_i\right)
\end{align}
Generally this gives the useful equation
\begin{align}
\frac{1}{n_e}&=\frac{1}{n_i}\left(1+\epsilon\frac{\nabla\cdot\bE^*}{n_i}+O(\epsilon^2)\right)
\end{align}
which is of use below. We note here that the order $\epsilon$ terms when expanding $\bE^*$ in the above expression for $\bu_e$ are given by
\begin{align}
\frac{d}{d\epsilon}\big|_{\epsilon=0}\bu_e&=\frac{1}{n_i}\partial_t\bE^*_0+(n_i\bu_i-\nabla\times\bB)\frac{\nabla\cdot\bE^*_0}{n_i^2}
\end{align}
Furthermore, we'll assume from here on that the collision frequency of the plasma is slow, in particular that $T\nu=\epsilon^2$. Now expanding $\bE^*$ in terms of $\epsilon$ in (\ref{Att}) and putting $n_e,\bu_e$ in terms of $n_i,\bu_i$ and collecting the first order $\epsilon$ terms we have
\begin{align}
\partial_{tt}\bE_0^*&=\nabla\cdot\left(\nabla\cdot\bE^*_0\frac{\bbP_e}{n_i}\right)-\calD\bE^*_1\nonumber\\
\,&-\frac{1}{n_i}(n_i\bu_i-\nabla\times\bB)\nabla\cdot\partial_t\bE^*_0\nonumber\\
\,&-\nabla\cdot(n_i\bu_i)\left(\frac{1}{n_i}\partial_t\bE^*_0+(n_i\bu_i-\nabla\times\bB)\frac{\nabla\cdot\bE^*_0}{n_i^2}\right)\nonumber\\
\,&-\frac{\bE^*_0}{m_r}\nabla\cdot\bE^*_0-\frac{1}{m_r}(\partial_t\bE^*_0)\times\bB\nonumber
\end{align}
Thus we have that 
\begin{align}
\bE^*_1&=\calD^{-1}[\nabla\cdot\left(\nabla\cdot\bE^*_0\frac{\bbP_e}{n_i}\right)-\partial^2_{tt}\bE^*_0\\
\,&-\frac{1}{n_i}(n_i\bu_i-\nabla\times\bB)\nabla\cdot\partial_t\bE^*_0\nonumber\\
\,&-\nabla\cdot(n_i\bu_i)\left(\frac{1}{n_i}\partial_t\bE^*_0+(n_u\bu_i-\nabla\times\bB)\frac{\nabla\cdot\bE^*_0}{n_i^2}\right)\nonumber\\
\,&-\frac{\bE^*_0}{m_r}\nabla\cdot\bE^*_0-\frac{1}{m_r}(\partial_t\bE^*_0)\times\bB]\nonumber    
\end{align}
Hence we need only calculate derivatives of $\bE^*_0$. However, for the time derivatives, thanks to the chain rule, we need to carefully compute Fr\'echet derivatives. Indeed
\begin{align}
\partial^2_{tt}\bE^*_0&=D^2\bE_0^*(\bZ_0)[\dot{\bZ}_0,\dot{\bZ}_0]+D\bE^*_0(\bZ_0)[D_{\bZ}\dot{\bZ}_0(\bE^*_0,\bZ_0)[\dot{\bZ}_0]\\
\,&+D_{\bE}\dot{\bZ}_0(\bE^*_0,\bZ_0)[D\bE^*_0(\bZ_0)[\dot{\bZ_0}]]]\nonumber
\end{align}
with $\bZ=(\bB,n_i,\bu_i,\rho_i,\rho_e)$. We'll start with the last term on the right hand side which first requires computing $D\bE^*_0(\bZ_0)[\bZ_0]$. 
\begin{align}
D\bE^*_0(\bZ_0)[\dot{\bZ}_0]&=D(\calD^{-1})(\bZ_0)[\dot{\bZ}_0]\bS^*_0+\calD^{-1}D\bS^*_0(\bZ_0)[\dot{\bZ}_0]
\end{align}
Now to compute $D\bS^*_0(\bZ_0)[\dot{\bZ}_0]$ one will need the following quantities (note the abuse of notation for the Fr\'echet derivatives)
\begin{align}
D(\nabla\cdot\bbP_\sigma)&=\nabla\cdot\int(\rho_\sigma\dot{n}_i+n_i\dot{\rho}_\sigma)\bw\bw\,d^3\bw\\
D(2(\bu_i\cdot\nabla)(\nabla\times\bB))&=2(\dot{\bu}_i\cdot\nabla)(\nabla\times\bB)-2(\bu_i\cdot\nabla)(\nabla\times(\nabla\times \bE^*_0))\\
D\left((\nabla\times\bB\cdot\nabla)\left(\frac{1}{n_i}\nabla\times\bB\right)\right)&=-(\nabla\times(\nabla\times\bE^*_0)\cdot\nabla)\left(\frac{1}{n_i}\nabla\times\bB\right)\nonumber\\
\,&-(\nabla\times\bB\cdot\nabla)\left(\frac{\dot{n}_i}{n^2_i}\nabla\times\bB\right)\nonumber\\
\,&-(\nabla\times\bB\cdot\nabla)\left(\frac{1}{n_i}\nabla\times(\nabla\times\bE^*_0)\right)\\
D(M_rn_i\bu_i\times\bB)&=M_r\dot{n}_i\bu_i\times\bB+M_r\dot{\bu}_i\times\bB-M_rn_i\bu_i\times(\nabla\times\bE^*_0)\\
D\left((\frac{1}{m_r}(\nabla\times\bB)\times\bB\right)&=-\frac{1}{m_r}(\nabla\times(\nabla\times\bE^*_0))\times\bB-\frac{1}{m_r}(\nabla\times\bB)\times(\nabla\times\bE^*_0)\label{FrEQ}
\end{align}
Furthermore we also have that
\begin{align}
D(\calD^{-1})(\bZ)[\dot{\bZ}_0]&=-\calD^{-1}(M_r\dot{n}_iId)\calD^{-1}
\end{align}
At this point the full expression for the needed Fr\'echet derivatives are much too unwieldy to write out all at once. Hence we will only list expressions that are needed to complete the term $\partial^2_{tt}\bE^*_0$. In order to compute
\[
D_{\bE}\dot{\bZ}_0(\bE^*_0,\bZ_0)[D\bE^*_0(\bZ_0)[\dot{\bZ_0}]]
\]
we differentiate the evolution equations for $\bZ$ with respect to the electric field and evaluate at $D\bE^*_0(\bZ_0)[\dot{\bZ_0}]$ for the electric field value. This results in 
\begin{align}
D_{\bE}\dot{\bZ}_0(\bE^*_0,\bZ_0)[D\bE^*_0(\bZ_0)[\dot{\bZ_0}]]&=(-\nabla\times (D\bE^*_0(\bZ_0)[\dot{\bZ_0}]),\bm{0},D\bE^*_0(\bZ_0)[\dot{\bZ_0}],\bm{0},\bm{0})
\end{align}
Now we will find $D_{\bZ}\dot{\bZ}_0(\bE^*_0,\bZ_0)[\dot{\bZ}_0]$. Indeed the RHS of the evolution equations \eqref{Faraday}-\eqref{rhoievol}, evaluated at $(\bE^*_0,\bZ_0)$, need to be each differentiated with respect to time. Thus notating the components of $D_{\bZ}\dot{\bZ}_0(\bE^*_0,\bZ_0)[\dot{\bZ}_0]$ as $\bD_i$ we have
\begin{align}
\bD_1&=\bm{0}\\
\bD_2&=-\nabla\cdot(\dot{n}_i\bu_i+n_i\dot{\bu}_i)\\
\bD_3&=\frac{\dot{n}_i}{n_i^2}\nabla\cdot\dot{\bbP}-\dot{\bu}_i\cdot\nabla\bu_i-\bu_i\cdot\nabla\dot{\bu}_i+\bE^*_0+\dot{\bu}_i\times\bB+\bu_i\times\dot{\bB}\\
\bD_4&=(\bw\times\dot{\bB})\cdot\partial_{\bw}\rho_i+(\bw\times{\bB})\cdot\partial_{\bw}\dot{\rho}_i+\frac{\dot{n_i}}{n_i^2}\nabla\cdot\bbP_i\cdot\partial_{\bw}\rho_i\\
\,&-\frac{1}{n_i}\nabla\cdot\dot{\bbP}_i\cdot\partial_{\bw}\rho_i-\frac{1}{n_i}\nabla\cdot{\bbP}_i\cdot\partial_{\bw}\dot{\rho_i}\\
\,&+\nabla\cdot\left((n_i\rho_i\bu_i)_t+(n_i\rho_i)_t\bw\right)-\partial_{\bw}\cdot(\bw\cdot\nabla n_i\rho_i\bu_i)_t\\
\bD_5&=\frac{1}{m_r}(\bw\times\dot{\bB})\cdot\partial_{\bw}\rho_e+\frac{1}{m_r}(\bw\times{\bB})\cdot\partial_{\bw}\dot{\rho}_e+\frac{\dot{n_i}}{\sqrt{m_r}n_i^2}\nabla\cdot\bbP_e\cdot\partial_{\bw}\rho_e\\
\,&-\frac{1}{\sqrt{m_r}n_i}\nabla\cdot\dot{\bbP}_e\cdot\partial_{\bw}\rho_e-\frac{1}{\sqrt{m_r}n_i}\nabla\cdot{\bbP}_e\cdot\partial_{\bw}\dot{\rho_e}\\
\,&+\nabla\cdot\left((n_i\rho_e\bu_e)_t+\frac{1}{\sqrt{m_r}}(n_i\rho_e)_t\bw\right)-\partial_{\bw}\cdot(\bw\cdot\nabla n_i\rho_e\bu_e)_t
\end{align} 
Thus all the pieces are in place to compute the whole second term of the expression for $\partial^2_{tt}\bE^*_0$. Finally we need to compute the second order Fr\'echet derivative. This gives
\[
D^2\bE^*_0(\bZ_0)[\dot{\bZ}_0,\dot{\bZ}_0]=D^2(\calD^{-1})\bS^*_0+D(\calD^{-1})D\bS^*_0+D(\calD^{-1})D\bS^*_0+\calD^{-1}D^2\bS^*_0
\]
Once again this expression is particularly cumbersome and the computation will be left at this stage. It is now straightforward, albeit tedious, to completely find an expression for $\partial^2_{tt}\bE^*_0$ and hence for $\bE^*_1$. 

The size of the expression for the first order correction to the Quasi-Neutral limit is of the VM equations is rather large and requires inverting elliptic operators. Numerically however the first order term should tractable. We note also that the collision operators play no role in the first order approximation due to our assumption of slow collision frequency. In order to see the collision operators play a role one would have to implement a second order correction. Such calculations are likely unreasonable to be done by hand and it is suggested any attempts at such should implement a computer algebra program. 
%\section{Discussion}
%relation to analytic approach?

\section{Acknowledgement}
Research presented in this Review was supported by the Los Alamos National Laboratory
LDRD program under project number 20180756PRD4. This research was supported in part by an appointment with the National Science Foundation (NSF) Mathematical Sciences Graduate Internship (MSGI) Program sponsored by the NSF Division of Mathematical Sciences. This program is administered by the Oak Ridge Institute for Science and Education (ORISE) through an interagency agreement between the U. S. Department of Energy (DOE) and NSF. ORISE is managed for DOE by ORAU. All opinions expressed in this paper are the author's and do not necessarily reflect the policies and views of NSF, ORAU/ORISE, or DOE.

%As a final comment, it is worth mentioning that the governing equations admit a Kelvin--Noether circulation theorem. For any closed loop $\overline{\gamma}(\meanv)$ moving along the Lagrangian-mean velocity $\meanv$, then
%%
%\begin{equation}
%	\frac{\mathrm{d}}{\mathrm{d}t} 
%		\oint_{\overline{\gamma}(\meanv) } \meanu \cdot \mathrm{d} \vec{x} =0.
%\end{equation}
%%
%Note that, within the current GLM framework, the closed loop is advected by the Lagrangian-averaged velocity $\meanv$ while the velocity appearing in the circulation integral is the Eulerian-averaged velocity $\meanu$.

%\bibliography{/Users/josh/Dropbox/Apps/Texpad/latex/cumulative_bib_file.bib}
\bibliography{cumulative_bib_file.bib}
%%%%%%%%%%%%%%%%%%%%%%%%%%%%%%%%%%%%%

%% put content of bbl file here when ready to submit

\providecommand{\noopsort}[1]{}\providecommand{\singleletter}[1]{#1}%
%

%%%%%%%%%%%%%%%%%%%%%%%%%%%%%%%%%%%%

\end{document}